\documentclass[12pt]{amsart}

\usepackage{amsmath,amsthm,amsfonts,amssymb,latexsym,verbatim,bbm,hyperref}
\usepackage[shortlabels]{enumitem}
\usepackage[top=1.5in,bottom=1.5in,left=1.5in,right=1.5in]{geometry}
\usepackage{subcaption}
\usepackage{afterpage}

\DeclareRobustCommand{\SkipTocEntry}[5]{}

\allowdisplaybreaks

\setlist{itemsep=.5\baselineskip,topsep=.5\baselineskip}

\numberwithin{equation}{section}
\theoremstyle{plain}
\newtheorem{theorem}{Theorem}[section]
\newtheorem{thm}[theorem]{Theorem}
\newtheorem{lemma}[theorem]{Lemma}

\newtheorem{prop}[theorem]{Proposition}
\newtheorem{example}[theorem]{Example}

\newtheorem{defn}[theorem]{Definition}

\newtheorem{cor}[theorem]{Corollary}

\newtheorem{rmk}[theorem]{Remark}

\newcommand{\iso}{\cong}
\newcommand{\arr}{\rightarrow}
\newcommand{\incl}{\hookrightarrow}

\newcommand{\R}{\mathbb{R}}

\newcommand{\Z}{\mathbb{Z}}

\newcommand{\eps}{\epsilon}
\newcommand{\vareps}{\varepsilon}
\newcommand{\Id}{\mathbbm{1}}

\newcommand{\mcA}{\mathcal{A}}
\newcommand{\mcB}{\mathcal{B}}
\newcommand{\mcC}{\mathcal{C}}
\newcommand{\mcD}{\mathcal{D}}

\newcommand{\mcF}{\mathcal{F}}
\newcommand{\mcG}{\mathcal{G}}
\newcommand{\mcH}{\mathcal{H}}
\newcommand{\mcI}{\mathcal{I}}
\newcommand{\mcJ}{\mathcal{J}}

\newcommand{\mcN}{\mathcal{N}}

\newcommand{\mcP}{\mathcal{P}}

\newcommand{\mcU}{\mathcal{U}}

\newcommand{\mcW}{\mathcal{W}}

\newcommand{\mfP}{\mathfrak{P}}

\usepackage{tikz}
\usetikzlibrary{positioning,calc}
\usetikzlibrary{decorations.pathmorphing,decorations.pathreplacing,patterns,shapes.symbols}
\definecolor{fillpaleyellow}{RGB}{255,250,205} 
\definecolor{fillbrownyellow}{RGB}{209,189,148} 
\definecolor{fillyellow}{RGB}{217,191,140}
\definecolor{fillgray}{RGB}{199,204,209}

\DeclareMathOperator{\mult}{mult}
\DeclareMathOperator{\bd}{bd}
\DeclareMathOperator{\sign}{sign}
\DeclareMathOperator{\ch}{ch}
\DeclareMathOperator{\Inv}{Inv}
\DeclareMathOperator{\res}{res}
\DeclareMathOperator{\germ}{germ}
\DeclareMathOperator{\Cycle}{\#Cycle}
\newcommand{\surg}{\Rightarrow}

\title[Tsirelson's problem and an embedding theorem]{Tsirelson's problem and an
embedding theorem for groups arising from non-local games}

\author{William Slofstra}
\address{Institute for Quantum Computing and Department of Pure Mathematics, University of Waterloo, Canada}
\email{weslofst@uwaterloo.ca}
\thanks{Partially supported by NSERC grant number 2018-03968.}

\begin{document}

\begin{abstract}
    Tsirelson's problem asks whether the commuting operator model for two-party
    quantum correlations is equivalent to the tensor-product model. We give a
    negative answer to this question by showing that there are non-local games
    which have perfect commuting-operator strategies, but do not have perfect
    tensor-product strategies. The weak Tsirelson problem, which is known to
    be equivalent to Connes embedding problem, remains open.  

    The examples we construct are instances of (binary) linear system games.
    For such games, previous results state that the existence of perfect
    strategies is controlled by the \emph{solution group} of the linear system.
    Our main result is that every finitely presented group embeds in some
    solution group. As an additional consequence, we show that the problem of
    determining whether a linear system game has a perfect commuting-operator
    strategy is undecidable.
\end{abstract}
\maketitle

\vspace{-1cm}

{\small
\tableofcontents }

\section{Introduction}\label{S:introduction}

In a two-player non-local game, the players, commonly called Alice and Bob, are
physically separated and unable to communicate. They each receive a question
chosen at random from a finite question set, and reply with a response from a
finite answer set. If the joint answers meet a predetermined winning condition
dependent on the joint questions, then Alice and Bob win; otherwise they lose. 
The rules of the game, including the winning condition and distribution on
questions, are completely known to Alice and Bob, and they can arrange in advance
a strategy which will maximize their success probability. However, since they
cannot communicate during the game, they may not be able to play perfectly, i.e.
win with probability one. 

Classically, Alice and Bob's strategy for a non-local game is described by a
local hidden variable model. Bell's famous theorem states that Alice and Bob
can achieve better results than is possible with a local hidden variable model
if they share an entangled quantum state \cite{Be64}. Since Bell's discovery,
non-local games have been heavily studied\footnote{Usually under the name
``Bell tests'' or ``Bell inequalities''. The term ``non-local games'' is
more recent.} in physics, mathematics, and computer science; see
\cite{CHSH69,FC72,Ts80,AGR82,CHTW04,Br08,NPA08,JPPVW10,KV11,KKMTV11,He15} for a
small sample of results. Despite this, a number of foundational questions
remain open, chief among which is Tsirelson's problem: A quantum strategy for a
non-local game can be described as a set of measurement operators on Hilbert
spaces $H_A$ and $H_B$, along with a quantum state in the joint space $H = H_A
\otimes H_B$. We refer to this as the \emph{tensor-product model}. While the
tensor-product model is often the default, there is another choice: a quantum
strategy can be described as a set of measurements and quantum state on a
shared Hilbert space $H$, with the property that Alice's measurement operators
commute with Bob's measurement operators. This \emph{commuting-operator model}
is used, for instance, in algebraic quantum field theory \cite{HK64}. The
observable consequences of a strategy in either model are captured by the
correlation matrix of the strategy. If $H$ is finite-dimensional, every
correlation matrix arising from a commuting-operator strategy can be realized
using a tensor-product strategy. Tsirelson's problem asks whether this is true
for a general Hilbert space. 

This problem has an interesting history. Tsirelson originally stated the
problem in a survey on Bell inequalities \cite{Ts93}, and claimed without proof
that the two models gave rise to the same set of correlation matrices. He later
retracted this claim, and posted the question to a list of open problems in
quantum information theory \cite{Ts06}. Tsirelson and subsequent authors
\cite{NPA08, SW08, JNPPSW11,Fr12,NCPV12,PT15,DP16,PSSTW16} studying Tsirelson's
problem have considered several different sets of tensor-product strategies. 
Paulsen and Todorov \cite{PT15} (see also Dykema and Paulsen \cite{DP16})
observe that these variations lead to a hierarchy of sets of correlation
matrices
\begin{equation}\label{E:hierarchy}
    C_q \subseteq C_{qs} \subseteq C_{qa} \subseteq C_{qc},
\end{equation}
where $C_q$ is the set of correlations arising from tensor-product strategies
on finite-dimensional Hilbert spaces, $C_{qs}$ is the set of correlations
arising from tensor-product strategies (on possibly infinite-dimensional
Hilbert spaces) with a vector state, $C_{qa} = \overline{C_{qs}}$ is the set of
correlations which are limits of correlations in $C_{qs}$, and $C_{qc}$ is the set
of correlations arising from commuting-operator strategies. If we restrict to
non-local games with question sets of size $n$ and answer sets of size $m$,
then all these sets are convex subsets of $\R^{m^2 n^2}$, and prior to this
paper, none of the inclusions (including $C_q \subseteq C_{qc}$) were known to
be strict. Thus for each $t \in \{q,qs,qa\}$ there is a Tsirelson problem asking
whether $C_{qc} = C_{t}$.  Ozawa, building on the work of Junge, Navascu\'es,
Palazuelos, P\'erez-Garc\'ia, Sholz, and Werner \cite{JNPPSW11} and work of
Fritz \cite{Fr12}, has shown that $C_{qc} = C_{qa}$ if and only if Connes'
embedding conjecture is true \cite{Oz13}. At the other end of the hierarchy, if
$C_q$ was equal to $C_{qc}$ then every correlation matrix, whether commuting
operator or tensor-product, would arise from a finite-dimensional Hilbert
space. The ``middle'' version, which asks whether $C_{qc} = C_{qs}$, seems
closest to Tsirelson's original claim. 

The first main result of this paper is that there is a non-local game which can
be played perfectly with a commuting-operator strategy, but which cannot be
played perfectly using a tensor-product strategy with a vector state. Thus we
resolve the middle version of Tsirelson's problem by showing that $C_{qc}
\neq C_{qs}$. This also shows that $C_{qc} \neq C_q$. As a result, this game
is interesting from the perspective of quantum information and computation,
where a non-local game is often regarded as a computational scenario in which
better results can be achieved with entanglement as a resource. From this point
of view, it is natural to ask how much entanglement is needed to play a game
optimally, and in particular, whether every game can be played optimally on a
finite-dimensional Hilbert space.  The game we construct shows that this is not
possible, at least if we allow commuting-operator strategies. Previously-known
examples of this type have involved either quantum questions \cite{LTW13,RV15}
or infinite answer sets \cite{MV14}. 

One reason we would have desired that every game have an optimal strategy
on a finite-dimensional Hilbert space is that it would make it possible to
determine the optimal winning probability of a non-local game over entangled
strategies. At present the only known methods for this task, aside from
brute-force search over strategies, are variants of the Navascu\'es-Pironio-Ac\'in (NPA)
hierarchy \cite{NPA08,DLTW12}. Given a non-local game, the NPA hierarchy provides a
sequence of upper bounds which converge to the optimal winning probability in
the commuting-operator model. However, the hierarchy does not provide, outside
of special cases, a stopping criterion, i.e. a way to tell if the value will
fall below a given threshold. Our second main result is that it is undecidable
to determine if a non-local game can be played perfectly with a
commuting-operator strategy. In particular, this implies that there is no
stopping criterion for the NPA hierarchy which applies to all games.  

The games we consider are binary linear system games, so named because they
arise from linear systems over $\Z_2$. Such games have been studied previously
in \cite{CM14,Ar12,Ji14}. Cleve and Mittal implicitly associate a certain group
to every linear system over $\Z_2$, such that perfect tensor-product strategies
for the game correspond to certain finite-dimensional representations of the
group \cite{CM14}. We call this group, which is analogous to the solution space
of a linear system, the \emph{solution group}. In \cite{CLS16} it is shown that
perfect commuting-operator strategies for a binary linear system game
correspond to certain possibly-infinite-dimensional representations of the
solution group.  Solution groups form an interesting class of groups. They are
finitely presented, but their presentations must satisfy a property which in
\cite{CLS16} is called \emph{local compatibility}: if 
\begin{equation*}
    x_1 \cdots x_n = 1
\end{equation*}
is a relation, where $x_i$, $1 \leq i \leq n$, are not necessarily distinct
generators of the group, then the presentation must also contain the relations
$x_i x_j = x_j x_i$ for all $1 \leq i,j \leq n$. This condition is natural from
the perspective of quantum mechanics, where two observables commute if and only
if the observables correspond to quantities which can be measured (or known)
simultaneously. Group relations of this exact type can be found in
contextuality theorems of Mermin and Peres \cite{Me90,Pe90,Me93}. Local
compatibility is \textit{a priori} a strong constraint on group presentations.
Our primary result, on which our other two results are based, is that any
finitely presented group can be embedded in a solution group. This embedding
theorem allows us to extend results from combinatorial group theory to solution
groups. In particular, the embedding theorem implies the existence of a
non-residually-finite solution group, and this leads directly to the separation
$C_{qc} \neq C_{qs}$. Using non-residually-finite groups to recognize
infinite-dimensional state spaces was previously proposed in \cite{Fr13}.

There have been a number of related results since this paper was first posted.
The proof of the embedding theorem takes up most of this paper, and it's
natural to ask whether there might be a shorter proof.  In a follow-up paper
\cite{Sl17}, we give a shorter and more direct proof of an embedding theorem
for a restricted class of finitely-presented groups which is still sufficient
for our primary applications. This shorter proof allows us to keep control over
the dimensions of approximate representations in the embedding, leading to the
finer separation $C_{qa} \neq C_{qs}$. It is also natural to ask how large the
question and answer sets need to be to separate $C_{qs}$ from $C_{qa}$ or
$C_{qc}$. The group embedding theorems in this paper and \cite{Sl17} give
answer sets of size $8$, but question sets with hundreds of questions. Dykema,
Paulsen, and Prakash have made dramatic progress on this question, showing that
$C_{qa} \neq C_{qs}$ for correlations with only five questions and two answers
\cite{DPP17}. Other related results can be found in \cite{GHJ17, Ha18}.
Also, Coladangelo and Stark have shown that $C_{qs} \neq C_q$ \cite{CS18},
leaving the question of whether $C_{qa} = C_{qc}$ as the last unknown
separation in Equation \eqref{E:hierarchy}. Finally, 
Tobias Fritz has used the existence of a solution group with an undecidable
word problem to show that quantum logic is undecidable \cite{Fr16}.

The rest of this paper is structured as follows. In the next section, we first
recall the basic definitions of quantum correlation sets and non-local games,
and then define binary linear system games and the solution group of a linear system. 
In Section \ref{S:embedding}, we state the embedding theorem for solution
groups, and prove our two main results as corollaries.  The rest of the paper
is concerned with the proof of the embedding theorem.  The main technical tool
used is pictures of groups. We give an overview of pictures tailored to
solution groups in Sections \ref{S:pictures} and \ref{S:hyperpics}; expert
readers will want to skip or briefly review these sections on first reading. In
dealing with pictures, it is more convenient to use hypergraphs instead of
linear systems, and we introduce hypergraphs into our definitions immediately
in Section \ref{S:bcs}. In Section \ref{S:category} we develop a notion of
morphisms between hypergraphs; the concept is similar to graph minors, but
differs from standard notions of hypergraph minors as in \cite{RS10}. The more
technical aspects of the proof are contained in Sections
\ref{S:cycle}-\ref{S:constellation}; we suggest that the reader skim these
sections on first reading and proceed to Section \ref{S:proof}. For a second
reading, Corollary \ref{C:suncycle} might provide a good initial target. 

\addtocontents{toc}{\SkipTocEntry}
\subsection{Acknowledgements}

I thank Richard Cleve, Jason Crann, Zhengfeng Ji, Li Liu, Andrew Marks, Vern
Paulsen, and Vincent Russo for many helpful discussions.  I thank Richard
Cleve, Li Liu, and the anonymous referees for detailed comments on the
manuscript, and Tobias Fritz and Doug Park for helpful suggestions. Any
remaining errors are my own. 

\section{Quantum correlation sets and linear system games}

\subsection{Quantum correlation sets}

In quantum mechanics, a physical system corresponds to a Hilbert space $H$, and
the state of the system is recorded in a unit vector $v \in H$. A measurement
in a system with finite outcome set $[m] := \{1,\ldots,m\}$ can be represented
by a family $\{P_i\}_{i=1}^m$ of self-adjoint projections $P_i$ such that
$\sum_{i=1}^m P_i = \Id$.  Such a family is called a \emph{projection-valued
measure (PVM)}.  The probability of measuring outcome $i$ with PVM
$\{P_i\}_{i=1}^{m}$ and state $v$ is $v\cdot P_i  v$, where $v \cdot w$ is the
inner product in $H$. 

As described in the introduction, there are two natural ways to represent a
bipartite physical system (i.e. a system composed of two physically-separated
subsystems) in quantum mechanics. In the tensor-product formalism, the
subsystems correspond to Hilbert spaces $H_A$ and $H_B$, and the Hilbert space
of the joint system is $H_A \otimes H_B$. The state of the joint system is a
unit vector $v \in H_A \otimes H_B$. If $\{P_a\}_{a=1}^{m_A}$, is a measurement
on $H_A$ and $\{Q_b\}_{b=1}^{m_B}$ is a measurement on $H_B$, then the joint
measurement is represented by $\{P_{a} \otimes Q_{b}\}_{(a,b) \in [m_A] \times
[m_B]}$, and the probability of measuring $(a,b) \in [m_A] \times [m_B]$ is $v
\cdot P_a \otimes Q_b v$. Let $m_A,m_B,n_A,n_B$ be positive integers.  A
function $p \in \R^{[m_A] \times [m_B] \times [n_A] \times [n_B]}$ is a
\emph{tensor-product correlation} if there are Hilbert spaces $H_A$ and $H_B$,
a unit vector $v \in H_A \otimes H_B$, a PVM $\{P_a^x\}_{a=1}^{m_A}$ on $H_A$
for every $x \in [n_A]$, and a PVM $\{Q_b^y\}_{b=1}^{m_B}$ on $H_B$ for every
$y \in [n_B]$, such that 
\begin{equation*}
    p(a,b|x,y) = v \cdot P_a^x \otimes Q_b^y v
\end{equation*}
for every $(a,b,x,y) \in [m_A] \times [m_B] \times [n_A] \times [n_B]$.  In
other words, $p$ is a tensor-product correlation if there is a bipartite system
and state vector, along with a set of $n_A$ (resp. $n_B$) measurements on the
first (resp. second) subsystem, each with $m_A$ (resp. $m_B$) outcomes, such
that $p(a,b|x,y)$ is the probability of measuring outcome $(a,b)$ on
measurements $(x,y)$. A tensor-product correlation $p$ is
\emph{finite-dimensional} if the Hilbert spaces $H_A$ and $H_B$ can be chosen
to be finite-dimensional. The set of tensor-product correlations for a given
tuple $(m_A,m_B,n_A,n_B)$ is denoted by $C_{qs}(m_A,m_B,n_A,n_B)$, and the set
of finite-dimensional correlations is denoted by $C_{q}(m_A,m_B,n_A,n_B)$. The
closure of $C_{qs}(m_A,m_B,n_A,n_B)$ is denoted by
$C_{qa}(m_A,m_B,n_A,n_B)$.\footnote{The letters qa and qs can be read as
``quantum-approximable'' and ``quantum-spatial''. The latter term refers to 
the spatial tensor product of $C^*$-algebras (see below).} 
By a result of Scholz-Werner, $C_{qa}(m_A,m_B,n_A,n_B)$ is also the closure of
$C_{q}(m_A,m_B,n_A,n_B)$ \cite{SW08}. 

The other natural way to represent a bipartite system is via the
commuting-operator formalism. In this framework, each subsystem is represented
by subalgebras $\mcA$ and $\mcB$ of bounded operators on the Hilbert space $H$ of
the joint system, such that $a b = b a$ for all elements $a \in \mcA$ and $b
\in \mcB$. If $\{P_a\}_{a=1}^{m_A}$ and $\{Q_b\}_{b=1}^{m_B}$ are PVMs in
$\mcA$ and $\mcB$ respectively, then $\{P_a Q_b\}_{(a,b)\in [m_A] \times [m_B]}$ is
also a PVM, and the probability of measuring outcome $(a,b)$ on state $v \in H$
is $v \cdot P_a Q_b v$. Let $m_A,m_B,n_A,n_B$ be positive integers. A function
$p \in \R^{[m_A] \times [m_B] \times [n_A] \times [n_B]}$ is a
\emph{commuting-operator correlation} if there is a Hilbert space $H$, 
a unit vector $v \in H$, a PVM $\{P^x_a\}_{a=1}^{m_A}$ on $H_A$ for every
$x \in [n_A]$, and a PVM $\{Q^y_b\}_{b=1}^{m_B}$ on $H_B$ for every $y \in
[n_B]$, such that $P^x_a Q^y_b = Q^y_b P^x_a$ and
\begin{equation*}
    p(a,b|x,y) = v \cdot P^x_a Q^y_b v
\end{equation*}
for all $(a,b,x,y) \in [m_A] \times [m_B] \times [n_A] \times [n_B]$. The set
of commuting operator correlations for a given tuple $(m_A,m_B,n_A,n_B)$ is
denoted by $C_{qc}(m_A,m_B,n_A,n_B)$. It is not hard to see that this set is
closed.

We could also use positive-operator valued measures (POVMs) instead of PVMs to
model measurements in the above definitions. However, any correlation
achievable with POVMs can also be achieved with PVMs. For tensor-product
correlations this follows immediately from the Naimark dilation theorem, while
for commuting-operator correlations this fact can be found in \cite[Section
3]{Fr12} or \cite{PT15}. We stick with PVMs because of the clearer connection
with operator algebras below.

A \emph{two-player non-local game} with question sets $[n_A]$, $[n_B]$ and
answer sets $[m_A]$,$[m_B]$ is specified by a probability distribution $\pi$ on
$[n_A]\times[n_B]$ and a function $V : [m_A] \times [m_B]\times [n_A] \times
[n_B] \arr \{0,1\}$. In the operational interpretation of the game, a referee chooses
a pair of questions $(x,y) \in [n_A] \times [n_B]$ at random according to $\pi$, giving
$x$ to the first player (Alice), and $y$ to the second player (Bob). The
players respond with answers $a \in [m_A]$ and $b \in [m_B]$ respectively. The
players win the game if $V(a,b|x,y)=1$, and lose if $V(a,b|x,y)=0$. The
players' strategy in such a game can be described by the function $p \in
\R^{[m_A] \times [m_B] \times [n_A] \times [n_B]}$, where $p(a,b|x,y)$ is the
probability that the players output $(a,b)$ on inputs $(x,y)$. The winning
probability on a strategy $p$ is given by
\begin{equation*}
    \sum_{x=1}^{n_A} \sum_{y=1}^{n_B} \pi(x,y) \sum_{a=1}^{m_A} \sum_{b=1}^{m_B} V(a,b|x,y)
        p(a,b|x,y).
\end{equation*}
The game has a \emph{perfect tensor-product strategy} (resp. \emph{perfect
commuting-operator strategy}) if there is a strategy in $C_{qs} (m_A, m_B, n_A, n_B)$ (resp.
$C_{qc} (m_A,\allowbreak m_B,\allowbreak n_A,\allowbreak n_B)$) for which the winning probability is $1$. We will show in Section
\ref{S:embedding} that there is a game with a perfect commuting-operator
strategy, but no perfect tensor-product strategy, demonstrating that
$C_{qs}(m_A,m_B,n_A,n_B) \neq C_{qc}(m_A,m_B,n_A,n_B)$ for some
$(m_A,m_B,n_A,n_B)$. 

For simplicity, many authors work only with the correlation sets $C_t(m,n) :=
C_t(m,m,n,n)$, $t \in \{q,qs,qa,qc\}$. We work with games which typically have
$m_A \neq m_B$ and $n_A \neq n_B$. However, by padding the game with irrelevant
questions and answers, we can assume that $m_A = m_B$ and $n_A = n_B$. Indeed,
suppose that $\pi, V$ is the data of game with question sets $[n_A]$,$[n_B]$
and answer sets $[m_A]$,$[m_B]$. Let $m=\max(m_A,m_B)$ and $n=\max(n_A,n_B)$,
extend $\pi$ to a probability distribution $\widetilde{\pi}$ on $[n] \times
[n]$ by setting $\widetilde{\pi}(x,y) = 0$ if $x > n_A$ or $y > n_B$, and
similarly extend $V$ to a function $\widetilde{V} : [m]^2 \times [n]^2 \arr
\{0,1\}$ by setting $\widetilde{V}(a,b|x,y) = 0$ if $(a,b,x,y) \not\in [m_A]
\times [m_B] \times [n_A] \times [n_B]$. Then it is easy to see that the game
determined by $\widetilde{\pi}$, $\widetilde{V}$ has a perfect strategy in
$C_t(m,n)$ if and only if the game determined by $\pi$,$V$ has a perfect
strategy in $C_t(m_A,m_B,n_A,n_B)$. Hence our argument will imply that
$C_{qs}(m,n) \neq C_{qc}(m,n)$ for some $m,n$. 

We also note that it is not necessary for the questions and answers in a
non-local game to be integers: the question and answer sets $[n_A]$, $[n_B]$,
$[m_A]$, $[m_B]$ can be replaced with any finite sets of the same size. 

To finish the section, we note that the different correlation sets can also
be defined using different tensor products of $C^*$-algebras. A PVM
$\{P_a\}_{a=1}^m$ on a Hilbert space $H$ can be thought of as the spectral
projections of the unitary matrix $U = \sum_{a=1}^{m} e^{2\pi(a-1)/m} P_a$, and
in this way PVMs with $m$ outcomes correspond to representations of the
$C^*$-group algebra $C^* \Z_m$. Similarly, collections of $n$ PVMs, each with
$m$ outcomes, correspond to representations of $C^* \Z_m^{*n}$, where
$\Z_m^{*n}$ is the free product of $n$ copies of $\Z_m$. Let $p^x_a$, $a \in
[m_A]$, $x \in [n_A]$ be the $a$th spectral projection in the $x$th factor of
$\Z_m$ in $\mcA = C^* \Z_{m_A}^{*{n_A}}$, and similarly let $q^y_b$, $b \in
[m_B]$, $y \in [n_B]$ be the $b$th spectral projection in the $y$th factor of
$\Z_m$ in $\mcB = C^* \Z_{m_B}^{*n_B}$. Then $C_{qc}(m_A,m_B,n_A,n_B)$ consists
of the functions $p \in \R^{[m_A] \times [m_B] \times [n_A] \times [n_B]}$ of
the form $p(a,b|x,y) = f(p^x_a q^y_b)$ for some state
$f$ on the maximal tensor product $\mcA \otimes_{max} \mcB$, and
$C_{qa}(m_A,m_B,n_A,n_B)$ consists of the functions $p \in \R^{[m_A] \times
[m_B] \times [n_A] \times [n_B]}$ of the form $p(a,b|x,y) = f(p^x_a q^y_b)$ for
some state $f$ on the minimal (also known as the spatial) tensor product $\mcA
\otimes_{min} \mcB$ \cite{SW08,Fr12,JNPPSW11}.  

A well-known result of Kirchberg states that Connes embedding problem is
equivalent to what is known as Kirchberg's conjecture, that $C^* \mcF_n
\otimes_{max} C^* \mcF_n = C^* \mcF_n \otimes_{min} C^* \mcF_n$ for all (or
some) $n \geq 2$, where $\mcF_n$ is the free group on $n$ generators
\cite{Ki93}. In Kirchberg's conjecture, the group $\mcF_n$ can be replaced with
the group $\Z_m^{*n}$ for $m,n \geq 2$, $(m,n) \neq (2,2)$ (see for instance
the surveys in \cite{Fr12,Oz13}). Hence an affirmative answer to the the Connes
embedding problem would imply that $C_{qc}(m,n) = C_{qa}(m,n)$ for all $(m,n)$.
Ozawa has shown that the converse is also true, so $C_{qc}(m,n) = C_{qa}(m,n)$
for all $(m,n)$ if and only if Connes embedding problem is true.

The set $C_{qs}(m_A,m_B,n_A,n_B)$ consists of all functions $p \in \R^{[m_A]
\times [m_B] \times [n_A] \times [n_B]}$ of the form $p(a,b|x,y) = v \cdot
\phi_A(p^x_a) \otimes \phi_B(q^y_b) v$ for any pair of representations $\phi_A$
and $\phi_B$ of $\mcA$ and $\mcB$ on Hilbert spaces $H_A$ and $H_B$
respectively, and unit vector $v \in H_A \otimes H_B$. A state on a tensor-product
$\mcA \otimes \mcB$ of the form $T \mapsto v \cdot (\phi_A \otimes
\phi_B)(T) v$ is called a spatial state; for the minimal tensor product, these
states are dense in the set of all states. Hence separating
$C_{qc}(m_A,m_B,n_A,n_B)$ from $C_{qs}(m_A,m_B,n_A,n_B)$ shows in particular
that there are states on $\mcA \otimes_{max} \mcB$ which do not come from
spatial states. 

\subsection{Linear system games, hypergraphs, and solution groups}\label{S:bcs}

Binary linear system games are based on linear systems $Ax=b$ over $\Z_2$.  It
is convenient to think of linear systems in terms of hypergraphs.  By a
\emph{hypergraph}, we mean a triple $\mcH = (V,E,A)$, where $V = V(\mcH)$ and
$E = E(\mcH)$ are finite sets of \emph{vertices} and \emph{edges} respectively,
and $A \in \Z_{\geq 0}^{V \times E}$ is the \emph{incidence matrix} between $V$
and $E$, so $A_{ve} \geq 0$ is the degree of incidence between edge $e$ and
vertex $v$. We say that $v$ and $e$ are \emph{incident} if $A_{ve} > 0$.  If $v
\in V$, then the degree of $v$ is $|v| = \sum_e A_{ve}$.  Similarly if $e \in
E$ then $|e| = \sum_v A_{ve}$. We say that $\mcH$ is \emph{simple} if $A_{ve}
\leq 1$ for all $v \in V$ and $e \in E$, \emph{$k$-regular} if $|v| = k$ for
all $v \in V$, and a \emph{graph} if $|e| = 2$ for all $e \in E$. 

Note that this definition of hypergraphs allows both isolated vertices and
isolated edges, i.e. vertices (resp. edges) which are incident to no edges
(resp.  vertices). A \emph{$\Z_2$-vertex labelling} of $\mcH$ is a function $b
: V \arr \Z_2$. With these conventions, there is a correspondence between
linear systems $Ax=b$ and simple hypergraphs $\mcH$ with a vertex labelling
$b$.  From this point of view, the edges of a hypergraph correspond to the
variables of a linear system, and the vertices correspond to constraints.
Similarly, pairs $(\mcH,b)$ where $\mcH$ is not necessarily simple correspond
to linear systems $Ax =b$ over $\Z_2$ with a choice of non-negative integer
representatives for the coefficients $A_{ve}$. 

To any $m \times n$ linear system $Ax=b$, we can associate a linear system
non-local game $\mcG$, and a group $\Gamma$ \cite{CM14,CLS16}. In the game
$\mcG$, Alice receives the index $1 \leq x \leq m$ of a row of $A$, and returns
a function $a \in \Z_2^{V_x}$, where $V_x := \{1 \leq j \leq n : A_{ij} \neq
0\}$. Bob receives the index $1 \leq y \leq n$ of a column of $A$, and returns
$b \in \Z_2$. The players win if $\sum_{j \in V_x} a(j) = b_x$, and either $y
\not\in V_x$, or $y \in V_x$ and $a(y) = b$. In other words, Alice outputs
values for every variable in the $x$th equation, Bob outputs a value for the
$y$th variable, and the players win if Alice's answer satisfies the $x$th
equation and their answers are consistent with each other.
The group $\Gamma$ is the focus of attention of this paper, and is defined as
follows.
\begin{defn}\label{D:solutiongroup}
    Let $\mcH = (V,E,I)$ be a (not necessarily simple) hypergraph and let $b$
    be a function $V \arr \Z_2 : v \mapsto b_v$.  The \emph{solution group} $\Gamma =
    \Gamma(\mcH;b)$ associated to $\mcH$ and $b$ is the group generated by
    $\left\{x_e, e \in E\right\} \cup \{J\}$, subject to relations:
    \begin{enumerate}[(1)]
        \item $x_e^2 = 1$ for all $e \in E$ and $J^2 = 1$ (i.e. $\Gamma$ is generated by involutions)
        \item $[x_e, J] = 1$ for all $e \in E$ (i.e. $J$ is central), 
        \item $[x_e,x_{e'}] = 1$ if there is some vertex $v$ incident to both
            $e$ and $e'$, and
        \item \begin{equation*}
            \prod_{e} x_e^{A_{ve}} = J^{b_v} \text{ for all } v \in V.
        \end{equation*}
    \end{enumerate}
    The \emph{null solution group} is the group $\Gamma(\mcH) := \Gamma(\mcH,0)$.
\end{defn}
We call the last two types of relations \emph{commuting relations} and \emph{linear
relations} respectively. The definition of the linear relations assumes that $E$ is 
ordered, but the choice of order is irrelevant because of the commuting
relations. Note that if $v$ and $e$ are incident and $A_{ve}$ is even, then the
linear relations
\begin{equation*}
    \prod_{e'} x_{e'}^{A_{ve'}} = J^{b_v} \text{ and } 
    \prod_{e' \neq e} x_{e'}^{A_{ve'}} = J^{b_v}
\end{equation*}
are equivalent. However, the fact that $A_{ve} > 0$ might still lead to
commuting relations that wouldn't hold otherwise. 
\begin{example}\label{Ex:simple}
    \begin{figure}[tbhp]
        \begin{tikzpicture}[auto,ultra thick,vertex/.style={circle,draw,thin,inner sep=2.5},every node/.style={scale=.8},
                    hyperedge/.style={thin,pattern=north west lines}]

    \node[vertex] (1) at (2,3) {1};
    \node[vertex] (2) at (0,0) {2};
    \node[vertex] (3) at (4,0) {3};

    \draw (1) to [out=135,in=180] (2,4) to [out=0,in=45] (1);
    \draw[hyperedge] (1.-120) to (2.60) to[bend left] (2.0) to[bend left] (3.180) to[bend left] (3.120) to (1.-60) to[bend left] (1.-120);
    \draw[hyperedge] (2) to[out=-160,in=-90] (-1,0) to[out=90,in=160] (2) to[bend right] (2.-160);
    \draw[hyperedge] (3) to[out=-20,in=-90] (5,0) to[out=90,in=20] (3) to[bend left] (3.-20);
    \draw (2) to[bend right] (3);
\end{tikzpicture}
        \caption{The hypergraph from Example \ref{Ex:simple}. Edges with
            degree two are drawn as lines, while edges with degree not equal
            to two are drawn as shaded regions.}
        \label{F:hypergraphex}
    \end{figure}
    Consider the hypergraph $\mcH$ with matrix  
    \begin{equation*}
        A = \begin{pmatrix} 1 & 0 & 0 & 0 & 2 \\
                            1 & 1 & 0 & 1 & 0 \\
                            1 & 0 & 1 & 1 & 0 \\
            \end{pmatrix},
    \end{equation*}
    represented visually in Figure \ref{F:hypergraphex}, and let
    $b = (0,0,1)$.
    Then $\Gamma(\mcH,b)$ is the group generated by $x_1,\ldots,x_5,J$,
    satisfying the relations 
    \begin{equation*}
        J^2=1, x_i^2 = [x_i,J]=1 \text{ for all }1 \leq i \leq 5,
    \end{equation*}
    commuting relations 
    \begin{equation*}
        [x_i,x_j] =1 \text{ for }(i,j) \in \{(1,5),(1,2),(1,4),(2,4),(1,3),(3,4)\},
    \end{equation*}
    and linear relations
    \begin{equation*}
        x_1 x_5^2 = x_1 x_2 x_4 = 1, x_1 x_3 x_4 = J.
    \end{equation*}
    In this example, the first linear relation implies $x_1=1$, so $\Gamma(\mcH,b)$
    is isomorphic to $(\Z_2 * \Z_2) \times \Z_2$.
\end{example}

The one-dimensional representations $\pi$ of $\Gamma(\mcH;b)$ in which $\pi(J) \neq 1$
correspond to the solutions of the linear system $Ax = b$. Higher-dimensional
representations of $\Gamma(\mcH;b)$ with this property can be thought of as
quantum solutions of $Ax=b$. This is justified by the following theorem, which
relates solution groups to non-local games.
\begin{thm}[\cite{CM14, CLS16}]\label{T:games}
    Let $Ax = b$ be a linear system over $\Z_2$, where $A$ is a non-negative
    integral matrix, let $\mcG$ be the associated linear system non-local game, and let
    $\Gamma$ be the corresponding solution group. Then:
    \begin{itemize}
        \item $\mcG$ has a perfect quantum commuting-operator strategy
            if and only if $J \neq 1$ in $\Gamma$. 
        \item $\mcG$ has a perfect quantum tensor-product strategy if and only
            if $\mcG$ has a perfect finite-dimensional quantum strategy, and this
            happens if and only if $\Gamma$ has a finite-dimensional
            representation $\pi$ with $\pi(J) \neq \Id$. 
    \end{itemize}
\end{thm}
The first part of this theorem is due to \cite{CLS16}, while the second part is
due to \cite{CM14}.

\section{The embedding theorem and consequences}\label{S:embedding}

In light of Theorem \ref{T:games}, we would like to understand the structure
(or lack thereof) of solution groups for linear system games.  Recall that a
presentation $\langle S : R \rangle$ of a group $G$ is a set $S$ and subset $R$
of the free group $\mcF(S)$ generated by $S$, such that $G = \mcF(S) / (R)$,
where $(R)$ is the normal subgroup generated by $R$. A group is finitely
presented if it has a presentation $\langle S:R \rangle$ where both $S$ and $R$
are finite. Our primary result is that understanding solution groups is as hard
as understanding finitely presented groups.
\begin{thm}\label{T:embedding}
    Let $G$ be a finitely presented group, let $J' \in G$ be a central element
    with $(J')^2 = 1$, and let $w_1,\ldots,w_n$, $n \geq 0$ be a sequence of
    elements in $G$ such that $w_i^2 = 1$ for all $1 \leq i \leq n$. Then there
    is a hypergraph $\mcH$, a vertex labelling function $b : V(\mcH) \arr
    \Z_2$, a sequence of edges $e_1,\ldots,e_n$ in $\mcH$, and a
    homomorphism $\phi : G \arr \Gamma(\mcH,b)$  such that $\phi$ is an
    embedding, $\phi(J')=J$, and $\phi(w_i) = x_{e_i}$ for all $1 \leq i \leq n$.
\end{thm}
When $J'=1$ and $n=0$, Theorem \ref{T:embedding} states that every finitely
presented group $G$ embeds in some solution group $\Gamma$. The full version of
Theorem \ref{T:embedding} says that if we also specify a list of involutions in
$G$, then we can construct the embedding so that these involutions map to
generators of $\Gamma$, and in particular we can map a specified central involution to
the element $J \in \Gamma$. This last condition is essential for our
applications.

The hypergraph $\mcH$, vertex labelling $b$ and homomorphism
$\phi$ in Theorem \ref{T:embedding} can be explicitly constructed from a
presentation of $G$ and a choice of representatives for $w_1,\ldots,w_n$ and
$J'$. The construction is recursive, in the sense that there is a Turing
machine which, given a presentation for $G$ and words representing
$w_1,\ldots,w_n$, and $J'$ as input, outputs $\mcH$, $b$, and $\phi$ as in the
theorem. In particular, the construction  does not depend on whether the
elements $w_1,\ldots,w_n,J$ are trivial or non-trivial in $G$. This
construction is described in Sections \ref{S:involutions} and
\ref{S:wagonwheel}. Theorem \ref{T:embedding} is proved at the beginning of
Section \ref{S:wagonwheel} via reduction to another embedding theorem. The
proof of this latter embedding theorem (and hence the proof of Theorem
\ref{T:embedding}) is completed in Section \ref{S:proof}.

In the remainder of this section, we prove two consequences of Theorem
\ref{T:embedding}. The first is an answer to Tsirelson's problem.
\begin{cor}\label{C:tsirelson}
    There is a linear system non-local game which has a perfect quantum commuting-operator
    strategy, but does not have a perfect quantum tensor-product strategy. 
\end{cor}
For the proof, we recall the notion of an HNN extension. Let $\alpha : K_0 \arr
K_1$ be an isomorphism between two subgroups $K_0$ and $K_1$ of a group $H =
\langle S : R \rangle$. The HNN extension of $H$ by $\alpha$ is the group with
presentation 
\begin{equation*}
    G = \langle S \cup \{t\} : R \cup \{t k t^{-1} = \alpha(k) : k \in K_0\} \rangle,
\end{equation*}
where $t$ is an indeterminate not in $S$. It is well-known that the natural
map $H \incl G$ is an injection, so $H$ is a subgroup of $G$ (see, e.g., \cite[Section IV.2]{LS77}).
\begin{proof}[Proof of Corollary \ref{C:tsirelson}]
    Suppose that $G$ is a finitely presented group with a central element $J'$
    of order two, such that $\pi(J')=\Id$ for every finite-dimensional
    representation $\pi$ of $G$. By Theorem \ref{T:embedding}, there is an
    embedding of $G$ in a solution group $\Gamma$ which identifies $J'$ with
    $J$. In particular, this implies that $J \neq 1$ in $\Gamma$, so the
    associated linear system non-local game $G$ must have a perfect quantum
    commuting-operator strategy by the first part of Theorem \ref{T:games}. If $\pi$ is
    a finite-dimensional representation of $\Gamma$, then $\pi(J) = \pi|_G(J')
    = \Id$. By the second part of Theorem \ref{T:games}, the associated linear
    system non-local game does not have a perfect quantum tensor-product
    strategy.

    To finish the proof, we construct a group $G$ with the above property.
    Consider Higman's group
    \begin{equation*}
        H_0 = \langle a,b,c,d : a b a^{-1} = b^2, b c b^{-1} = c^2, c d c^{-1} = d^2, d a d^{-1} = a^2 \rangle.
    \end{equation*}
    It is well-known that $H_0$ has no non-trivial linear representations, and
    that the generators $a,b,c,d$ of $H_0$ have infinite order \cite{Hi51,
    Be94}.  Let $H = H_0 \times \Z_2$, and let $J \in H$ denote the generator
    of the $\Z_2$-factor. Let $G$ be the HNN extension of $H$ by the automorphism
    of $\langle a, J \rangle \iso \Z \times \Z_2$ sending $J \mapsto J$ and 
    $a \mapsto aJ$. By the properties of the HNN extension, $H$ is a subgroup of $G$, 
    and in particular $J$ is non-trivial in $G$. Furthermore, we can construct
    a presentation for $G$ from a presentation of $H$ by adding a generator $x$
    and relations $[x,J]=1$ and $[x,a]=J$. The former relation implies that 
    $J$ is central in $G$. Finally, if $\pi$ is a finite-dimensional
    representation of $G$, then $\pi|_{H_0}$ is trivial, and in particular,
    $\pi(a) = \Id$. But this implies that 
    \begin{equation*}
        \pi(J) = \pi([x,a]) = \Id,
    \end{equation*}
    as required. 
\end{proof}
We note that any non-residually-finite group can be used in place of Higman's
group in the above proof. 

It would be interesting to know the smallest linear system for which the
corresponding game can be played perfectly only with commuting-operator
strategies. No effort is made to reduce the size of the linear system in the
proof of Theorem \ref{T:embedding}, and the main construction from Section
\ref{S:wagonwheel} depends on the number of variables and the total length of
the relations in the presentation of $G$. If we naively follow the proof
through for the group in Corollary \ref{C:tsirelson}, we get a linear system
with roughly 600 variables and 450 linear relations. Making some obvious
improvements in Section \ref{S:involutions} can get this down to 400 variables
and 300 relations. 

The second consequence concerns the difficulty of determining whether a
non-local game has a perfect commuting-operator strategy.
\begin{cor}\label{C:undecidable}
    It is undecidable to determine if a binary linear system game has a perfect
    commuting-operator strategy.
\end{cor}
\begin{proof}
    By Theorem \ref{T:games}, determining if a binary linear system game
    has a perfect strategy is equivalent to determining if $J \neq 1$ in a
    solution group.  Because Theorem \ref{T:embedding} is constructive, this is
    in turn equivalent to the following decision problem: \textit{given a finite group
    presentation $G=\langle S:R \rangle$ and a word $J$ in the generators $S$
    such that $J \in Z(G)$ and $J^2=1$, decide if $J=1$ in $G$}.

    We claim that the word problem for groups can be reduced to this latter
    problem. Indeed, given a finitely presented group $K = \langle S:R \rangle$
    and a word $w \in \mcF(S)$ in the generators of $K$, it is possible to
    recursively construct a finitely presented group $L_w$ with the property
    that $K$ is a subgroup of $L_w$ if $w \neq 1$, and $L_w$ is trivial if $w =
    1$ (see \cite[pg.  190]{LS77}, where this construction is attributed to
    Rabin). The presentation of $L_w$ can be constructed by adding finitely
    many generators and relations to the presentation of $K$. Let $H_w$ be the
    result of 
    applying this construction to the group $K \times \Z$, and let $z$ be the
    generator of $H_w$ corresponding to the generator of the $\Z$-factor in $K \times \Z$. 
    If $w = 1$, then $H_w$ is trival, and $z=1$. If $w \neq 1$, then
    $K \times \Z$ is a subgroup of $H_w$, and the order of $z$ is infinite.
    Finally, construct a group $G_w$ by adding two generators $x$ and $J$ to
    the presentation of $H_w$, along with relations $J^2=[x,J]=1$, $[s,J] = 1$
    for all generators $s$ of $H_w$, and $[x,z] = J$. As in Corollary
    \ref{C:tsirelson}, if $w \neq 1$ then $G_w$ is the HNN extension of
    the group $H_w \times \Z_2$ by the automorphism of the subgroup 
    $\langle z, J \rangle \iso \Z \times \Z_2$ sending $J \mapsto J$ and
    $z \mapsto zJ$. Thus, if $w \neq 1$ then $J \neq 1$ in $G_w$. 
    If $w=1$, then $z=1$ in $H_w$, and consequently $J=[x,z]=1$ in $G_w$.
    This completes the reduction.
\end{proof}
Although not used in either of the above corollaries, Theorem \ref{T:embedding}
also allows us to embed a finitely presented group $G$ in a solution group
$\Gamma$ so that a given set of involutions of $G$ become generators of
$\Gamma$. This can be used to prove that other tasks involving solution groups
are undecidable. For instance, it is undecidable to determine if a generator
$x_e$ of a null solution group is non-trivial. 

\section{Presentations by involutions}\label{S:involutions}

We are interested primarily in groups which (a) have a distinguished central
element of order $\leq$ two, and (b) are generated by involutions.  For clarity
in subsequent sections, we encode these conditions in two formal definitions.
\begin{defn}\label{D:overZ2}
    A \emph{group over $\Z_2$} is a group $G$ with a distinguished
    central element $J = J_{G}$ such that $J^2 = 1$.

    A \emph{morphism $G_1 \arr G_2$ over $\Z_2$} is a group homomorphism sending
    $J_{G_1} \mapsto J_{G_2}$. Similarly, an \emph{embedding over $\Z_2$} is an
    injective morphism over $\Z_2$. 
\end{defn} 
Note that $J$ is allowed to be the identity in this definition. This is so that
we can construct groups over $\Z_2$ by starting with some finite presentation,
picking an element $J' \in \mcF(S)$, and adding relations $(J')^2 = 1$ and
$[J',s]=1$ for all $s \in S$. By allowing $J=1$, we can do this even if $J'$
becomes trivial.

Elements of $\mcF(S)$ are represented by words over $\{s, s^{-1} : s \in S\}$.
Every element $r \in \mcF(S)$ can be represented uniquely as $s_1^{a_1} \cdots
s_n^{a_n}$, where $a_i \in \{\pm 1\}$, and $a_i = a_{i+1}$ whenever $s_i =
s_{i+1}$. A word meeting these conditions is said to be \emph{reduced}. The
number $n$ is the \emph{length} of $r$. The element $r$ is said to be
\emph{cyclically reduced} if, in addition, $s_n = s_1$ implies that $a_n =
a_1$. 

\begin{defn}\label{D:involutions}
    Given a set $S$, let $\mcF_2(S) = \langle S : s^2 = 1, s \in S \rangle$.  A
    \emph{presentation by involutions over $\Z_2$} for a group $G$ is a set of
    generators $S$ and a set of relations $R \subset \mcF_2(S) \times \Z_2$ such
    that $G = \mcF_2(S) \times \Z_2 / (R)$, where $(R)$ is the normal subgroup
    generated by $R$. We denote presentations of this form by $\Inv\langle S : R\rangle$,
    and write $G = \Inv\langle S:R\rangle$ when the meaning is clear. 

    We use $J$ (written in multiplicative notation) to denote the generator of
    the $\Z_2$-factor in $\mcF_2(S) \times \Z_2$. If $G = \Inv\langle S:R \rangle$,
    we can regard $G$ as a group over $\Z_2$ by letting $J = J_G$ be the image
    of $J \in \mcF_2(S) \times \Z_2$ in $G$. 

    Elements of $\mcF_2(S)$ are represented by words over $S$.  Every element
    $r \in \mcF_2(S) \times \Z_2$ can be represented uniquely as $J^a s_1
    \cdots s_n$, where $a \in \Z_2$ and $s_1,\ldots,s_n$ is a sequence in $S$
    with $s_i \neq s_{i+1}$. Again, a word of this form is said to be 
    \emph{reduced}, and $n$ is called the \emph{length} of $r$. If, in
    addition, $s_n \neq s_1$ then we say that $r$ is \emph{cyclically reduced}.
    We say that a set of relations $R$ is \emph{cyclically reduced} if every
    element of $R$ is cyclically reduced. 

    If $R$ is a set of relations, the \emph{symmetrization} of $R$ is the 
    set of relations $R^{sym}$ containing all relations of the form 
    \begin{equation*}
        J^a s_i s_{i+1} \cdots s_n s_1 \cdots s_{i-1}\text{ and }J^a s_{i} s_{i-1}
            \cdots s_1 s_n \cdots s_{i+1}, 1 \leq i \leq n, 
    \end{equation*}
    for every relation $J^a s_1 \cdots s_n$ in $R$. 
\end{defn}
Every group presented by involutions over $\Z_2$ has a presentation
$\Inv\langle S:R \rangle$ where $R$ is cyclically reduced. The presentations
$\Inv\langle S : R \rangle$ and $\Inv \langle S : R^{sym} \rangle$ are
equivalent (i.e.  they define isomorphic groups), and if $R$ is cyclically
reduced then $R^{sym}$ is cyclically reduced. 

By definition, solution groups are examples of groups presented by involutions
over $\Z_2$. Theorem \ref{T:embedding} states that every finitely presented
group over $\Z_2$ embeds (over $\Z_2$) in a solution group.  The first step in
proving Theorem \ref{T:embedding} is showing that every finitely presented
group embeds in a group presented by involutions.
\begin{prop}\label{P:involutions}
    Suppose $(G,J)$ is a group over $\Z_2$ with finite presentation $\langle S : R
    \rangle$, and $J' \in \mcF(S)$ is a representative of $J_G$. Let $T$ be the
    set of indeterminates $\{z_{s1},z_{s2} : s \in S\}$, and choose integers $k_s
    \geq 1$ for all $s \in S$. Finally, let $\phi : \mcF(S) \arr \mcF_2(T)
    \times \Z_2$ be the morphism sending $s \mapsto \left( z_{s1}
    z_{s2}\right)^{k_s}$. Then the induced morphism
    \begin{equation*}
        \phi : G \arr K := \Inv\langle T : R' \rangle, \text{ where } R' := \{\phi(r) : r \in R\} \cup \{J_K\phi(J')\}
    \end{equation*}
    is an embedding over $\Z_2$. 

    Furthermore, if $R \cup \{J'\}$ is cyclically reduced then $R'$ is cyclically reduced. 
\end{prop}
\begin{proof}
    Let $s_1,\ldots,s_n$ be a list of the elements of $S$, and let $m_i$ be the
    order of $s_i$ in $G$. For convenience, we write $z_{ij}$ in place of
    $z_{s_i,j}$. For each $0 \leq r \leq n$, let 
    \begin{align*}
        K_r := \langle S \cup \{z_{11},z_{12},\ldots,z_{r1},z_{r2},J\} : R & \cup \{ z_{ij}^2 = [z_{ij},J] = 1 : 
            1 \leq i \leq r, j=1,2\} \\ & \cup \{ s_j = \left(z_{i1} z_{i2}\right)^{k_{s_j}} \} 
            \cup \{ J = J' \} \rangle.
    \end{align*}
    The presentation $K_n$ is equivalent to the presentation $\Inv\langle T :
    R' \rangle$ of the group $K$, so we just need to show that the natural map
    $G \arr K_n$ is an inclusion. But this follows from the fact that each $K_i$,
    $1 \leq i \leq n$, is an amalgamated product of $K_{i-1}$ with either a
    dihedral group, or the product of a dihedral group and $\Z_2$. Indeed, 
    suppose that the natural map of $G \arr K_{i-1}$ is an inclusion, and let
    \begin{equation*}
        D_i := \langle z_{i1}, z_{i2} : z_{i1}^2 = z_{i2}^2 = \left(z_{i1} z_{i2}\right)^{k_{s_i} m_i}
                                = 1 \rangle,
    \end{equation*}
    the dihedral group of order $2 k_{s_i} m_i$ (if $m_i$ is infinite, then the
    last relation is omitted, so that $D_i$ is the infinite dihedral group).
    If $J_G \not\in \langle s_i \rangle$, then $\langle s_i, J_G \rangle
    \iso \Z_{m_i} \times \Z_2$, and $K_i$ is the amalgamated product of
    $K_{i-1}$ with $D_i \times \Z_2$ over $\langle s_i, J_G \rangle$, where we
    identify $s_i \in G \subseteq K_{i-1}$ with $(z_{i1} z_{i2})^{k_{s_i}}$,
    and $J_G$ with the generator of $\Z_2$ in $D_i \times \Z_2$. If, on the
    other hand, $J_G \in \langle s_i \rangle$, then $J_G = s_i^a$, where
    $a = 0$ or $m_i/2$. In both cases, $(z_{i1} z_{i2})^{k_{s_i} a}$ is
    central in $D_i$, so $K_i$ is the amalgamated product of $K_{i-1}$ with
    $D_i$ over $\langle s_i \rangle$, where we again identify $s_i$ with
    $(z_{i1} z_{i2})^{k_{s_i}}$ (it's necessary to check that $z_{i1}$ and
    $z_{i2}$ commute with $s_i^a$ because of the relations $[z_{ij},J]$ in
    $K_i$).  It follows that the natural map $G \arr K_i$ is an inclusion.
    Since $G \iso K_0$, we ultimately conclude that the natural map $G \arr
    K_n$ is an inclusion as desired. 

    Finally, it is easy to see that if $r \in R$ is cyclically reduced, then
    $\phi(r)$ is cyclically reduced. If $J'$ is cyclically reduced, then $J
    \phi(J')$ is also cyclically reduced. 
\end{proof}

\begin{defn}\label{D:evenodd}
    A relation $r = J^{a} s_1 \cdots s_n \in \mcF_2(S) \times \Z_2$ is
    \emph{odd} (resp. \emph{even}) if $a$ is odd (resp. even). Equivalently,
    a relation is odd (resp. even) if it is of the form $r' = J$ (resp. $r'=1$)
    for some $r' \in \mcF_2(S)$. 

    The \emph{even part} of the relation $r = J^a s_1 \cdots s_n$ is
    $r^+ = s_1 \cdots s_n$. If $\Inv\langle S : R \rangle$ is a presentation by involutions
    over $\Z_2$, then the corresponding \emph{even presentation} over $\Z_2$ is
    $\Inv\langle S : R^+\rangle$, where $R^+ = \{ r^+ : r \in R\}$. 
    
    Similarly, if $G$ is any group over $\Z_2$, then the \emph{even quotient}
    is $G^+ := G/(J_G) \times \Z_2$. The group $G^+$ is regarded as a group
    over $\Z_2$ with $J_{G^+}$ equal to the generator of the $\Z_2$ factor.
\end{defn}
It is easy to see that if $G = \Inv\langle S : R\rangle$, then $G^+ = \Inv\langle S : R^+\rangle$. For
instance, the null solution group $\Gamma(\mcH) = \Gamma(\mcH,0)$ of a
hypergraph $\mcH$ is the even quotient $\Gamma(\mcH,b)^+$ of the solution group
$\Gamma(\mcH,b)$ for any $b$. 

\begin{defn}
    Let $J^a s_1 \cdots s_n$ be a reduced word for an element $r \in \mcF_2(S)
    \times \Z_2$. The multiplicity of $s \in S$ in $r$ is
    \begin{equation*}
        \mult(s;r) := |\{ 1 \leq i \leq n : s_i = s\}|.
    \end{equation*}
    We say that $s \neq t \in S$ are \emph{adjacent in $r$} if either $\{s,t\} =
    \{s_i,s_{i+1}\}$ for some $i=1,\ldots,n-1$ or $\{s,t\} = \{s_1,s_n\}$. 
\end{defn}
The reason that we introduce numbers $k_s$ in Proposition \ref{P:involutions}
is that, for the proof of Theorem \ref{T:embedding}, we would like to work with
relations $r$ where $\mult(s;r)$ is even for all $s$. In fact, we will be
able to handle slightly more general presentations, which we now define.

\begin{defn}\label{D:collegial}
    We say that a presentation $\Inv\langle S:R \rangle$ by involutions over $\Z_2$ is
    \emph{collegial} if 
    \begin{enumerate}[(a)]
        \item the presentation is finite and cyclically reduced,
        \item $R \cap \{1,J\} = R \cap S = \emptyset$, and
        \item if $\mult(s;r_0)$ is odd for some $r_0 \in R$, and $t$
            is adjacent to $s$ in some $r_1 \in R$, then $\mult(t;r')$ is
            even for all $r' \in R$.
    \end{enumerate}
\end{defn}
\begin{rmk}\label{R:atleastthree}
    Note that if $\Inv\langle S:R \rangle$ is collegial, then every relation $r \in R$ must
    have length at least four, i.e. $r = J^a s_1 \cdots s_n$ where $n \geq 4$. 
    This is because relations of length zero and one are explicitly excluded by
    condition (b), relations of the form $s^2$, $s t s$, $s^2 t$, or $t s^2$
    are not cyclically reduced, and relations $s t$ and $s t r$, where $s,t,r$
    are distinct, do not satisfy condition (c).
\end{rmk}
\begin{cor}\label{C:collegial}
    Let $G$ be a finitely presented group over $\Z_2$, with a sequence of
    elements $w_1,\ldots,w_n \in G$ such that $w_i^2 = 1$ for all $1 \leq i\leq
    n$. Then there is collegial presentation $\Inv\langle S:R \rangle$, and an
    embedding $\phi : G \arr K := \Inv\langle S:R\rangle$ over $\Z_2$ such that
    $\phi(w_i) \in S \subset K$ for all $1 \leq i \leq n$.
\end{cor}
\begin{proof}
    We can find a cyclically reduced presentation $\langle S_0 : R_0 \rangle$
    for $G$ in which $J_G$ is a generator, $1 \not\in R_0$, and each $w_i$ has
    a representative $w_i' \in \mcF(S_0)\setminus\{1\}$ (it is always possible
    to find such a presentation, since if necessary we can add an extra
    generator $z$, along with the relation $z=1$, and use this as a
    representative of the identity). In particular, this gives us a
    presentation where $J_G$ is represented by a cyclically reduced
    non-identity element of $\mcF(S_0)$, namely itself. 

    Applying Proposition \ref{P:involutions} to this presentation with $k_s =
    2$ (or any other even number) for all $s \in S_0$ gives us an embedding
    $\phi$ of $G$ in a finite presentation $\Inv\langle T:R'\rangle$, where $R'$ is
    cyclically reduced. Since all $k_s$'s are even, $\mult(t;\phi(r))$ is even
    for every $r \in \mcF(S_0)$ and $t \in T$, and every relation in $R'$ has
    length $\geq 4$. We conclude that $\Inv\langle T:R' \rangle$ is collegial. 

    Now let
    \begin{equation*}
        S = T \cup \{\overline{w}_1,\ldots,\overline{w}_n\}, 
    \end{equation*} 
    where $\overline{w}_1,\ldots, \overline{w}_n$ are new indeterminates, and set 
    \begin{equation*} 
        R = R' \cup \{\overline{w}_i \phi(w'_i) : 1 \leq i \leq n\},
    \end{equation*}
    where $\phi : \mcF(S_0) \arr \mcF_2(T) \times \Z_2$ as in Proposition
    \ref{P:involutions}.  Since $\overline{w}_i$ does not appear in
    $\phi(w_i')$, the relation $r = \overline{w}_i \phi(w_i')$ is cyclically
    reduced. Furthermore, none of the $\overline{w}_i$'s are adjacent, and
    $\mult(s;r)$ is even for all $s \in T$ and $r \in R$, so $\Inv\langle S:R\rangle$ is
    collegial. But $\Inv\langle S:R\rangle$ is plainly equivalent to $\Inv\langle T:R'\rangle$, so the
    corollary follows.
\end{proof}
Note that the proof of Corollary \ref{C:collegial} is constructive, in the
sense that there is a Turing machine which outputs $S$, $R$, and $\phi$ (as
specified by a set of representatives in $\mcF_2(S)$ for the images through
$\phi$ of the generators of $G$) given a finite presentation for $G$ and a
set of representatives for $w_1,\ldots,w_n$. 

\section{The wagon wheel embedding}\label{S:wagonwheel}

Using Corollary \ref{C:collegial}, the proof of Theorem \ref{T:embedding}
reduces to the following:
\begin{theorem}\label{T:invembedding}
    Let $G$ be a group with a collegial presentation $\mcI = \Inv\langle S:R\rangle$. Then
    there is a hypergraph $\mcW := \mcW(\mcI)$ and vertex labelling $b := b(\mcI)$
    such that $S \subset E(\mcW)$, and the resulting map 
    \begin{equation*}
        \mcF(S) \times \Z_2 \arr \Gamma(\mcW,b) : s \mapsto x_s
    \end{equation*} 
    descends to an embedding $G \incl \Gamma(\mcW,b)$ over $\Z_2$.  
\end{theorem}
\begin{proof}[Proof of Theorem \ref{T:embedding} using Theorem \ref{T:invembedding}]
    Let $G$ be a group over $\Z_2$ with elements $w_1,\ldots,w_n$ such that
    $w_i^{2} = 1$ for $i=1,\ldots,n$. By Corollary \ref{C:collegial}, there is
    a collegial presentation $\mcI := \Inv\langle S:R \rangle$ and an embedding
    $\phi_1 : G \arr K := \Inv\langle S:R\rangle$ over $\Z_2$ with $\phi_1(w_i)
    \in S$ for all $i=1,\ldots,n$. 

    By Theorem \ref{T:invembedding}, there is an embedding $\phi_2 : K \arr
    \Gamma(\mcW(\mcI),b(\mcI))$ over $\Z_2$ with $\phi_2(s) = x_s$ for all $s
    \in S$. The composition $\phi_2 \circ \phi_1$ satisfies the conditions of
    Theorem \ref{T:embedding}.
\end{proof}
Although we are still very far from being able to prove Theorem
\ref{T:invembedding}, in this section we shall describe the hypergraph
$\mcW(\mcI)$, which we call the \emph{wagon wheel hypergraph} of $\mcI$. The
proof of Theorem \ref{T:invembedding} will be given in Section \ref{S:proof}.

\begin{figure}
    \begin{tikzpicture}[auto,ultra thick,scale=.6,emptynode/.style={inner sep=0},
    every node/.style={scale=.8},
    vertex/.style={circle,draw,thin,inner sep=2.5,scale=.55,fill=white},
    helabel/.style={fill=white,scale=.8}]
    \draw (-10:3) arc [start angle=-10,end angle=150,radius=3];
    \draw[dashed] (150:3) arc [start angle=150,end angle=350,radius=3];
    \node[vertex] (0) at (-10:3) {$3,3$};
    \node[vertex] (1) at (30:3) {$2,3$};
    \node[vertex] (2) at (70:3) {$1,3$};
    \node[vertex] (3) at (110:3) {$0,3$};
    \node[vertex] (4) at (150:3) {$-1,3$};
    \path (0) arc[start angle=-10,end angle=30,radius=3] node[swap,pos=.35] {$d_3$};
    \path (1) arc[start angle=30,end angle=70,radius=3] node[swap,pos=.5] {$d_2$};
    \path (2) arc[start angle=70,end angle=110,radius=3] node[swap,pos=.5] {$d_1$};
    \path (3) arc[start angle=110,end angle=150,radius=3] node[swap,pos=.5] {$d_{0}$};

    \draw (-10:6) arc [start angle=-10,end angle=150,radius=6];
    \draw[dashed] (150:6) arc [start angle=150,end angle=350,radius=6];
    \node[vertex] (5) at (30:6) {$2,2$}
        edge node[swap,pos=.3] {$c_{2}$} (1);
    \node[vertex] (6) at (50:6) {$2,1$};
    \node[vertex] (7) at (70:6) {$1,2$}
        edge node[swap] {$c_{1}$} (2);
    \node[vertex] (8) at (90:6) {$1,1$};
    \node[vertex] (9) at (110:6) {$0,2$}
        edge node[swap,pos=.3] {$c_{0}$} (3);
    \node[vertex] (10) at (130:6) {$0,1$};
    \node[vertex] (11) at (150:6) {$-1,2$}
        edge node[swap] {$c_{-1}$} (4);
    \node[vertex] (12) at (10:6) {$3,1$};
    \node[vertex] (13) at (-10:6) {$3,2$}
        edge node[swap] {$c_3$} (0);
    \path (5) arc[start angle=30, end angle=50,radius=6] node[pos=.4,swap] {$b_2$};
    \path (6) arc[start angle=50, end angle=70,radius=6] node[pos=.58,swap] {$a_2$};
    \path (7) arc[start angle=70, end angle=90,radius=6] node[pos=.58,swap] {$b_1$};
    \path (8) arc[start angle=90, end angle=110,radius=6] node[pos=.32,swap] {$a_1$};
    \path (9) arc[start angle=110, end angle=130,radius=6] node[pos=.35,swap] {$b_{0}$};
    \path (10) arc[start angle=130, end angle=150,radius=6] node[pos=.5,swap] {$a_{0}$};
    \path (13) arc[start angle=-10,end angle=10,radius=6] node[pos=.5,swap] {$b_{3}$};
    \path (12) arc[start angle=10,end angle=30,radius=6] node[pos=.5,swap] {$a_{3}$};

    \path[fill,pattern=north west lines] (12.20) to ++(40:2) to ($(12.0)+(-20:2)$) 
        to (12.0) to [bend right] (12.20);
    \draw[thin] (12.20) to ++(40:2);
    \draw[thin] (12.0) to node[helabel,pos=.35,yshift=4] {$s_{3}$} ++(-20:2);

    \path[fill,pattern=north west lines] (6.60) to ++(80:2) to ($(6.40)+(20:2)$) 
        to (6.40) to [bend right] (6.60);
    \draw[thin] (6.60) to ++(80:2);
    \draw[thin] (6.40) to node[helabel,pos=.55,yshift=2] {$s_{2}$} ++(20:2);

    \path[fill,pattern=north west lines] (8.100) to ++(120:2) to ($(8.80)+(60:2)$) 
        to (8.80) to [bend right] (8.100);
    \draw[thin] (8.100) to ++(120:2);
    \draw[thin] (8.80) to node[helabel,pos=.5,xshift=-5] {$s_{1}$} ++(60:2);

    \path[fill,pattern=north west lines] (10.140) to ++(160:2) to ($(10.120)+(100:2)$) 
        to (10.120) to [bend right] (10.140);
    \draw[thin] (10.140) to ++(160:2);
    \draw[thin] (10.120) to node[helabel,pos=.5,xshift=-5] {$s_{n_i}$} ++(100:2);
\end{tikzpicture}
    \caption{The portion of the wagon wheel hypergraph containing vertices $V_i$
            and all incident edges. To save space, $(i,j,k)$ is written as
            $j,k$, and $s_{ij},a_{ij},\ldots$ are written as $s_j, a_j,\ldots$.}
    \label{F:wagonwheel}
\end{figure}
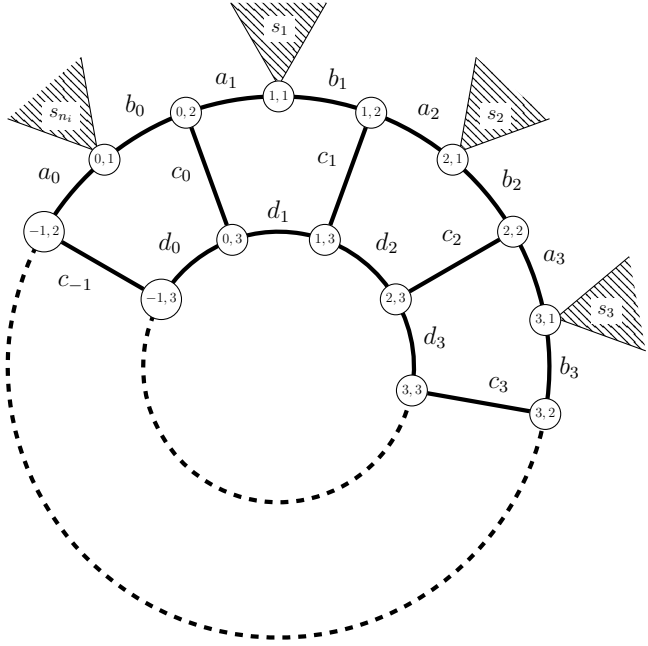
The wagon wheel hypergraph can be defined for any (not necessarily collegial)
presentation $\Inv\langle S,R\rangle$. Let $R = \{r_1,\ldots,r_m\}$, let $n_i$ be the length
of $r_i$, and write $r_i = J^{p_i} s_{i1} \cdots s_{in_i}$, where $s_{ij} \in
S$. The wagon wheel hypergraph is a  simple hypergraph $\mcW$ with vertex set
\begin{equation*}
    V := \{(i,j,k) : 1 \leq i \leq m, j \in \Z_{n_i}, 1 \leq k \leq 3\}, 
\end{equation*}
and edge set
\begin{equation*}
    E := S \sqcup \{a_{ij},b_{ij},c_{ij},d_{ij} : 1 \leq i \leq m, j \in \Z_{n_i}\}.
\end{equation*}
As a result, if $M := \sum_{i=1}^k n_i$, then $\mcW$ has $3M$ vertices and
$4M+|S|$ edges. $\mcW$ has the following incidence relations for every $1 \leq
i \leq m$ and $1 \leq j \leq n_i$:
\begin{itemize}
    \item $s \in S$ is incident with $(i,j,1)$ if and only if $s_{ij} = s$,
    \item $a_{ij}$ is incident with $(i,j-1,2)$ and $(i,j,1)$, 
    \item $b_{ij}$ is incident with $(i,j,1)$ and $(i,j,2)$,
    \item $c_{ij}$ is incident with $(i,j,2)$ and $(i,j,3)$, and
    \item $d_{ij}$ is incident with $(i,j-1,3)$ and $(i,j,3)$.
\end{itemize}
Note that the only edges incident with vertices
\begin{equation*}
    V_i := \{(i,j,k) : j \in \Z_{n_i}, 1 \leq k \leq 3 \}
\end{equation*}
are the edges in 
\begin{equation*}
    E_i := \{a_{ij},b_{ij},c_{ij},d_{ij} : j \in \Z_{n_i}\}
\end{equation*}
and the edges $s_{i1},\ldots,s_{in_i}$. Furthermore, all the edges in $E_i$ are
incident with exactly two vertices, both belonging to $V_i$. The portion
of the hypergraph $\mcW$ incident with $V_i$ is shown in Figure \ref{F:wagonwheel}.
An example of a wagon wheel hypergraph for a small (non-collegial) presentation
is shown in Figure \ref{F:wagonwheelexample}.
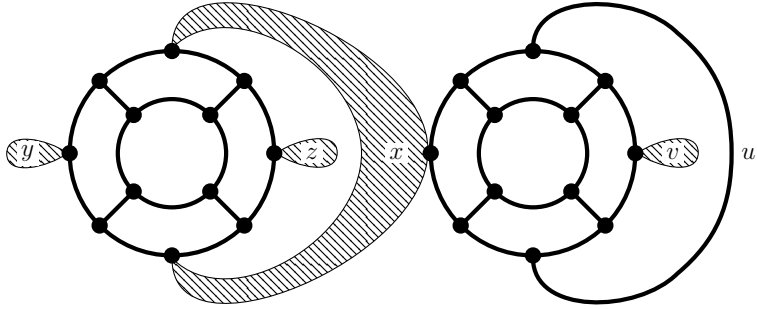
\begin{figure}
    \begin{tikzpicture}[auto,ultra thick,scale=.4,emptynode/.style={inner sep=0},every node/.style={scale=.8}]

    \begin{scope}
        \foreach \x in {1,...,4}
            \node[emptynode] (\x) at (90*\x-45:1.8) {};
        \foreach \x in {5,...,12}
            \node[emptynode] (\x) at (45*\x-225:3.4) {};
        \draw circle [radius=1.8];
        \draw circle [radius=3.4];
        \draw (1) -- (6);
        \draw (2) -- (8);
        \draw (3) -- (10);
        \draw (4) -- (12);

        \foreach \x in {5,9} {
            \draw[fill,pattern=north west lines,thin] (\x.45*\x-215) 
                to [out=45*\x-190,in=45*\x-135] ++(45*\x-225:2)
                to [out=45*\x-315,in=45*\x-260] (\x.45*\x-235)
                to [bend right] (\x.45*\x-215);
        }

        \foreach \x in {1,...,12}
            \draw[fill] (\x) circle [radius=.2];

        \node[right,fill=white,inner sep=2,xshift=12] at (5) {$z$};
        \node[left,fill=white,inner sep=2,xshift=-14] at (9) {$y$};
    \end{scope}

    \node[emptynode] (xt) at (7) {};
    \node[emptynode] (xb) at (11) {};

    \begin{scope}[shift={++(12,0)}]
        \foreach \x in {1,...,4}
            \node[emptynode] (\x) at (90*\x-45:1.8) {};
        \foreach \x in {5,...,12}
            \node[emptynode] (\x) at (45*\x-225:3.4) {};
        \draw circle [radius=1.8];
        \draw circle [radius=3.4];
        \draw (1) -- (6);
        \draw (2) -- (8);
        \draw (3) -- (10);
        \draw (4) -- (12);

        \foreach \x in {5} {
            \draw[fill,pattern=north west lines,thin] (\x.45*\x-215) 
                to [out=45*\x-190,in=45*\x-135] ++(45*\x-225:2)
                to [out=45*\x-315,in=45*\x-260] (\x.45*\x-235)
                to [bend right] (\x.45*\x-215);
        }

        \draw (7) to [out=90,in=135] ($(5)+(1.4,4)$) 
                  to [out=-40, in=90] ($(5)+(3.2,0)$) node[right] {$u$} 
                  to [out=-90, in=40] ($(5)+(1.4,-4)$) 
                  to [out=-135,in=-90] (11);

        \foreach \x in {1,...,12}
            \draw[fill] (\x) circle [radius=.2];

        \node[right,fill=white,inner sep=2,xshift=12] at (5) {$v$};
    \end{scope}

    \node[emptynode] (xr) at (9) {};
    \draw[fill,pattern=north west lines,thin] (xt.90) 
        to [out=90,in=90] (xr.135)
        to [out=-90,in=-90] (xb.-90)
        to [out=-45,in=-90] ($(xr)+(-2.3,0)$)
        to [out=90,in=45] (xt.90);
    \node[left,xshift=-10,fill=white,inner sep=2] at (xr) {$x$};
\end{tikzpicture}
    \caption{An example of the wagon wheel hypergraph $\mcW(\mcI)$ when $\mcI =
    \Inv\langle x,y,z,u,v : xyxz=xuvu=1 \rangle$.}

    \label{F:wagonwheelexample}
\end{figure}

We also need to define the vertex labelling in Theorem \ref{T:invembedding}.
\begin{defn}
    An \emph{$\mcI$-labelling} of $\mcW$ is a vertex labelling $b : V \arr \Z_2$
    such that $|b^{-1}(1) \cap V_i| = p_i \mod 2$ for all $1 \leq i \leq m$.
\end{defn}
For Theorem \ref{T:invembedding}, we can choose any $\mcI$-labelling. This
is because all $\mcI$-labellings are equivalent in the following sense:
\begin{lemma}\label{L:vertex}
    Let $b$ and $b'$ be two $\mcI$-labellings of $\mcW$. Then there is an
    isomorphism $\Gamma(\mcW,b) \arr \Gamma(\mcW,b')$ which sends $x_s \mapsto
    x_s$ for all $s \in S$.
\end{lemma}
\begin{proof}
    Suppose $\mcH$ is a hypergraph with incidence matrix $A(\mcH)$ and vertex
    labelling $b^{(0)}$.  Given $e \in E(\mcH)$, let $b^{(1)}$ be the vertex
    labelling with $b^{(1)}_v = b^{(0)}_v + A(\mcH)_{ve}$ (i.e. we toggle the
    sign of all vertices incident with $e$ according to multiplicity). Then
    there is an isomorphism
    \begin{equation*}
        \Gamma(\mcH,b^{(0)}) \arr \Gamma(\mcH,b^{(1)}) : x_f \mapsto \begin{cases}
            x_f & f \neq e \\
            J x_e & f = e
        \end{cases}.
    \end{equation*}
    For $\mcW$, since $|b^{-1}(1) \cap V_i|$ and $|(b')^{-1}(1) \cap V_i|$ have
    the same parity, it is easy to see that $b|_{V_i}$ can be transformed to
    $b'|_{V_i}$ by toggling signs of vertices incident to edges $e \in E_i$
    as necessary. The lemma follows.
\end{proof} 
Finally, note that it is not hard to write down a Turing machine to construct
the wagon wheel hypergraph $\mcW(\mcI)$ and choose an $\mcI$-labelling $b$
for a given input presentation $\mcI$. Since Corollary \ref{C:collegial} is also
constructive, the proof of Theorem \ref{T:embedding} is also constructive as claimed. 

\section{Pictures for groups generated by involutions}\label{S:pictures}

In this section we give an overview of the main technical tool used in the
proof of Theorem \ref{T:invembedding}: pictures of groups. These pictures,
which are dual to the somewhat better known van Kampen diagrams, are a standard
tool in combinatorial group theory. The purpose of pictures is to encode
derivations of group identities from a set of starting relations; see
\cite{Sh07} for additional background.  Here we introduce a variant adapted to
groups generated by involutions.

\subsection{Pictures as planar graphs} 

By a curve, we shall mean the image of an piecewise (regular) real-analytic
function\footnote{Using real-analytic curves has some advantages in
streamlining the following definitions; for instance, we can use the fact that
two real-analytic curves intersect in at most finitely many points. However,
the arguments in the subsequent sections also work with other formal
definitions of pictures, such as the definition in \cite{Sh07} using smooth
curves and ``fat'' vertices.} $\gamma$ from a closed interval
$[a,b]$ (where $a < b$) to either the plane or the sphere. A curve $\gamma$ is
\emph{simple} if $\gamma(s) \neq \gamma(t)$ for all $a \leq s < t \leq b$,
except possibly when $s=a$ and $t=b$. The points $\gamma(a)$ and $\gamma(b)$
are called the endpoints of the curve. A curve has either one or two endpoints;
if $\gamma(a) =\gamma(b)$ then the curve is said to be \emph{closed}. A
connected region (in the plane or on the sphere) is \emph{simple} if its
boundary is a simple closed curve. 
\begin{defn}\label{D:picture}
    A \emph{picture} is a collection $(V,E,\mcD)$, where
    \begin{enumerate}[(a)]
        \item $\mcD$ is a closed simple region, 
        \item $V$ is a finite collection of points, called \emph{vertices}, in
            $\mcD$,
        \item $E$ is a finite collection of simple curves, called \emph{edges},
            in $\mcD$, and
        \item for all edges $e \in E$ and points $p$ of $e$, 
        \begin{enumerate}[(i)]
            \item if $e$ is not closed and $p$ is an endpoint of $e$, then
                either $p \in V$, or $p$ belongs to the boundary of $\mcD$ and
                is not the endpoint of any other edge; 
            \item if $e$ is closed and $p$ is an endpoint of $e$, 
                then $p$ does not belong to the boundary of $\mcD$; 
            \item if $p$ is not an endpoint of $e$, then $p \not\in V$, and $p$
                does not belong to any other edge or the boundary of $\mcD$. 
        \end{enumerate}
    \end{enumerate}
    If an edge $e$ contains a vertex $v$, then we say that $e$ and $v$ are
    \emph{incident}. If $e$ contains a point of the boundary of $\mcD$, then we
    say that $e$ is \emph{incident with the boundary}.  A picture is
    \emph{closed} if no edges are incident with the boundary of $\mcD$. The
    \emph{size} of a $G$-picture $\mcP$ is the number of vertices in $\mcP$. 
\end{defn}
According to this definition, a picture is a type of planar embedding of a
graph, albeit a graph where we can have multiple edges between vertices, loops 
at a vertex, and even closed loops which are not incident to any vertex.  From
this point of view, the boundary of $\mcD$ can be regarded as a special type of
vertex; if we think of the picture as drawn on a sphere, then this vertex would
naturally be drawn at infinity. An illustration of these two equivalent points
of view is shown in Figure \ref{F:specialvertex}. However, it is more
convenient not to include the boundary of $\mcD$ in the vertex set of a
picture, and we stick with the convention of treating the boundary separately.
In particular, the picture in Figure \ref{F:specialvertex} has size $7$.

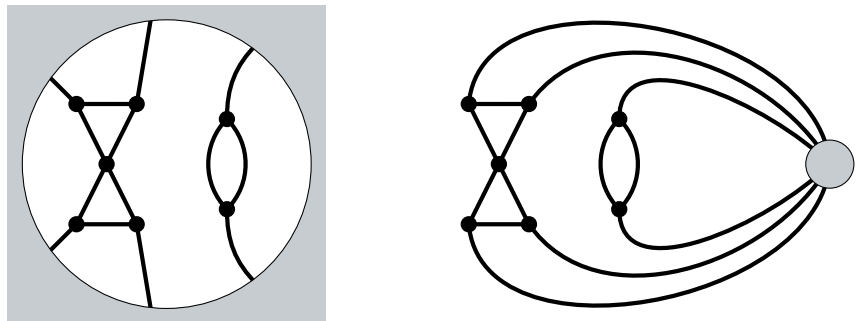
\begin{figure}
    \begin{tikzpicture}[baseline,scale=.4]
    
    \path [fill=fillgray] (-.3,-.3) rectangle (10.3,10.3);
    \draw [fill=white] (5,5) circle [radius=4.8];
    \begin{scope}[ultra thick]
        \clip (5,5) circle [radius=4.8];
        \draw (1,8) -- (2,7) coordinate(A) 
                    -- (3,5) coordinate (B) 
                    -- (2,3) coordinate(C) 
                    -- (4,3) coordinate(D) 
                    -- (B)
                    -- (4,7) coordinate(E) 
                    -- (A);
        \draw (E) -- (4.5,10);
        \draw (D) -- (4.5,0);
        \draw (C) -- (1,2);

        \draw (8,9) to [out=-135, in=90] (7,6.5) coordinate(F)
                    to [out=-135, in=135] (7,3.5) coordinate(G) 
                    to [out=-90, in=135] (8,1);
        \draw (F) to [out=-45, in=45] (G);

        \draw [fill] (A) circle [radius=.2];
        \draw [fill] (B) circle [radius=.2];
        \draw [fill] (C) circle [radius=.2];
        \draw [fill] (D) circle [radius=.2];
        \draw [fill] (E) circle [radius=.2];
        \draw [fill] (F) circle [radius=.2];
        \draw [fill] (G) circle [radius=.2];
    \end{scope}
\end{tikzpicture}
\qquad \qquad
\begin{tikzpicture}[baseline, scale=.4]
    \begin{scope}[ultra thick]
        \draw (2,7) coordinate(A) 
                    -- (3,5) coordinate (B) 
                    -- (2,3) coordinate(C) 
                    -- (4,3) coordinate(D) 
                    -- (B)
                    -- (4,7) coordinate(E) 
                    -- (A);

        \draw (7,6.5) coordinate(F) to [out=-135, in=135] (7,3.5) coordinate(G) 
                                    to [out=45, in=-45] (F);

        \draw [fill] (A) circle [radius=.2];
        \draw [fill] (B) circle [radius=.2];
        \draw [fill] (C) circle [radius=.2];
        \draw [fill] (D) circle [radius=.2];
        \draw [fill] (E) circle [radius=.2];
        \draw [fill] (F) circle [radius=.2];
        \draw [fill] (G) circle [radius=.2];

        \draw (A) to [out=85,in=95] (14,5) coordinate(X);
        \draw (C) to [out=-85,in=-95] (X);
        \draw (E) to [out=55,in=120] (X);
        \draw (D) to [out=-55,in=-120] (X);
        \draw (F) to [out=90,in=135] (X);
        \draw (G) to [out=-90,in=-135] (X);

    \end{scope}

    \path [fill=white] (14,5) circle [radius=.8];
    \draw [fill=fillgray] (14,5) circle [radius=.8];
\end{tikzpicture}
    \caption{A picture embedded in a disk (left) and on the plane with
             the exterior of the disk shrunk down to a special vertex at
             infinity (right).}
    \label{F:specialvertex}
\end{figure}
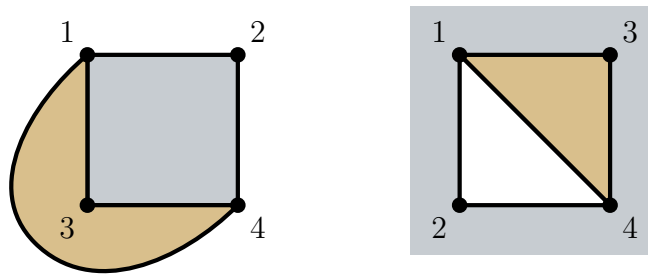
\begin{figure}
    \begin{tikzpicture}[baseline,scale=.5]
    \coordinate (1) at (0,4);
    \coordinate (2) at (4,4);
    \coordinate (3) at (0,0);
    \coordinate (4) at (4,0);
    \coordinate (fake) at (-1.4,-1);

    \begin{scope}[ultra thick]
        \draw [fill=fillgray] (1) -- (2) -- (4) -- (3) -- (1);
        \draw [fill=fillyellow] (4) to [out=-135,in=-45] (fake) to [out=135,in=-140] (1)
            -- (3) -- (4);

        \foreach \x in {1,...,4}
            \draw[fill] (\x) circle [radius=.15];

        \node [above left] at (1) {1};
        \node [above right] at (2) {2};
        \node [below left] at (3) {3};
        \node [below right] at (4) {4};
    \end{scope}
\end{tikzpicture}
\qquad \qquad
\begin{tikzpicture}[baseline,scale=.5]
    \coordinate (1) at (0,4);
    \coordinate (2) at (0,0);
    \coordinate (3) at (4,4);
    \coordinate (4) at (4,0);

    \path [fill=fillgray] (-1.3,-1.3) rectangle (5.3,5.3);
    
    \begin{scope}[ultra thick]
        \draw [fill=fillyellow] (1) -- (3) -- (4) -- (1);
        \draw [fill=white] (1) -- (2) -- (4) -- (1);

        \foreach \x in {1,...,4}
            \draw[fill] (\x) circle [radius=.15];

        \node [above left] at (1) {1};
        \node [below left] at (2) {2};
        \node [above right] at (3) {3};
        \node [below right] at (4) {4};
    \end{scope}
\end{tikzpicture}
    \caption{A closed picture embedded in the sphere seen (up to isotopy) with
             two different choices for the location of the point at infinity.
             Faces are distinguished by different colours.}
    \label{F:closedpic}
\end{figure}

There is one important exception where we want to forget the boundary of
$\mcD$, and that is when the picture is closed. If this happens, we often want
to think of the picture as embedded in the sphere, without the point at
infinity being marked. An example of a closed picture on the sphere, as seen
from two different positions, is given in Figure \ref{F:closedpic}. To handle
this case, we allow $\mcD$ to be the whole sphere, in which case the boundary
is empty. Also note that we consider two pictures equal if they differ up to
isotopy, either in the plane or on the sphere as appropriate. Such isotopies
are allowed to move the boundary of the simple region, as well as the location
of endpoints of edges on the boundary, as long as endpoints are not identified. 

If $\mcP$ is a picture in $\mcD_0$ and $\mcD$ is a closed simple subregion of
$\mcD_0$, then the portion of $\mcP$ contained in $\mcD$ can be interpreted as
a picture inside $\mcD$.  We do, however, have to make sure that boundary edges
of the picture inside $\mcD$ do not have a common endpoint. This leads to two
natural notions of the restriction of $\mcP$ to $\mcD$. 
\begin{defn}
    Let $\mcP$ be a picture in $\mcD_0$, and let $\mcD$ be a closed simple
    region in $\mcD_0$ with interior $\mcD^o$. Given $\eps > 0$, let
    $\mcD^{\eps}$ denote the $\eps$-relaxation of $\mcD$ (the set of points
    which are at most distance $\eps$ from $\mcD$), and let $\mcD^{-\eps}$
    denote the $\eps$-contraction (the set of points inside $\mcD$ which are at
    least distance $\eps$ from the boundary of $\mcD$). The boundary
    of $\mcD^{\pm \eps}$ will be a simple closed curve for small $\eps > 0$.%
    \footnote{In particular, the boundary curves will be piecewise
    real-analytic. When the boundary curve $\gamma$ of $\mcD$ is real-analytic,
    this follows from the tubular neighbourhood theorem: for small enough
    $\eps$, the boundary curves of $D^{\pm \eps}$ will be normal translates of
    $\gamma$. If $\gamma$ is piecewise real-analytic, then the situation
    becomes more complicated at the corners of $\gamma$. When two real-analytic segments of $\gamma$ meet
    at an acute angle, the boundary of $D^{\pm \eps}$ will consist of the
    normal translates of the two segments, cut off at their intersection point.
    For two segments meeting at an obtuse angle, the normal translates of each
    segment will be joined by a portion of the circle of radius $\eps$ centered at the corner.}

    Recall that a curve intersects the boundary of $\mcD$ \emph{transversally}
    if, in every small disk around the intersection point, there are points of
    the curve which lie both on the interior and the exterior of
    $\mcD$.\footnote{Transversality is typically defined in terms of tangent
    lines to the curves, but this more permissive definition is fine for our
    purposes.} We say that $\mcD$ is \emph{transverse} to $\mcP$ if every edge
    which intersects the boundary of $\mcD$ does so transversally, and the
    boundary of $\mcD$ does not contain any vertices of $\mcP$. 

    The \emph{restriction} of $\mcP$ to a transverse region $\mcD$ is the
    picture $\res(\mcP,\mcD)$ with vertex set $V(\mcP) \cap \mcD$, and whose
    edges are the closures of the connected components of $e \cap \mcD^o$, for
    $e \in E(\mcP)$. In other words, edges are cut off at the boundary, and
    edges which intersect the boundary at multiple points may be cut into
    multiple edges.\footnote{Since boundaries and edges are both real-anlytic,
    there will be only finitely many edges in this picture, as required by
    our definition.}

    For a general region $\mcD$, the contraction $\mcD^{-\eps}$ will be 
    transverse to $\mcP$ for small enough $\eps > 0$.
    The restrictions
    $\res(\mcP,\mcD^{-\eps})$ are thus well-defined, and can be identified via
    isotopy with pictures in $\mcD$. These pictures belong to a single isotopy
    class $\res(\mcP,\mcD)$, which we identify as the \emph{restriction} of
    $\mcP$ to $\mcD$.\footnote{This definition also uses the fact that edges
    and boundaries are real-analytic. For smooth curves, more care must be
    taken: the contraction may fail to be transverse to $\mcP$ for arbitrarily
    small $\eps$, and even if we only look at curves transverse to $\mcP$ in
    $\mcD^o$, there can be curves which are arbitrarily close to the boundary
    with arbitrarily many intersections with $\mcP$ (so the isotopy class may
    not be unique).}

    Similarly, the \emph{germ of $\mcD$ in $\mcP$} is the isotopy class
    $\germ(\mcP,\mcD)$ of $\res(\mcP,\mcD^{\eps})$ for small $\eps > 0$. 
\end{defn}
Unless otherwise noted, we assume that subregions are closed. If $\mcD$ is
transverse to $\mcP$, then $\res(\mcP,\mcD)$ and $\germ(\mcP,\mcD)$ agree. An
example of this type of restriction is shown in Figure \ref{F:restrict}. In
general, vertices in the boundary of $\mcD$ will not appear in
$\res(\mcP,\mcD)$, but are preserved, along with all their outgoing edges, in
$\germ(\mcP,\mcD)$.  One way to get a simple region is to take a region $\mcD$
enclosed by a simple cycle in $\mcP$. An example of this type, in which the
germ is different than the restriction, is given in Figure
\ref{F:cyclerestriction}.
\begin{figure}
    \begin{tikzpicture}[scale=.4]
    
    \begin{scope}
        \draw [fill=white] (5,5) circle [radius=4.8];
        \draw [dashed,rotate=30,fill=fillyellow] (B) ++(.5,-1.2) ellipse (3 and 2.1);
        \begin{scope}[ultra thick]
            \clip (5,5) circle [radius=4.8];
            \draw (1,8) -- (2,7) coordinate(A) 
                        -- (3,5) coordinate (B) 
                        -- (2,3) coordinate(C) 
                        -- (4,3) coordinate(D) 
                        -- (B)
                        -- (4,7) coordinate(E) 
                        -- (A);
            \draw (E) -- (4.5,10);
            \draw (D) -- (4.5,0);
            \draw (C) -- (1,2);

            \draw (8,9) to [out=-135, in=90] (7,6.5) coordinate(F)
                        to [out=-135, in=135] (7,3.5) coordinate(G) 
                        to [out=-90, in=135] (8,1);
            \draw (F) to [out=-45, in=45] (G);

            \draw [fill] (A) circle [radius=.2];
            \draw [fill] (B) circle [radius=.2];
            \draw [fill] (C) circle [radius=.2];
            \draw [fill] (D) circle [radius=.2];
            \draw [fill] (E) circle [radius=.2];
            \draw [fill] (F) circle [radius=.2];
            \draw [fill] (G) circle [radius=.2];
        \end{scope}
    \end{scope}
    \node at (12,5) {$\implies$};
 
   \begin{scope}[shift={++(14,0)}]
        \draw [clip] (5,5) circle [radius=4.8];
        \begin{scope}[ultra thick]
            \coordinate (1) at (2,4);
            \coordinate (2) at (5,3);
            \coordinate (3) at (3.5,7);

            \draw (1) -- (2) -- (3) -- cycle;
            \draw (1) -- (0,3);
            \draw (2) -- (7,0);
            \draw (3) -- (5,10);
            \draw (3) -- (2.5,10);

            \draw (10,7) to [out=180,in=90] (7,5) to [out=-90,in=180] (10,3);

            \foreach \x in {1,2,3}
                \draw [fill] (\x) circle [radius=.2];
        \end{scope}
        
    \end{scope}
\end{tikzpicture}
    \caption{Pictures can be restricted to a region homotopic to a disk and
             bounded by a simple closed curve which is transverse to the
            picture.}
    \label{F:restrict}
\end{figure}
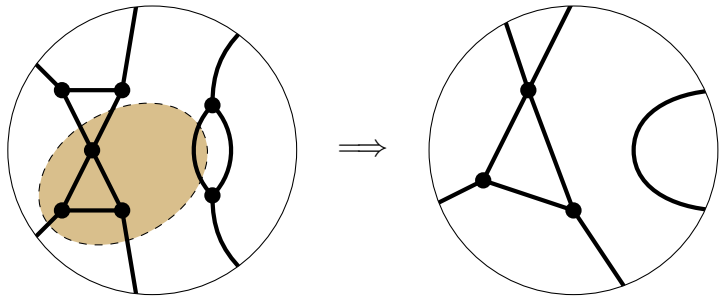
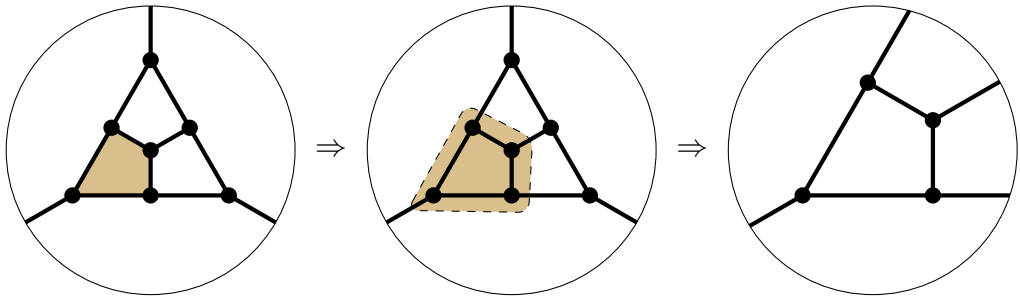
\begin{figure}
    \begin{tikzpicture}[baseline,scale=.4]

    \begin{scope}
        \draw [fill=white] (5,5) circle [radius=4.8];

        \coordinate (1) at (5,5);
        \coordinate (2) at ($(5,5) + (90:3)$);
        \coordinate (3) at ($(5,5) + (-150:3)$);
        \coordinate (4) at ($(5,5) + (-30:3)$);
        \coordinate (5) at ($(5,5) + {sin(30)}*(30:3)$);
        \coordinate (6) at ($(5,5) + {sin(30)}*(150:3)$);
        \coordinate (7) at ($(5,5) + {sin(30)}*(-90:3)$);

        \path[fill=fillyellow] (3) -- (6) -- (1) -- (7) -- cycle;

        \begin{scope}[ultra thick]
            \clip (5,5) circle [radius=4.8];

            \draw (2) -- (3) -- (4) -- cycle;
            \draw (1) -- (5);
            \draw (1) -- (6);
            \draw (1) -- (7);
            \draw (2) -- +(90:2);
            \draw (3) -- +(-150:2);
            \draw (4) -- +(-30:2);

            \foreach \x in {1,...,7}
                \draw [fill] (\x) circle [radius=.2];
        \end{scope}
    \end{scope}

    \begin{scope}[shift={++(12,0)}]
        \draw [fill=white] (5,5) circle [radius=4.8];

        \coordinate (1) at (5,5);
        \coordinate (2) at ($(5,5) + (90:3)$);
        \coordinate (3) at ($(5,5) + (-150:3)$);
        \coordinate (4) at ($(5,5) + (-30:3)$);
        \coordinate (5) at ($(5,5) + {sin(30)}*(30:3)$);
        \coordinate (6) at ($(5,5) + {sin(30)}*(150:3)$);
        \coordinate (7) at ($(5,5) + {sin(30)}*(-90:3)$);

        \draw[dashed,rounded corners,fill=fillyellow] 
                                    ($ (3)+(-150:1) $) 
                                        -- ($ (6)+(105:.8)$) 
                                        -- ($ (1)+(30:.8) $) 
                                        -- ($ (7)+(-45:.8) $) -- cycle;

        \begin{scope}[ultra thick]
            \clip (5,5) circle [radius=4.8];

            \draw (2) -- (3) -- (4) -- cycle;
            \draw (1) -- (5);
            \draw (1) -- (6);
            \draw (1) -- (7);
            \draw (2) -- +(90:2);
            \draw (3) -- +(-150:2);
            \draw (4) -- +(-30:2);

            \foreach \x in {1,...,7}
                \draw [fill] (\x) circle [radius=.2];
        \end{scope}
    \end{scope}

    \begin{scope}[shift={++(24,0)}]
        \draw [fill=white] (5,5) circle [radius=4.8];

        \coordinate (1) at (7,6);
        \coordinate (3) at ($(1) + (-150:5)$);
        \coordinate (6) at ($(1) + {sin(30)}*(150:5)$);
        \coordinate (7) at ($(1) + {sin(30)}*(-90:5)$);

        \begin{scope}[ultra thick]
            \clip (5,5) circle [radius=4.8];

            \draw (6) -- (3) -- (7) -- (1) -- cycle;
            \draw (1) -- +(30:3);
            \draw (6) -- +(60:3);
            \draw (3) -- +(-150:3);
            \draw (7) -- +(0:3);

            \foreach \x in {1,3,6,7}
                \draw [fill] (\x) circle [radius=.2];
        \end{scope}
    \end{scope}

    \node at (11,5) {$\Rightarrow$};
    \node at (23,5) {$\Rightarrow$};
    
\end{tikzpicture}
    \caption{The \emph{germ} of the region enclosed by a simple cycle is
            computed by first taking an $\eps$-relaxation of the region. In
            this case the cycle is facial, so the \emph{restriction} would be empty.}
    \label{F:cyclerestriction}
\end{figure}
\begin{defn}
    Let $\mcP$ be a picture in $\mcD_0$. A \emph{simple cycle} in $\mcP$ is
    a collection of edges whose union is a simple closed curve. 

    A \emph{closed loop} is an edge of $\mcP$ which is not incident to any
    vertex or to the boundary, and thus forms a simple cycle by itself. 

    A \emph{face} of $\mcP$ is an open connected region $\mcD$ of $\mcD_0$
    which does not contain any points of $\mcP$, and such that the boundary of
    $\mcD$ is a union of points of $\mcP$ and points in the boundary of
    $\mcD_0$.
    An \emph{outer face} is a face whose boundary contains points of the
    boundary of $\mcD_0$. 

    A simple cycle is \emph{facial} if it is the boundary of a face.
\end{defn}
Every simple cycle in the disk bounds a unique simple region (the interior of
the cycle), while a simple cycle on the sphere bounds two simple regions. A
face does not have to be simple, but a facial cycle always bounds a simple face
by definition.

\subsection{Groups and labellings of pictures}

Recall that $r^+$ refers to the even part of a relation $r \in \mcF_2(S) \times
\Z_2$ (see Definition \ref{D:evenodd}).
\begin{defn}\label{D:invpictures}
    Let $G = \Inv\langle S : R\rangle$. A \emph{$G$-picture} is a picture $\mcP$ with every
    vertex $v$ labelled by a relation $r(v) \in R$ and every edge $e$ labelled
    by a generator $s(e) \in S$, such that if $e_1,\ldots,e_n$ is the sequence
    of edges incident to $v$, read in counter-clockwise order with multiplicity
    from some starting point, then $s(e_1) s(e_2) \cdots s(e_n) \in
    \{r(v)^+\}^{sym}$.

    The \emph{boundary} of $\mcP$ is the cyclic word $\bd(\mcP) = s(e_1) \cdots
    s(e_n)$ over $S$, where $e_1,\ldots,e_n$ is the list of edges incident with
    the boundary, read in counter-clockwise order around the boundary of the
    disc, with multiplicity. If $\mcP$ is closed then we say that
    $\bd(\mcP)=1$, the empty word.

    Two pictures $\mcP_1$ and $\mcP_2$ are \emph{equivalent} if $\bd(\mcP_1) =
    \bd(\mcP_2)$. 

    The \emph{sign} of a picture $\mcP$ is $\sign(\mcP) = |\{v \in V(\mcP) :
    r(v) \text{ is odd}\}| \mod 2$. 

    If $\mcD$ is a simple region and $\mcP$ is a $G$-picture, then
    $\res(\mcP,\mcD)$ and $\germ(\mcP,\mcD)$ both inherit the structure of a
    $G$-picture from $\mcP$ via restricting the labelling functions. 
\end{defn}
Typically pictures use directed edges to represent inverses of generators, but
this is not necessary for groups generated by involutions. As previously
mentioned, the point of pictures is that they capture relations in the group,
in the following sense:
\begin{prop}[van Kampen lemma]\label{P:vankampen1}
    Let $G = \Inv\langle S : R\rangle$, let $r$ be a word over $S$, and let $a
    \in \Z_2$.  Then $r = J^a$ in $G$ if and only if there is a $G$-picture $\mcP$
    with $\bd(\mcP) = r$ and $\sign(\mcP) = a$. 
\end{prop}
The original version of the van Kampen lemma goes back to \cite{Ka33}. We
give a proof of this version for the convenience of the reader. 
\begin{proof}[Proof of Proposition \ref{P:vankampen1}]
    Given any word $w = s_1 \cdots s_n$ over $S$, we abuse notation slightly
    and let $w^{-1}$ refer to the word $s_n \cdots s_1$. Because $J$ is
    central, it is not hard to see that $r = J^a$ in $G$ if and only if there
    is a sequence $r_0,\ldots,r_n$ of (not necessarily reduced) words over $S$
    such that
    \begin{itemize}
        \item for every $i=1,\ldots,n$, the word $r_i$ can be constructed from
            $r_{i-1}$ by either replacing a subword $w_0$ with $w_1$, where
            $w_0 w_1^{-1}$ is the even part $w^+$ of some relation $w \in R^{sym}$,
            or inserting or deleting $s^2$ for some $s \in S$; 
        \item $a$ is (the parity of) the number of replacements $w_0 \arr w_1$
            in the point above, in which the corresponding relation $w \in R^{sym}$ is odd; and
        \item $r_0 = r$ and $r_n = 1$.
    \end{itemize}
    From any such sequence, it is not hard to construct a $G$-picture $\mcP$
    with $\bd(\mcP) = r$ and $\sign(\mcP)=a$. Indeed, for each $i=1,\ldots,n$
    we can build a $G$-picture $\mcP_i$ in a rectangle with $\bd(\mcP_i) =
    r_{i-1} r_i^{-1}$, where the boundary edges labelled by $r_{i-1}$ are
    connected to the top of the rectangle, and the boundary edges labelled by
    $r_i^{-1}$ are connected to the bottom of the rectangle, using the following
    recipe: 
    \begin{itemize}
        \item If $r_{i-1} = x w_0 y$ and $r_{i} = x w_1 y$, where $x,y$
            are words over $S$ and $w_0 w_1^{-1} \in \{w^+\}^{sym}$ for some $w
            \in R$, then $\mcP_i$ is a $G$-picture with a single vertex
            labelled by $w$. This vertex is connected to the top half of the rectangle
            by edges labelled (from right to left) by $w_0$ and to the bottom
            half by edges labelled by $w_1$. To the right and left of the vertex
            we add edges connecting the top to the bottom of the rectangle, labelled
            by $x$ and $y$ respectively.
        \item If $r_{i-1} = xy$ and $r_i = x s^2 y$ for some $s \in S$, then
            $\mcP_i$ is a picture with edges labelled (from right to left) by
            $x$ and $y$ connecting the top of the rectangle to the bottom, and
            an edge labelled by $s$ with both endpoints incident to the bottom
            of the rectangle.
        \item If $r_{i-1} = x s^2 y$ and $r_i = xy$ for some $s \in S$,
            then $\mcP_i$ is a picture with edges labelled (from right to left)
            by $x$ and $y$ connecting the top of the rectangle to the bottom,
            and an edge labelled by $s$ with both endpoints incident to the
            top of the rectangle.
    \end{itemize}
    Note that $\mcP_n$ has no edges incident to the bottom of the rectangle.
    Putting these pictures together from top to bottom gives a picture $\mcP$
    with $\bd(\mcP)=r$ and $\sign(\mcP)=a$ as desired. 

    \begin{figure}[tbhp]
    \begin{tikzpicture}[scale=.9, auto,every node/.style={scale=.8}]

    \draw [clip] (0,0) -- (0,7) -- (10,7) -- (10,0) -- (0,0);

    \coordinate (t1) at (3,7);
    \coordinate (t2) at (6,7);
    \coordinate (t3) at (9,7);
    \coordinate (p1) at (1.5,6.5);
    \coordinate (1) at (4.5,5.5);
    \coordinate (p2) at (6.5,4.5);
    \coordinate (p3) at (1.5,3.5);
    \coordinate (2) at (8,1.5);
    \coordinate (3) at (5,2.5);
    \coordinate (p4) at (6.5,0.5);
    \begin{scope}[ultra thick]
        \draw (t1) to node[swap,pos=.3] {$c$} (1);
        \draw (t2) to node[pos=.3] {$b$} (1);
        \draw (1) to node[pos=.55,swap] {$f$} (3);
        \draw (3) to [out=70,in=180] (p2) to [out=0,in=110] node [pos=0] {$d$} (2);
        \draw (t3) to node[pos=.6] {$a$} (2);
        \draw (2) to [out=-110,in=0] (p4) to [out=180,in=-70] node [pos=0] {$g$} (3);
        \draw (p1) to [out=-30,in=30] (p3) to [out=150,in=-150] node[pos=.58] {$a$} (p1);
        \foreach \x in {1,2,3}
            \draw [fill] (\x) circle [radius=.1];
    \end{scope}
    \draw[dashed] (0,6) -- (10,6);
    \draw[dashed] (0,5) -- (10,5);
    \draw[dashed] (0,4) -- (10,4);
    \draw[dashed] (0,3) -- (10,3);
    \draw[dashed] (0,2) -- (10,2);
    \draw[dashed] (0,1) -- (10,1);
\end{tikzpicture}
    \caption{A $G$-picture $\mcP$ with $\bd(\mcP)=abc$ and $\sign(\mcP)=1$
        for $G = \Inv\langle a,b,c,d,f,g : adg,bcf,Jdfg \rangle$}
    \label{F:vankampen}
    \end{figure}

    An example of this process with 
    \begin{equation*}
        G = \Inv\langle a,b,c,d,f,g : adg=1,bcf=1,dfg=J\rangle
    \end{equation*}
    is given in Figure \ref{F:vankampen}. The figure
    shows a picture $\mcP$ with $\bd(\mcP)=abc$ and $\sign(\mcP)=1$,
    corresponding to the identity $abc = J$. The picture is constructed by
    concatenating pictures $\mcP_1,\ldots,\mcP_7$ (demarcated by the dashed
    lines) corresponding to the derivation $abc = abca^2 = afa^2 = ad^2fa^2 =
    ad^2f = Jad g = Jg^2 = J$.  As can be seen from the example, the algorithm
    outlined above can create pictures with closed loops, but these are allowed
    in our definition of $G$-pictures.

    Conversely, suppose we are given a $G$-picture with $\bd(\mcP)=r$ and
    $\sign(\mcP)=a$. After isotopy, we can assume that $\mcP$ is a picture in a
    rectangle in the plane, with all boundary edges incident to the top of the
    rectangle. We can also assume that every edge intersects every horizontal
    line through the rectangle in a finite number of points, and intersects
    only finitely many horizontal lines non-transversely. 
    Say that a point on the interior of the diagram is a \emph{critical point}
    if it is a vertex of $\mcP$, or a point on an edge which intersects a
    horizontal line non-transversely.  By moving critical
    points up or down, we can assume that every horizontal line through the
    rectangle hits at most one critical point. We can then cut $\mcP$ along
    horizontal lines into a sequence of pictures $\mcP_1,\ldots,\mcP_n$ each of
    which contains a single critical point. As in the above argument, this
    sequence can be turned into a derivation that $r = J^a$ in $G$.
\end{proof}

\begin{defn}\label{D:surgery}
    Let $\mcD$ be a simple region of a $G$-picture $\mcP_0$, and let $\mcP'=
    \res(\mcP,\mcD)$ or $\germ(\mcP, \mcD)$. If $\mcP'$ is equivalent to a
    picture $\mcP''$, then we can cut out $\mcP'$ and glue in $\mcP''$ in its
    place to get a new picture $\mcP_1$. We refer to the process $\mcP_0 \surg
    \mcP_1$ as \emph{surgery}.

    If $\mcP'$ has size zero,\footnote{Size zero pictures do not have vertices,
    but they can still have edges.} then we call $\mcP_0 \surg \mcP_1$ a
    \emph{null surgery}. We say that $\mcP_0$ and $\mcP_1$ are \emph{equivalent
    via null surgeries} if there is a sequence of null surgeries transforming
    $\mcP_0$ to $\mcP_1$. 
\end{defn}
\begin{figure}[tbhp]
    \begin{tikzpicture}[auto,scale=.4]

    \begin{scope}[shift={++(0,12)},every node/.style={scale=.7}]
        \draw (5,5) circle [radius=4.8];
        \coordinate (t1) at (0,8);
        \coordinate (t2) at (2,10);
        \coordinate (t3) at (4,10);
        \coordinate (t4) at (6,10);
        \coordinate (t5) at (9,10);
        \coordinate (b1) at (2,0);
        \coordinate (b2) at (5,0);
        \coordinate (b3) at (8,0);
        \coordinate (1) at (7,7);
        \coordinate (2) at (5,5);
        \coordinate (3) at (3,3);

        \draw [dashed,rotate=45,fill=fillyellow] (2) ++(1.5,.5) ellipse (2.5 and 1.5);
        \begin{scope}[ultra thick]

            \clip (5,5) circle [radius=4.8];
            \draw (t1) to node [pos=.45,swap] {$s_1$} (3);
            \draw (t2) to node [pos=.3,swap] {$s_2$} (3);
            \draw (t3) to [bend right=10] node [pos=.15,swap] {$s_3$} (2);
            \draw (t4) to node [pos=.15,swap] {$s_3$} (1);
            \draw (t5) to node [pos=.6] {$s_1$} (1);
            \draw (1) to [bend right=50] node [swap,pos=.25] {$s_1$} (2);
            \draw (1) to [bend left=50] node {$s_3$} (2);
            \draw (2) to node [pos=.3] {$s_1$} (3);
            \draw (3) to node [swap] {$s_2$} (b1);
            \draw (3) to node [pos=.8] {$s_1$} (b2);
            \draw (3) to [bend left] node [pos=.7] {$s_2$} (b3);
            \foreach \x in {1,2,3}
                \draw [fill] (\x) circle [radius=.2];
        \end{scope}
    \end{scope}
    \begin{scope}[shift={++(14,12)},every node/.style={scale=.7}]
        \draw (5,5) circle [radius=4.8];
        \coordinate (t1) at (0,8);
        \coordinate (t2) at (2,10);
        \coordinate (t3) at (4,10);
        \coordinate (t4) at (6,10);
        \coordinate (t5) at (9,10);
        \coordinate (b1) at (2,0);
        \coordinate (b2) at (5,0);
        \coordinate (b3) at (8,0);
        \coordinate (3) at (3,3);

        \begin{scope}[ultra thick]

            \clip (5,5) circle [radius=4.8];
            \draw (t1) to node [pos=.45,swap] {$s_1$} (3);
            \draw (t2) to node [pos=.3,swap] {$s_2$} (3);
            \draw (t3) to [out=-90,in=180] (5,8) node [below] {$s_3$} to [out=0,in=-90] (t4);
            \draw (t5) to node [pos=.25] {$s_1$} (3);
            \draw (3) to node [swap] {$s_2$} (b1);
            \draw (3) to node [pos=.8] {$s_1$} (b2);
            \draw (3) to [bend left] node [pos=.7] {$s_2$} (b3);
            \foreach \x in {1,2,3}
                \draw [fill] (\x) circle [radius=.2];
        \end{scope}
    \end{scope}
    \begin{scope}[shift={++(7,0)},every node/.style={scale=.7}]
        \draw (5,5) circle [radius=4.8];
        \coordinate (t1) at (0,8);
        \coordinate (t2) at (2,10);
        \coordinate (t3) at (4,10);
        \coordinate (t4) at (6,10);
        \coordinate (t5) at (9,10);
        \coordinate (b1) at (2,0);
        \coordinate (b2) at (5,0);
        \coordinate (b3) at (8,0);
        \coordinate (1) at (7,7);
        \coordinate (2) at (5,5);
        \coordinate (3) at (3,3);

        \begin{scope}[ultra thick]

            \clip (5,5) circle [radius=4.8];
            \draw (t1) to node [pos=.45,swap] {$s_1$} (3);
            \draw (t2) to node [pos=.3,swap] {$s_2$} (3);
            \draw (t3) to [bend right=10] node [pos=.15,swap] {$s_3$} (2);
            \draw (t4) to node [pos=.15,swap] {$s_3$} (1);
            \draw (t5) to node [pos=.6] {$s_1$} (1);
            \draw (1) to [bend right=50] node [swap,pos=.25] {$s_1$} (2);
            \draw (1) to [bend left=50] (2);
            \draw (2) to node [pos=.5] {$s_1$} (3);
            \draw (3) to node [swap] {$s_2$} (b1);
            \draw (3) to node [pos=.8] {$s_1$} (b2);
            \draw (3) to [bend left] node [pos=.7] {$s_2$} (b3);
            \foreach \x in {1,2,3}
                \draw [fill] (\x) circle [radius=.2];
        \end{scope}
        \draw [dashed,rotate=45,fill=white] (2) ++(1.5,.5) ellipse (2.5 and 1.5);
    \end{scope}

    \begin{scope}[shift={+(17,0)},scale=.5,every node/.style={scale=.5}]
        \draw (5,5) circle [radius=4.8];
        \coordinate (t1) at (0,6.5);
        \coordinate (t2) at (4,10);
        \coordinate (t3) at (10,10);
        \coordinate (b1) at (0,0);
        \coordinate (1) at (6.5,6.5);
        \coordinate (2) at (3.5,3.5);

        \begin{scope}[ultra thick]
            \clip (5,5) circle [radius=4.8];
            \draw (t1) to node [pos=.2] {$s_3$} (2);
            \draw (t2) to node [swap,pos=.2] {$s_3$} (1);
            \draw (t3) to node [pos=.65] {$s_1$} (1);
            \draw (1) to [bend right=50] node [swap] {$s_1$} (2);
            \draw (1) to [bend left=50] node {$s_3$} (2);
            \draw (2) to node [pos=.4] {$s_1$} (b1);
            \foreach \x in {1,2}
                \draw [fill] (\x) circle [radius=.2];
        \end{scope}
    \end{scope}

    \begin{scope}[shift={+(17,5)},scale=.5,every node/.style={scale=.5}]
        \draw (5,5) circle [radius=4.8];
        \coordinate (t1) at (0,6.5);
        \coordinate (t2) at (4,10);
        \coordinate (t3) at (10,10);
        \coordinate (b1) at (0,0);

        \begin{scope}[ultra thick]
            \clip (5,5) circle [radius=4.8];
            \draw (t1) to [out=-55,in=-45] node [swap] {$s_3$} (t2);
            \draw (t3) to node {$s_1$} (b1);
        \end{scope}
    \end{scope}

    \draw[gray,->] (12.8,4.5) to [out=-90,in=180] (16.8,3);
    \draw[gray,<-] (14.2,5.5) to [out=-35,in=-135](17.8,5.6);
    \node at (12,17) {$\surg$};
    \node[rotate=-55] at (8,11) {$\implies$};
    \node[rotate=55] at (16,11) {$\implies$};
\end{tikzpicture}
    \caption{Surgery for an $S_4 \times \Z_2$-picture with boundary relation $s_1 s_2 s_3
             s_3 s_1 = s_2 s_1 s_2$.}
    \label{F:surgery}
\end{figure}

\begin{example}
    Consider the Coxeter group $S_4 \times \Z_2$ over $\Z_2$. This group has
    presentation
    \begin{equation*}
        G = \Inv\langle s_1,s_2,s_3 :\ s_1 s_3 = s_3 s_1, s_1 s_2 s_1 = s_2 s_1
            s_2, s_2 s_3 s_2 = s_3 s_2 s_3\rangle, 
    \end{equation*}
    and in particular is presented by involutions. An example of a surgery for this
    group is shown in Figure \ref{F:surgery}.
\end{example}

\section{Pictures over solution groups and hypergraphs}\label{S:hyperpics}

Let $\mcH$ be a hypergraph with vertex labelling function $b : V(\mcH) \arr
\Z_2$. The solution group $\Gamma(\mcH,b)$ is finitely presented by involutions
over $\Z_2$; to talk about $\Gamma$-pictures, we just need to pick a
presentation of $\Gamma(\mcH,b)$. One candidate is the presentation from
Definition \ref{D:solutiongroup}. This presentation contains two types of relations:
\emph{linear relations} of the form $\prod x_e = J^a$, and \emph{commuting
relations} of the form $x_e x_{e'} = x_{e'} x_e$. However, it will be more
convenient to use a presentation without commuting relations:
\begin{defn}\label{D:solutionpres}
    As a group presented by involutions over $\Z_2$, we let $\Gamma(\mcH,b) =
    \Inv\langle S, R\rangle$, where $S = \{x_e : e \in E(\mcH)\}$ and
    \begin{multline*}
        R = \{ J^{b_v} x_{e_1} \cdots x_{e_n} : \text{ all } v \in V  \text{ and all
            orderings } e_1,\ldots,e_n \\  \text{ of the edges incident to } v
            \text{ listed with multiplicity} \}
    \end{multline*}
\end{defn}
\begin{example}\label{Ex:commuting}
    Consider the solution group $\Gamma(\mcH,b)$ for the simple hypergraph
    $\mcH$ with a single vertex incident to four edges. 
    In the presentation in Definition~\ref{D:solutiongroup}, $\Gamma(A,b)$ is
    generated by $x_1,x_2,x_3,x_4,J$ subject to relations $J^2=1$, $x_i^2 =
    [x_i,J]=1$ for $1 \leq i \leq 4$, commuting relations $[x_i,x_j]=1$ for $1
    \leq i < j \leq 4$, and the single linear relation $x_1 x_2 x_3 x_4 = J^b$.
    In the presentation in Definition~\ref{D:solutionpres}, the commuting relations
    are omitted, and the single linear relation is replaced by the $12$ relations
    $x_{\sigma(1)} x_{\sigma(2)} x_{\sigma(3)} x_{\sigma(4)} = J^b$, where
    $\sigma$ is a permutation of $\{1,2,3,4\}$ (the other relations remain the same). 
\end{example}
In general, all the linear relations in Definition \ref{D:solutionpres} can be
recovered from a single linear relation and the commuting relations, while a
commuting relation can be recovered from two linear relations, as shown in Figure
\ref{F:commuting} (for the group in Example \ref{Ex:commuting}). Consequently,
the presentations of $\Gamma(\mcH,b)$ in \ref{D:solutiongroup} and
\ref{D:solutionpres} are equivalent. We will always use the presentation in
Definition \ref{D:solutionpres} when working with $\Gamma$-pictures. 

\begin{figure}
    \begin{tikzpicture}[auto, scale=.4]

    \begin{scope}
        \draw [clip] (5,5) circle [radius=4.8];
        \coordinate (t1) at (2,10); \coordinate (t2) at (8,10);
        \coordinate (b1) at (2,0); \coordinate (b2) at (8,0);
        \coordinate (p) at (5,5);

        \begin{scope}[ultra thick]
            \draw (t1) to node [pos=.4,swap] {$x_1$} (p);
            \draw (t2) to node [pos=.4] {$x_2$} (p);
            \draw (b1) to node [pos=.4] {$x_2$} (p);
            \draw (b2) to node [pos=.4,swap] {$x_1$} (p);
        \end{scope}
        \draw[fill] (p) circle [radius=.2];
    \end{scope}
    \node at (12,5) {$\Rightarrow$};

    \begin{scope}[shift={+(14,0)}]
        \draw [clip] (5,5) circle [radius=4.8];
        \coordinate (t1) at (2,10); \coordinate (t2) at (8,10);
        \coordinate (b1) at (2,0); \coordinate (b2) at (8,0);
        \coordinate (p1) at (5,7);
        \coordinate (p2) at (5,3);

        \begin{scope}[ultra thick]
            \draw (t1) to node [swap] {$x_1$} (p1);
            \draw (t2) to node {$x_2$} (p1);
            \draw (p1) to [out=-145,in=145] node [swap] {$x_3$} (p2);
            \draw (p1) to [out=-35,in=35] node {$x_4$} (p2);
            \draw (b1) to node {$x_2$} (p2);
            \draw (b2) to node [swap] {$x_1$} (p2);
        \end{scope}
        \draw[fill] (p1) circle [radius=.2];
        \draw[fill] (p2) circle [radius=.2];
    \end{scope}
\end{tikzpicture}
    \caption{In a solution group with linear relation $x_1 x_2 x_3 x_4 = J^a$,
             the commuting relation $x_1 x_2 = x_2 x_1$ can be represented
             pictorially in two different ways, depending on whether we use the
             presentation from Definition \ref{D:solutiongroup} (on the left),
             or the presentation from Definition \ref{D:solutionpres}
             (on the right).}
    \label{F:commuting}
\end{figure}
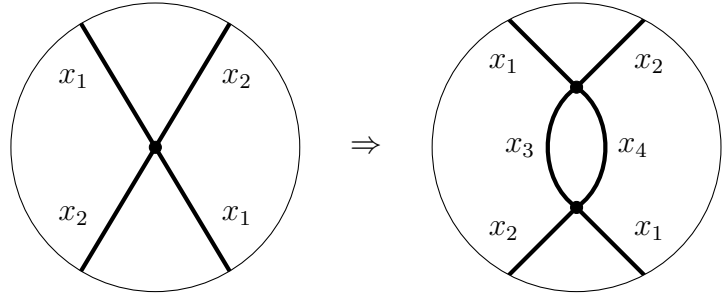

\begin{defn}\label{D:Hpicture}
    Let $\mcH$ be a hypergraph. An \emph{$\mcH$-picture} is a triple $(\mcP,
    h_V, h_E)$, where $\mcP$ is a picture, and $h_V$ and $h_E$ are
    \emph{labelling functions} $V(\mcP) \arr V(\mcH)$ and
    $E(\mcP) \arr E(\mcH)$ respectively, such that for all $v \in V(\mcP)$ and $e' \in
    E(\mcH)$, if we list the edges $e_1, \ldots, e_n$ of $\mcP$ incident to $v$
    with multiplicity then $A_{h_V(v)e'} = |\{1 \leq i
    \leq n : h_E(e_i) = e'\}|$. For convenience, we usually write $\mcP$
    for the triple $(\mcP,h_V,h_E)$, and $h$ for the labelling functions $h_V$ and $h_E$. 

    The \emph{boundary} of an $\mcH$-picture $\mcP$ is the cyclic word
    $\bd(\mcP) = h(e_1) \cdots h(e_n)$ over $E(\mcH)$,
    where as before $e_1,\ldots,e_n$ is the sequence of edges incident with the
    boundary, read counter-clockwise with multiplicity.  The \emph{character}
    of $\mcP$ is the vector $\ch(\mcP) \in \Z_2^{V(\mcH)}$ with 
    $\ch(\mcP)_{v} = |h^{-1}_V(v)| \mod 2$.
\end{defn}
It is easy to see that there is a one-to-one correspondence between
$\Gamma(\mcH,b)$-pictures and $\mcH$-pictures. If $\mcP$ is an $\mcH$-picture,
then the sign of the corresponding $\Gamma$-picture is the standard dot product
$\ch(\mcP) \cdot b$. The leads to the following restatement of the van Kampen
lemma for $\mcH$-pictures.
\begin{prop}[van Kampen lemma]\label{P:vankampen2}
    Let $\Gamma(\mcH,b)$ be a solution group. Then $x_{e_1} \cdots x_{e_n} =
    J^a$ in $\Gamma(\mcH,b)$ if and only if there is an $\mcH$-picture $\mcP$ with
    $\bd(\mcP) = e_1 \cdots e_n$ and $\ch(\mcP) \cdot\, b = a$. 
\end{prop}
\begin{example}
    Consider the solution group for the linear system
    \begin{align*}
        (1) \qquad x + y + z & = 1\\
        (2) \qquad x + y + z & = 0
    \end{align*}
    The underlying hypergraph $\mcH$ of this system is shown below.
    
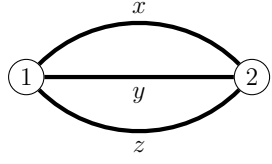
\begin{equation*}
        \begin{tikzpicture}[auto,ultra thick,vertex/.style={circle,draw,thin,inner sep=2.5},every node/.style={scale=.8}]
    \node[vertex] (r1) at (0,0) {$1$};
    \node[vertex] (r2) at (3,0) {$2$}
        edge [bend right=45] node [swap] {$x$} (r1)
        edge node {$y$} (r1)
        edge [bend left=45] node {$z$} (r1);
\end{tikzpicture}

    \end{equation*}
    This drawing of the hypergraph is also a closed $\mcH$-picture $\mcP$ with
    character $\ch(\mcP) = (1,1)$. Since $b = (1,0)$ and $\ch(\mcP) \cdot\, b
    = 1$, van Kampen's lemma tells us that $J = 1$.
\end{example}
\begin{rmk}\label{R:columnspace}
    Given a hypergraph $\mcH$ with incidence matrix $A$, let $X \subset \Z_2^V$
    be the set of vertex labellings $b$ such that $J \neq 1$ in
    $\Gamma(\mcH,b)$, and let $Y \subset \Z_2^V$ be the set of characters
    $\ch(\mcP)$ of $\mcH$-pictures $\mcP$. It is not hard to see that $X$ and
    $Y$ are subspaces of $\Z_2^V$. Proposition \ref{P:vankampen2} states that
    $Y$ is the orthogonal subspace to $X$ with respect to the standard
    bilinear product. By Theorem \ref{T:games}, $X$ can be regarded as
    a quantum-information-theoretic analogue of the columnspace of $A$. From
    this point of view, $Y$ is then an analogue of the left nullspace of $A$. 
\end{rmk}
If edge $e$ is incident to vertex $v$ in an $\mcH$-picture $\mcP$, then 
$h(e)$ will be incident to $h(v)$ in $\mcH$. Consequently, the labelling
function $h$ in Definition \ref{D:Hpicture} can be seen as a type of
weak hypergraph homomorphism. If $\mcP$ is closed, and $\mcH$ and $\mcP$ are
simple loopless graphs, then $h$ will be an actual graph homomorphism if and
only if $h(v) \neq h(v')$ for all adjacent vertices $v$ and $v'$ in $\mcP$.
Furthermore, if this happens then $h$ must be a planar graph cover. This
motivates the following definition:
\begin{defn}\label{D:cover}
    Let $\mcH$ be a hypergraph. A closed $\mcH$-picture $\mcP$ is a
    \emph{cover} of $\mcH$ if every edge of $\mcP$ is incident with two
    distinct vertices $v$ and $v'$ such that $h(v) \neq h(v')$. 
\end{defn}
Size, equivalence, restriction to a region, surgery, and null surgery are all
defined for $\mcH$-pictures via the correspondence with $\Gamma$-pictures.  For
instance, the size of an $\mcH$-picture is simply the number of vertices in the
picture. 
\begin{defn}\label{D:minimal}
    Let $\mcH$ be a hypergraph, and let $b : V(\mcH) \arr \Z_2$ be a
    vertex-labelling function. Two $\mcH$-pictures $\mcP_0$ and $\mcP_1$ are
    \emph{$b$-equivalent} (resp. \emph{character-equivalent}) if $\bd(\mcP_0) =
    \bd(\mcP_1)$ and $\ch(\mcP_0) \cdot b = \ch(\mcP_1) \cdot b$ (resp.
    $\ch(\mcP_0) = \ch(\mcP_1)$).

    An $\mcH$-picture $\mcP$ is \emph{$b$-minimal} (resp.
    \emph{character-minimal}) if $\mcP$ has minimum size among all
    $b$-equivalent (resp. character-equivalent) pictures. 
\end{defn}
Two pictures are character-equivalent if and only if they are $b$-equivalent for
all vertex-labelling functions $b$. Thus a $b$-minimal picture is also
character-minimal. 

If $\mcP_0 \surg \mcP_1$ is a surgery in which a region $\mcP'$ is replaced by
an equivalent region $\mcP''$, then $\bd(\mcP_0) = \bd(\mcP_1)$, and
$\ch(\mcP_1) = \ch(\mcP_0) - \ch(\mcP') + \ch(\mcP'')$. So if $\mcP'$ and
$\mcP''$ are $b$-equivalent (resp. character equivalent), then $\mcP_0$ and
$\mcP_1$ will be $b$-equivalent (resp. character-equivalent). 
We conclude that if $\mcP$ is $b$-minimal (resp. character-minimal) then
$\res(\mcP,\mcD)$ and $\germ(\mcP,\mcD)$ will be $b$-minimal (resp.
character-minimal) for all simple regions $\mcD$ (as otherwise we could
replace $\res(\mcP,\mcD)$ or $\germ(\mcP,\mcD)$ with smaller pictures).

We can also remove closed loops without changing the character-equivalence
class:
\begin{lemma}\label{L:closedloops}
    Suppose $\mcP$ is a picture, and let $\mcP'$ be the same picture but with
    all closed loops deleted. Then $\mcP$ and $\mcP'$ are character-equivalent
    and have the same size.
\end{lemma}

\section{A category of hypergraphs}\label{S:category}

\begin{defn}
    Let $\mcH = (V,E,A)$ be a hypergraph. A \emph{subhypergraph} of $\mcH$ is a
    hypergraph $\mcH' = (V',E',A')$ with $V' \subset V$, $E' \subset E$, and
    $A'_{ve} = A_{ve}$ for all $v \in V'$ and $e \in E'$.
\end{defn}
In other words, a subhypergraph is simply a subset of $V(\mcH) \cup E(\mcH)$.
Although this definition is substantially less restrictive than other notions
of subhypergraphs in the literature, it is natural in the context of
hypergraphs with isolated edges. 

\begin{defn}
    If $\mcH'$ is a subhypergraph of $\mcH$, then the \emph{neighbourhood}
    $\mcN(\mcH')$ of $\mcH'$ is the subhypergraph with $V(\mcN(\mcH')) = V(\mcH')$,
    and 
    \begin{equation*}
        E(\mcN(\mcH')) = E(\mcH') \cup \{e \in E(\mcH) : e \text{ is incident
            in } \mcH \text{ to some vertex } v \in V(\mcH') \}.
    \end{equation*}
    We say that $\mcH'$ is \emph{open} if $\mcH' = \mcN(\mcH')$.  
\end{defn}
The proof of the following proposition is elementary, and we omit it. 
\begin{prop}\label{P:topology}
    Let $\mcH$ be a hypergraph. The collection of open subhypergraphs of $\mcH$
    forms a topology on $V(\mcH) \cup E(\mcH)$. A subhypergraph $\mcH'$ is closed
    in this topology if and only if, for all $v \in V(\mcH)$, if $v$ is incident
    to $e \in E(\mcH')$ then $v \in V(\mcH')$. 
\end{prop}
We use this topology when referring to closed subhypergraphs and
the closure of a subhypergraph.

\begin{defn}\label{D:generalizedmorphism}
    Let $\mcH_1$ and $\mcH_2$ be hypergraphs. A \emph{generalized morphism} 
    $\phi : \mcH_1 \arr \mcH_2$ consists of a pair of functions
    \begin{equation*}
        \phi_V : V(\mcH_1) \arr V(\mcH_2) \cup \{\vareps\} \text{ and }
        \phi_E : E(\mcH_1) \arr E(\mcH_2) \cup \{\vareps\},
    \end{equation*}
    where $\vareps$ is a special symbol not occurring in $V(\mcH_i) \cup E(\mcH_i)$, $i=1,2$,
    such that for all $v \in V(\mcH_1)$, 
    \begin{enumerate}[(1)]
        \item if $\phi_V(v) \neq \vareps$, then
            \begin{equation*}
                \sum_{e \in \phi_E^{-1}(e')} A(\mcH_1)_{ve} = A(\mcH_2)_{\phi_V(v)e'}, 
            \end{equation*}
            for all $e' \in E(\mcH_2)$, and
        \item if $\phi_V(v) = \vareps$, then 
            \begin{equation*}
                \sum_{e \in E(\mcH_1) \setminus \phi_E^{-1}(\vareps)} A(\mcH_1)_{ve}
            \end{equation*}
            is even, and $\phi_E(e_1) = \phi_E(e_2)$ for all edges $e_1,e_2 \in
            E(\mcH_1) \setminus \phi_E^{-1}(\vareps)$ incident to $v$.
    \end{enumerate}
    When there is no confusion, we write $\phi$ for both $\phi_V$ and $\phi_E$. 
    The composition $\phi_2 \circ \phi_1$ of two morphisms $\phi_1 : \mcH_1
    \arr \mcH_2$ and $\phi_2 : \mcH_2 \arr \mcH_3$ is defined by setting
    $\phi_2(\vareps)=\vareps$.
\end{defn}
\begin{prop}\label{P:basicgen}\ 
    \begin{enumerate}[(a)]
        \item If $v$ and $e$ are incident in $\mcH_1$, and $\phi : \mcH_1 \arr \mcH_2$ is a
            generalized morphism with $\phi(v) \in V(\mcH_2)$, $\phi(e) \in E(\mcH_2)$,
            then $\phi(v)$ and $\phi(e)$ are incident in $\mcH_2$. 
        \item If $\phi_1 : \mcH_1 \arr \mcH_2$ and $\phi_2 : \mcH_2 \arr
            \mcH_3$ are generalized morphisms, then $\phi_2 \circ \phi_1$
            is a generalized morphism.
    \end{enumerate}
\end{prop}
\begin{proof}
    For part (a), let $e' = \phi(e) \in E(\mcH_2)$. If $v$ and $e$ are incident
    in $\mcH_1$, and $\phi(v) \neq \vareps$, then $A(\mcH_2)_{\phi(v)e'} \geq
    A(\mcH_1)_{ve} > 0$. So $e'$ and $\phi(v)$ are incident.
    For part (b), let $\phi = \phi_2 \circ \phi_1$. If $\phi_1(v) = \vareps$,
    then $\phi(v) = \phi_2(\vareps) = \vareps$ by definition. 
    So if $\phi(v) \neq \vareps$, then $\phi_1(v) \neq \vareps$, and
    \begin{align*}
        \sum_{e \in \phi^{-1}(e')} A(\mcH_1)_{ve} & = \sum_{e'' \in \phi_2^{-1}(e')} \;
            \sum_{e \in \phi_1^{-1}(e'')} A(\mcH_1)_{ve} \\
            & = \sum_{e'' \in \phi_2^{-1}(e')} A(\mcH_2)_{\phi_1(v)e''} = A(\mcH_3)_{\phi(v)e'}
    \end{align*}
    for all $e' \in E(\mcH_3)$. 

    Next, suppose $\phi(v) = \vareps$ and $\phi_1(v) \neq \vareps$. If $e_1,
    e_2 \in E(\mcH_1) \setminus \phi^{-1}(\vareps)$ are both incident to $v$,
    then $\phi(e_1)$ and $\phi(e_2)$ belong to $E(\mcH_2) \setminus
    \phi_2^{-1}(\vareps)$ and are incident to $\phi_1(v)$ by part (a). Thus
    $\phi(e_1) = \phi(e_2)$. Similarly, 
    \begin{align*}
        \sum_{e \in E(\mcH_1) \setminus \phi^{-1}(\vareps)} A(\mcH_1)_{ve}
      & = \sum_{e' \in E(\mcH_2) \setminus \phi_2^{-1}(\vareps)}\; \sum_{e \in \phi_1^{-1}(e')} A(\mcH_1)_{ve}\\
        & = \sum_{e' \in E(\mcH_2) \setminus \phi_2^{-1}(\vareps)} A(\mcH_2)_{\phi_1(v)e'}
    \end{align*}
    is even, since $\phi_2(\phi_1(v)) = \vareps$. 

    Finally, suppose $\phi(v) = \phi_1(v) = \vareps$. If $E(\mcH_1) \setminus
    \phi^{-1}(\vareps)$ does not contain any edges incident with $v$, then part
    (2) of Definition \ref{D:generalizedmorphism} is trivially satisfied.
    Suppose on the other hand that $E(\mcH_1) \setminus \phi^{-1}(\vareps)$
    contains edges incident with $v$. All such edges belong to $E(\mcH_1)
    \setminus \phi_1^{-1}(\vareps)$, and hence are sent by $\phi_1$ to some
    common edge $e' \in E(\mcH_2)$, where $\phi_2(e') \neq \vareps$. Conversely,
    any edge of $E(\mcH_1) \setminus \phi_1^{-1}(\vareps)$ incident with $v$
    is sent by $\phi_1$ to $e'$, and hence belongs to $E(\mcH_1) \setminus
    \phi^{-1}(\vareps)$. We conclude that all edges of $E(\mcH_1) \setminus
    \phi^{-1}(\vareps)$ incident with $v$ are sent to $\phi_2(e')$, and that
    \begin{equation*}
        \sum_{e \in E(\mcH_1) \setminus \phi^{-1}(\vareps)} A(\mcH_1)_{ve}
            = \sum_{e \in E(\mcH_1) \setminus \phi_1^{-1}(\vareps)} A(\mcH_1)_{ve}
    \end{equation*}
    is even. Consequently $\phi$ is a generalized morphism.
\end{proof}

\begin{figure}
    \begin{tikzpicture}[auto,ultra thick,vertex/.style={circle,draw,thin,inner sep=2.5},every node/.style={scale=.8},
                    hyperedge/.style={thin,pattern=north west lines}]

    \node[vertex] (1) {1};
    \node[vertex] (2) [right=of 1] {2}
        edge node [swap] {$a$} (1) ;
    \node[vertex] (3) [right=of 2] {3}
        edge node [swap] {$b$} (2) ;
    \node[vertex] (4) [below=of 2] {4}
        edge node {$c$} (1)
        edge node [swap] {$d$} (3) ;
    \node[vertex] (6) [below=of 4] {6};
    \node[vertex] (5) [left=of 6] {5}
        edge node {$e$} (1) 
        edge node {$f$} (4) ;
    \node[vertex] (7) [right=of 6] {7}
        edge node [swap] {$g$} (3) 
        edge node [swap] {$h$} (4) ;
    \node[vertex] (8) [below=of 6] {8};

    \node (dummy) [above=of 2] {};
    \draw [hyperedge] (2.100) to [out=135, in=180] (dummy) to [out=0,in=45] node {$i$} (2.80) to [bend right] cycle;
    \draw [hyperedge] (5.0) to [out=0,in=180] (6.180) to [bend right] (6.200) 
            to [out=-135,in=135] (8.160) to [bend right] (8.180)
            to [out=145,in=-60] node {$j$} (5.340);
    \draw [hyperedge] (7.180) to [out=180,in=0] (6.0) to [bend left] (6.340) 
            to [out=-45,in=45] (8.20) to [bend left] (8.0)
            to [out=35,in=-120] node [swap] {$k$} (7.200);

    \node[vertex] (1') at ($(3)+(2.7,0)$) {$1'$};
    \node[vertex] (3') at ($(1')+(3)-(1)$) {$3'$}
            edge node [swap] {$a'$} (1') ;
    \node[vertex] (5') at ($(7)+(2.7,0)$) {$5'$}
            edge node {$e'$} (1');
    \node[vertex] (6') at ($(5')+(6)-(5)$) {$6'$}
            edge node {$j'$} (5');
    \node[vertex] (7') at ($(5')+(7)-(5)$) {$7'$}
            edge node {$k'$} (6')
            edge node [swap] {$g'$} (3');

    \draw[hyperedge] (1'.325) to [out=-35,in=-145] (3'.215) to [bend right] (3'.235)
            to [out=-125,in=125] (7'.125) to [bend right] (7'.145) 
            to [out=145,in=35] (5'.35) to [bend right] (5'.55)
            to [out=55,in=-55] (1'.305);

    \node[fill=white] at ($(1')+(4)-(1)$) {$c'$};
    \node at ($(4)+(2.7,0)$) {$\implies$};
\end{tikzpicture}
    \caption{A generalized morphism $\phi$ between two hypergraphs. The
        morphism sends an edge or vertex $x$ to $x'$, with the following
        exceptions: $\phi(i) = \phi(2) = \phi(4) = \vareps$, $\phi(b) = a'$,
        $\phi(d) = \phi(f) = \phi(h) = c'$, and $\phi(8) = 6'$.}
    \label{F:genmorphism}
\end{figure}
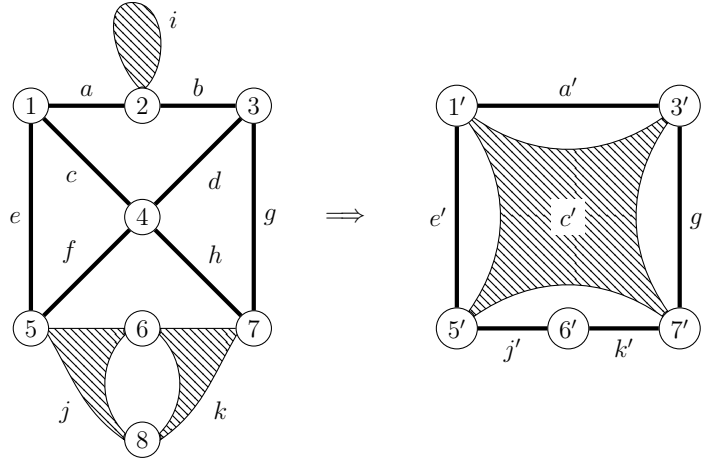
\begin{example}
    If $e$ is an edge in a hypergraph $\mcH$, we can delete $e$ to get
    a new hypergraph $\mcH \setminus e$. There is a generalized morphism
    $\mcH \arr \mcH \setminus e$ which sends $e \mapsto \vareps$. Similarly,
    we construct other generalized morphisms by identifying edges, deleting 
    isolated vertices, and collapsing vertices incident to an even number of
    edges (deleting a vertex and identifying all incident edges). An example of
    a series of these operations is shown in Figure \ref{F:genmorphism}.
\end{example}
Another way we can construct generalized morphisms is through subhypergraphs.
\begin{prop}\label{P:submorphism}
    Let $\mcH$ be a hypergraph.
    \begin{enumerate}[(a)]
        \item If $\mcH'$ is closed subhypergraph, then the function 
            \begin{equation*}
                r : V(\mcH) \cup E(\mcH) \arr V(\mcH') \cup E(\mcH') \cup \{\vareps\} : 
                    x \mapsto \begin{cases} x & x \in V(\mcH') \cup E(\mcH') \\
                                           \vareps & \text{ otherwise}
                              \end{cases}
            \end{equation*}
            is a generalized morphism $\mcH \arr \mcH'$.
        \item If $\mcH'$ is an open subhypergraph, then the inclusion map 
            \begin{equation*}
                \iota : V(\mcH') \cup E(\mcH') \arr V(\mcH) \cup E(\mcH) : x \mapsto x
            \end{equation*}
            is a generalized morphism $\mcH' \arr \mcH$.
    \end{enumerate}
\end{prop}
\begin{proof}
    If $\mcH'$ is any subhypergraph, then part (1) of Definition
    \ref{D:generalizedmorphism} holds for $r$ and part (2) holds vacuously for
    $\iota$.  If $\mcH'$ is closed and $r(v) = \eps$, then $r(e) = \eps$ for
    all edges $e$ incident to $v$, so part (2) of Definition
    \ref{D:generalizedmorphism} holds for $r$. If $\mcH'$ is open, $v \in 
    V(\mcH')$, and $e' \in E(\mcH)$ is incident to $v$, then $e'$ belongs to
    $\mcH'$, and hence part (1) holds for $\iota$. 
\end{proof}
\begin{defn}\label{D:retract}
    A subhypergraph $\mcH'$ of $\mcH$ is a \emph{retract} if there is a
    generalized morphism $r : \mcH \arr \mcH'$ such that $r|_{\mcH'}$ is the
    identity. 
\end{defn}
Part (a) of Proposition \ref{P:submorphism} shows that every closed
subhypergraph is a retract. 

Morphisms can also be constructed by gluing together morphisms over open
subhypergraphs.
\begin{prop}\label{P:gluing}
    Let $\{\mcH_i\}_{i \in I}$ be a family of open subhypergraphs of $\mcH$
    such that $\bigcup \mcH_i = \mcH$, and let $\{\phi_i\}_{i \in I}$ be a
    family of generalized morphisms $\phi_i : \mcH_i \arr \mcH'$ such that
    $\phi_i|_{\mcH_i \cap \mcH_j} = \phi_j|_{\mcH_i \cap \mcH_j}$.
    Then there is a unique generalized morphism $\phi : \mcH \arr \mcH'$
    such that $\phi|_{\mcH_i} = \phi_i$. 
\end{prop}
\begin{proof}
    Clearly $\phi$ is uniquely defined as a function. Given $v \in V(\mcH)$,
    find $\mcH_i$ with $v \in V(\mcH_i)$, so $\phi(v) = \phi_i(v)$. Since
    $\mcH_i$ is open, if $A_{ve}(\mcH) > 0$ then $e \in \mcH_i$ and
    $A(\mcH_i)_{ve} = A(\mcH)_{ve}$. Consequently, if $\phi(v) \neq \vareps$ then
    \begin{equation*}
        \sum_{e \in \phi^{-1}(e')} A(\mcH)_{ve} = \sum_{\substack{e \in \phi^{-1}(e') \\
            e \in E(\mcH_i)}} A(\mcH_i)_{ve} = \sum_{e \in \phi_i^{-1}(e')}
            A(\mcH_i)_{ve} = A(\mcH')_{\phi(v)e'}
    \end{equation*}
    since $\phi_i$ is a morphism. Similarly, if $\phi(v) = \vareps$, then
    \begin{equation*}
        \sum_{e \in E(\mcH) \setminus \phi^{-1}(\vareps)} A(\mcH)_{ve}
        = \sum_{e \in E(\mcH_i) \setminus \phi^{-1}(\vareps)} A(\mcH_i)_{ve}
        = \sum_{e \in E(\mcH_i) \setminus \phi_i^{-1}(\vareps)} A(\mcH_i)_{ve}
    \end{equation*}
    is even, and $\phi(e_1) = \phi_i(e_1) = \phi_i(e_2) = \phi(e_2)$ for all
    edges $e_1,e_2 \in E(\mcH) \setminus \phi^{-1}(\vareps)$ incident to $v$.
\end{proof}

By Proposition \ref{P:basicgen}, part (b), hypergraphs with generalized
morphisms form a category. The next proposition shows that the map $\mcH
\mapsto \Gamma(\mcH)$ extends to a functor from the category of 
hypergraphs to the category of groups over $\Z_2$.
\begin{prop}\label{P:morphism}
    Let $\phi : \mcH_1 \arr \mcH_2$ be a generalized morphism. Then there is a 
    morphism $\phi : \Gamma(\mcH_1) \arr \Gamma(\mcH_2)$ over $\Z_2$ defined by 
    \begin{equation*}
        \phi(x_e) = \begin{cases} 1 & \phi_E(e) = \vareps \\
                                        x_{\phi(e)} & \text{ otherwise}
                          \end{cases}.
    \end{equation*}
\end{prop}
\begin{proof}
    The morphism, if it exists, will be uniquely determined by the values
    $\phi(x_e)$, $e \in E(\mcH_1)$. To show that the morphism is well-defined,
    we need to show that
    \begin{equation}\label{E:welldefined}
        \prod_{i=1}^n \phi(x_{e_i}) = 1
    \end{equation}
    for every vertex $v$ of $\mcH_1$ and ordering $e_1,\ldots,e_n$ of the edges
    incident to $v$, listed with multiplicity. Suppose that $\phi(v) =
    \vareps$. If $\phi(e_i) = \vareps$ for all $1 \leq i \leq n$, then
    $\phi(x_{e_i}) = 1$ for all $i$, so Equation \eqref{E:welldefined} holds.
    If $\phi(e_i) \neq \vareps$ for some $i$, then
    \begin{equation*}
        \prod_{j=1}^n \phi(x_{e_j}) = \prod_{\phi(e_j) \neq \vareps} x_{\phi(e_j)}
            = x_{\phi(e_i)}^K = 1,
    \end{equation*}
    where the last equality holds because $K = \sum_{e \in E(\mcH_1) \setminus
    \phi_E^{-1}(\vareps)} A_{ve}$ is even.  Hence Equation
    \eqref{E:welldefined} holds in this case as well.

    If $\phi(v) \neq \vareps$, then for all $1 \leq i \leq n$, either
    $\phi(e_i) = \vareps$ or $\phi(e_i)$ is incident to $\phi(v)$. We conclude
    that $\phi(x_{e_1}),\ldots,\phi(x_{e_n})$ commute in $\Gamma(\mcH_2)$.
    Consequently 
    \begin{equation*}
        \prod_{i=1}^n \phi(x_{e_i}) = \prod_{e' \in E(\mcH_2)} \prod_{e \in \phi_E^{-1}(e')}
            \phi(x_{e})^{A_{ve}} = \prod_{e' \in E(\mcH_2)} x_{e'}^{A_{\phi(v)e'}} = 1.
    \end{equation*}
    We conclude that the morphism $\phi$ is well-defined. 
\end{proof}
As a consequence, open retracts are special:
\begin{cor}\label{C:retract}
    If $\mcH'$ is an open subhypergraph of $\mcH$ and a retract of $\mcH$ then
    $\Gamma(\mcH')$ is a (semidirect factor) subgroup of $\Gamma(\mcH)$. 
\end{cor}
\begin{proof}
    Let $r$ be the retraction morphism $\mcH \arr \mcH'$, and let $\iota$
    be the inclusion $\mcH' \arr \mcH$. Then $r \circ \iota$ is the identity
    on $\mcH'$, so the composition $r \circ \iota : \Gamma(\mcH') \arr
    \Gamma(\mcH) \arr \Gamma(\mcH')$ is the identity morphism, and consequently
    $\iota : \Gamma(\mcH') \arr \Gamma(\mcH)$ must be injective. It also
    follows immediately that $\Gamma(\mcH) = N \rtimes \Gamma(\mcH')$, where
    $N$ is the kernel of $r : \Gamma(\mcH) \arr \Gamma(\mcH')$. 
\end{proof}
\begin{example}\label{Ex:cube}
    Consider the graph $\mcG$ of a cube, shown on the left in Figure
    \ref{F:cube}, with vertices numbered for reference. 
    \begin{figure}
        \begin{tikzpicture}[auto,ultra thick,scale=.5,vertex/.style={circle,draw,thin,inner sep=2.5},every node/.style={scale=.8},
                    hyperedge/.style={thin,pattern=north west lines},empty/.style={inner sep=0}]

    \begin{scope}
        \node[empty] (1) at (0.000000,0.000000) {};
        \node[empty] (2) at (4.095760,-1.212019) {};
        \node[empty] (3) at (1.227878,-2.942962) {};
        \node[empty] (4) at (-2.867882,-1.730943) {};
        \node[empty] (5) at (0.000000,4.531539) {};
        \node[empty] (6) at (4.095760,3.319520) {};
        \node[empty] (7) at (1.227878,1.588576) {};
        \node[empty] (8) at (-2.867882,2.800596) {};

        \draw[dashed] (1) -- (2); \draw[dashed] (1) -- (4); \draw[dashed] (1) -- (5); 
        \draw (2) -- (3); \draw (2) -- (6);
        \draw (3) -- (4); \draw (3) -- (7);
        \draw (4) -- (8);
        \draw (5) -- (6); \draw (5) -- (8); 
        \draw (6) -- (7); \draw (7) -- (8);

        \node[above left] at (1) {$1$};
        \node[below right] at (2) {$2$};
        \node[below right] at (3) {$3$};
        \node[below left] at (4) {$4$};
        \node[above left] at (5) {$5$};
        \node[above right] at (6) {$6$};
        \node[below right] at (7) {$7$};
        \node[above left] at (8) {$8$};
        
        \foreach \x in {1,...,8}
            \draw[fill] (\x) circle [radius=.15];
    \end{scope}

    \begin{scope}[shift={++(13,0)}]
        \node[empty] (1) at (0.000000,0.000000) {};
        \node[empty] (2) at (4.095760,-1.212019) {};
        \node[empty] (3) at (1.227878,-2.942962) {};
        \node[empty] (4) at (-2.867882,-1.730943) {};
        \node[empty] (5) at (0.000000,4.531539) {};
        \node[empty] (6) at (4.095760,3.319520) {};
        \node[empty] (7) at (1.227878,1.588576) {};
        \node[empty] (8) at (-2.867882,2.800596) {};
        \node[empty] (9) at (2.662,-2.077) {};

        \draw[dashed] (1) -- (2); \draw[dashed] (1) -- (4); \draw[dashed] (1) -- (5);
        \draw (2) -- (3); \draw (2) -- (6);
        \draw (3) -- (4); \draw (3) -- (7);
        \draw (4) -- (8);
        \draw (5) -- (6); \draw (5) -- (8);
        \draw (6) -- (7); \draw (7) -- (8);


        \node[above left] at (1) {$1$};
        \node[below right] at (2) {$2$};
        \node[below right] at (3) {$3$};
        \node[below left] at (4) {$4$};
        \node[above left] at (5) {$5$};
        \node[above right] at (6) {$6$};
        \node[below right] at (7) {$7$};
        \node[above left] at (8) {$8$};
        \node[below right] at (9) {$9$};
        
        \foreach \x in {1,...,9}
            \draw[fill] (\x) circle [radius=.15];
    \end{scope}

\end{tikzpicture}
        \caption{In the hypergraph of the cube (left), the open neighbourhood
            of the bottom face is a retract. In the hypergraph on the right,
            this is no longer the case.}
        \label{F:cube}
    \end{figure}
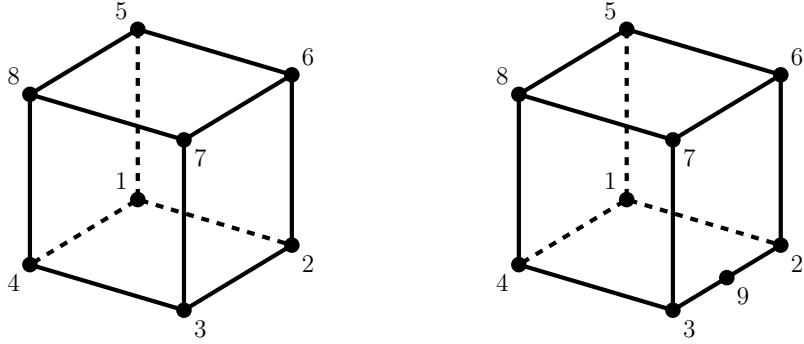
    The open neighbourhood of $\{1,2,3,4\}$ is the subhypergraph $\mcH$ with
    vertex set $\{1,2,3,4\}$ and edge set $\{12,23,34,14,15,26,37,48\}$. The
    morphism $\mcG \arr \mcH$ which is the identity on $\mcH$, and sends
    \begin{align*}
        & 5 \mapsto 1, \quad 6 \mapsto 2, \quad 7 \mapsto 3, \quad 8 \mapsto 4, \text{ and} \\
        & 56 \mapsto 12,\quad  67 \mapsto 23,\quad  78 \mapsto 34,\quad  58 \mapsto 14,
    \end{align*}
    is a retract of $\mcG$ onto $\mcH$. 

    On the other hand, if we subdivide the edge $12$ with a vertex as shown
    on the right of Figure \ref{F:cube}, then the open neighbourhood $\mcH$ of
    $\{1,2,3,4,9\}$ is no longer a retract of $\mcG$, since any retract must
    send $6 \mapsto 2$ and $7 \mapsto 3$, but $23$ is not an edge of $\mcH$.
    However, the open neighbourhood of $\{5,6,7,8\}$ is still a retract, since
    we can define the morphism as above but with edges $29$ and $39$ sent to
    $67$, and vertex $9$ sent to $\vareps$. 
\end{example}

\begin{figure}
    \begin{tikzpicture}[auto,vertex/.style={circle,draw,thin,inner sep=2.5},every node/.style={scale=.8},
                    hyperedge/.style={thin,pattern=north west lines},scale=.6]

    \begin{scope}
        \draw (5,5) circle [radius=4.8];
        \node[vertex] (2) at (1.75,6) {2};
        \node[vertex] (3) at (3,7.5) {3};
        \node[vertex] (1) at (3,4.5) {1};
        \node[vertex] (4) at (4.25,6) {4};
        \node[vertex] (7) at (6.5,8.5) {7};
        \node[vertex] (5) at (6,5) {5};
        \node[vertex] (6) at (8.5,4) {6};
        \node[vertex] (8) at (6,3) {8};
        \node[vertex] (5') at (4.25,2) {5};

        \begin{scope}[ultra thick]
            \clip (5,5) circle [radius=4.8];
            \draw (0,7) -- node {$i$} (2) 
                        -- node [swap] {$a$} (1) 
                        -- node {$c$} (4)
                        -- node [swap] {$d$} (3)
                        -- node [swap] {$b$} (2);
            \draw (4) -- node [swap,pos=.6] {$h$} (7) -- node [swap] {$g$} (3);
            \draw (7) -- node [swap,pos=.2] {$k$} (10,10);
            \draw (1) to [bend right] node [pos=.4,swap] {$e$} (5);
            \draw (5) -- node [swap,pos=.6] {$f$} (4);
            \draw (5) -- node {$j$} (6)
                      -- node {$k$} (8)
                      -- node {$j$} (5')
                      -- node [swap,pos=.35] {$e$} (0,0);
            \draw (5') -- node [pos=.55] {$f$} (4.25,0);
        \end{scope}
    \end{scope}

    \begin{scope}[shift={++(12,0)}]
        \draw (5,5) circle [radius=4.8];
        \coordinate (2) at (1.75,6);
        \node[vertex] (3) at (3,7.5) {$3'$};
        \node[vertex] (1) at (3,4.5) {$1'$};
        \node (4) at (4.25,6) {};
        \node[vertex] (7) at (6.5,8.5) {$7'$};
        \node[vertex] (5) at (6,5) {$5'$};
        \node[vertex] (6) at (8.5,4) {$6'$};
        \node[vertex] (8) at (6,3) {$6'$};
        \node[vertex] (5') at (4.25,2) {$5'$};

        \begin{scope}[ultra thick]
            \clip (5,5) circle [radius=4.8];
            \draw (1) to [out=20,in=-20] node {$c'$} (3);
            \draw (3) to [out=-170,in=170] node {$a'$} (1);
            \draw (5) to [out=170,in=-135] node [swap,pos=.65] {$c'$} (7);
            \draw (7) -- node [swap] {$g'$} (3);
            \draw (7) -- node [swap,pos=.12] {$k'$} (10,10);
            \draw (1) to [bend right] node [swap,pos=.4] {$e'$} (5);
            \draw (5) -- node {$j'$} (6)
                      -- node {$k'$} (8)
                      -- node {$j'$} (5')
                      -- node [swap,pos=.3] {$e'$} (0,0);
            \draw (5') -- node [pos=.55] {$c'$} (4.25,0);
        \end{scope}
    \end{scope}
    \node at (11,5) {$\implies$};
\end{tikzpicture}
    \caption{The generalized morphism from Figure \ref{F:genmorphism} applied
             to a picture. A choice must be made in regards to the edges formerly
             incident to vertex $4$. We have chosen to connect $1' - 3'$ and
             $5' - 7'$ instead of $3' - 7'$ and $1' - 5'$.}
    \label{F:morphismpic}
\end{figure}
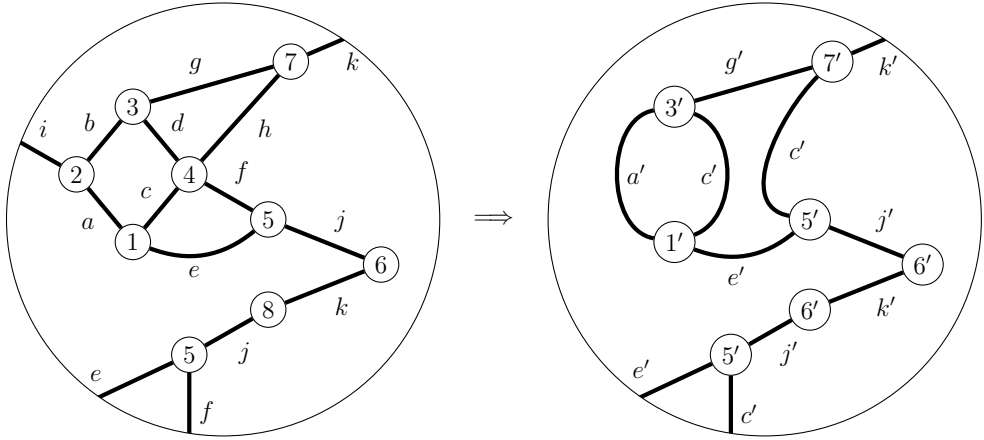
Finally, we can also apply generalized morphisms to pictures.
\begin{prop}\label{P:morphismpic}
    Let $\phi : \mcH_1 \arr \mcH_2$ be a generalized morphism, and let $\mcP$
    be an $\mcH_1$-picture with $\bd(\mcP) = e_1 \cdots e_n$. Construct a new
    picture $\mcP'$ as follows:
    \begin{enumerate}[(1)]
        \item If $e$ is an edge of $\mcP$ such that $\phi(h(e)) = \vareps$, then 
            delete $e$ from $\mcP$. 
        \item For all remaining edges $e$ of $\mcP$, change the label from
            $h(e) \in E(\mcH_1)$ to $\phi(h(e)) \in E(\mcH_2)$. 
        \item If $v$ is a vertex of $\mcP$ such that $\phi(h(v)) = \vareps$, then
            after applying steps (1) and (2), there is an even number of edges
            incident with $v$, and all have the same label in $E(\mcH_2)$. Delete $v$, and
            connect up the remaining incident edges so that no pair of edges
            cross. 
        \item For all remaining vertices $v$ of $\mcP$, change the label from
            $h(v) \in E(\mcH_1)$ to $\phi(h(v)) \in E(\mcH_2)$. 
    \end{enumerate}
    Then $\mcP'$ is an $\mcH_2$-picture of size less than or equal to the size
    of $\mcP$, with $\bd(\mcP') = \phi(e_1) \cdots \phi(e_n)$, where $\vareps$
    is regarded as the empty word.
\end{prop}
\begin{proof}
    To show that this picture is an $\mcH_2$-picture, we need to show that
    there are $A(\mcH_2)_{v'e'}$ edges labelled by $e'$ incident to any vertex labelled
    by $v'$ in $\mcP'$. But this follows immediately from the construction and
    part (1) of Definition \ref{D:generalizedmorphism}. Since the construction
    does not add any vertices, the size of $\mcP'$ must be at most the size of
    $\mcP$.
\end{proof}
An example is given in Figure \ref{F:morphismpic}. Since step (3) in
Proposition \ref{P:morphismpic} requires choosing a matching on edges, the
graph $\mcP'$ we end up with is not unique. One exception is when $\phi$
is one of the morphisms $r$ or $\iota$ from Proposition \ref{P:submorphism};
in this case, there is always a unique choice, and hence $\mcP'$ is uniquely
determined. However, in general the different choices can differ by a sequence of
null-surgeries.
\begin{defn}
    If $\phi : \mcH_1 \arr \mcH_2$ is a generalized morphism and $\mcP$ is an
    $\mcH_1$-picture, we use $\phi(\mcP)$ to denote either the null-surgery
    equivalence class of pictures constructed in Proposition
    \ref{P:morphismpic}, or some arbitrarily chosen representative of this
    class. 

    If $r : \mcH \arr \mcH'$ is the retraction morphism onto a closed
    subhypergraph, and $\mcP$ is a $\mcH$-picture, then we also denote
    $r(\mcP)$ by $\mcP[\mcH']$. 
\end{defn}
If $\iota$ is the inclusion of an open subhypergraph $\mcH'$ in $\mcH$, and
$\mcP$ is an $\mcH'$-picture, then $\mcP$ and $\iota(\mcP)$ are essentially
identical. In this case, Proposition \ref{P:morphismpic} states the obvious
fact that every picture over $\mcH'$ can be regarded as a picture over $\mcH$.
Note that if such a picture is character (or $b$)-minimal as an $\mcH$-picture,
then it is also minimal as an $\mcH'$-picture; however, the converse is not
true. 

\section{Cycles and outer faces}\label{S:cycle} 
In this section, we lay the foundation for the proof of Theorem
\ref{T:invembedding} by looking at the interaction between pictures and cycles
in hypergraphs. 

\begin{defn}\label{D:cycle}
    A \emph{cycle} in a hypergraph $\mcH$ is a closed subhypergraph $\mcC$
    which is a simple connected $2$-regular graph. A cycle $\mcC$ is
    \emph{cubic} if every vertex $v \in V(\mcC)$ has degree three in $\mcH$.

    A $\mcC$-cycle in an $\mcH$-picture is a simple cycle $C$ such that every
    edge of $C$ is labelled by an edge of $\mcC$. 
\end{defn}

\begin{lemma}\label{L:retractcycle}
    Let $\mcC$ be a cycle in a hypergraph $\mcH$, and suppose $\mcP$ is an
    $\mcH$-picture such that the edges of $\mcC$ do not appear in $\bd(\mcP)$.
    Then every connected component of $\mcP[\mcC]$ is a $\mcC$-cycle.
\end{lemma}
\begin{proof}
    $\mcP[\mcC]$ is a closed $\mcC$-picture in which every vertex has
    degree two. As a result, every connected component will be a 
    simple closed curve.
\end{proof}
In several upcoming proofs, we will use the following measure of the complexity 
of $\mcP$. 
\begin{defn}
    If $\mcC$ is a cycle in $\mcH$, and $\mcP$ is an $\mcH$-picture, then we let
    $\Cycle(\mcP,\mcC)$ denote the number of $\mcC$-cycles in $\mcP$.  If
    $\Phi$ is a collection of cycles, we let
    \begin{equation*}
        \Cycle(\mcP,\Phi) = \sum_{\mcC \in \Phi} \Cycle(\mcP,\mcC).
    \end{equation*}
\end{defn}
If we include incident vertices, then a $\mcC$-cycle $C$ is a closed
$\mcC$-picture, and thus the definition of cover from Definition \ref{D:cover}
applies. If $C$ is both facial and a cover of $\mcC$, then we say that $C$ is a
\emph{facial cover}. Note that closed loops are not covers. 
\begin{prop}\label{P:facialcover}
    If $\mcC$ is a cubic cycle in $\mcH$, and $\mcP$ is a character-minimal
    $\mcH$-picture with no closed loops, then every facial $\mcC$-cycle in
    $\mcP$ is a facial cover.
\end{prop}
\begin{proof}
    Suppose that $C$ is a facial $\mcC$-cycle. Since $\mcC$ is not a loop at a
    vertex, and $\mcP$ has no closed loops, every edge of $C$ must be incident
    to two distinct vertices of $\mcP$. Suppose that $C$ has an edge $e$ connecting two
    vertices $u_0$ and $u_1$ with $h(u_0) = v = h(u_1)$. By hypothesis, $v$
    has degree $3$ in $\mcH$, and thus is incident with $e'=h(e)$, another edge
    $f'$ of $\mcC$, and an edge $g'$ which is not in $\mcC$. Hence each vertex
    $u_i$ is incident with an edge $f_i$ belonging to $C$ with $h(f_i) = f'$, and
    an edge $g_i$ not in $C$ with $h(g_i) = g'$. Note that the edges $f_0$ and
    $f_1$ could be equal, as could $g_0$ and $g_1$. Since $C$ is facial, $g_0$
    and $g_1$ must lie on the same side of $C$. Thus there is a surgery which
    removes $u_0$ and $u_1$, connecting $g_0$ with $g_1$ and $f_0$ with $f_1$,
    as shown below:
    
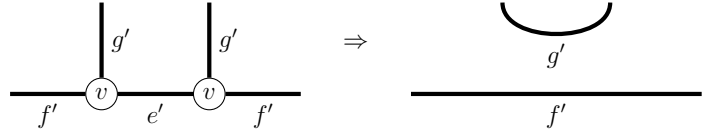
\begin{equation*}
        \begin{tikzpicture}[auto,vertex/.style={circle,draw,thin,inner sep=2.5},every node/.style={scale=.8}, ultra thick]

    \node [vertex] (a) {$v$};
    \node [vertex] (b) [right=of a] {$v$} 
        edge node {$e'$} (a);

    \node (l) [left=of a] {}
        edge node [swap] {$f'$} (a);
    \node (t1) [above=of a] {}
        edge node {$g'$} (a);
    \node (r) [right=of b] {}
        edge node {$f'$} (b);
    \node (t2) [above=of b] {}
        edge node {$g'$} (b);

    \node (m) at ($(r)+(.6,0)+.5*(t1)-.5*(a)$) {$\surg$};

    \node (l') at ($(r)+(1.2,0)$) {};
    \node (r') at ($(l')+(r)-(l)$) {}
        edge node {$f'$} (l');
    \node (t1') at ($(l')+(t1)-(l)$) {};
    \node (t2') at ($(l')+(t2)-(l)$) {}
        edge [out=-90,in=-90] node {$g'$} (t1');
    
\end{tikzpicture}
    \end{equation*}
    But this means that $\mcP$ is not character-minimal.        
\end{proof}

A connected closed cover $C$ of a cycle $\mcC$ is determined up to isotopy by an
orientation of $\mcC$, and the \emph{ply}, that is, the size of $h^{-1}(v)$ for
any $v \in V(\mcC)$, where $h : C \arr \mcC$ is the labelling function.
\begin{defn} 
    A $\mcC$-cycle $C$ is a \emph{copy} of $\mcC$ if the labelling function $h
    : C \arr \mcC$ is a graph isomorphism, or equivalently if $|h^{-1}(v)|=1$
    for all $v \in V(\mcC)$. 
\end{defn}

We are also interested in how different cycles interact. 
\begin{lemma}\label{L:boundary1}
    Let $\mcC$ be a cycle in $\mcH$, and let $\mcD$ be a simple region in an
    $\mcH$-picture $\mcP$. If $C$ is a facial $\mcC$-cycle in $\mcP$, then either $C$
    is contained in $\mcD$, or the edges of $\germ(C,\mcD)$ form the boundaries
    of simple outer faces of $\germ(\mcP,\mcD)$. 
\end{lemma}
\begin{proof}
    Suppose $C$ is not contained in $\mcD$. If $\germ(C,\mcD)$ is empty, then
    the lemma is vacuously true. Otherwise, choose $\eps > 0$ such that
    $\germ(\mcP,\mcD) = \res(\mcP,\mcD^{\eps})$, and in particular such that
    $C$ intersects the boundary of $\mcD^{\eps}$ transversally. Then $C \cap
    \mcD^{\eps}$ is divided into a number of segments which start and end on
    the boundary of $\mcD^{\eps}$. Thus every face of $\germ(C,\mcD)$ will be
    a simple region whose boundary contains some (possibly disconnected) portion of the
    boundary of $\mcD^{\eps}$.  Each edge of $\germ(C,\mcD)$ will be incident
    to exactly two such faces, and one of these two faces will be contained in
    $\mcF \cap \mcD^{\eps}$, where $\mcF$ is a face in $\mcP$ bounded by $C$. Since
    $\mcF$ is a face, any face of $\germ(C,\mcD)$ contained in $\mcF \cap
    \mcD^{\eps}$ will be a face of $\germ(\mcP,\mcD)$.
\end{proof}
\begin{lemma}\label{L:boundary2}
    Let $C$ be a simple cycle in a picture $\mcP$, let $\mcD$ be a simple
    region bounded by $C$, and let $\mcF$ be a simple outer face of
    $\germ(\mcP,\mcD)$. Then the edges of $\germ(\mcP,\mcD)$ in the boundary
    of $\mcF$ form a single simple path $P$ with at least two edges.

    Furthermore, if we write the edges of $P$ in order as $e_1,\ldots,e_{n}$,
    then $e_1$ and $e_n$ are incident with the boundary of $\mcD$, and
    $e_{2},\ldots, e_{n-1}$ belong to $C$.
\end{lemma}
\begin{proof}
    Every edge of $\bd(\germ(\mcP,\mcD))$ is incident to a vertex in $C$.
    If $\mcF$ is simple, then $\bd(\germ(\mcP,\mcD))$ must contain at least two
    edges. If these edges are incident with two or more vertices of $C$, then
    it follows from the definition of $\germ(\mcP,\mcD)$ that every outer
    face is a simple region with boundary of the required form. If all the
    edges of $\bd(\germ(\mcP,\mcD))$ are incident with a single vertex in $C$,
    then all but one of the outer faces is simple, and the simple outer
    faces are bounded by a path formed by two edges of $\bd(\germ(\mcP,\mcD))$.
\end{proof}

\begin{defn}
    Let $\mcP$ be a picture in $\mcD_0$, where $\mcD_0$ is not the whole
    sphere. An \emph{outer quadrilateral} of $\mcP$ is the closure of a simple
    outer face $\mcF$ of $\mcP$, whose boundary contains exactly three edges of
    $\mcP$, together forming a path of length three. 
\end{defn}
An outer quadrilateral must have a fourth side consisting of points in the
boundary of $\mcD_0$, hence the name. A picture with highlighted outer
quadrilaterals is shown in Figure \ref{F:boundaryquad}.

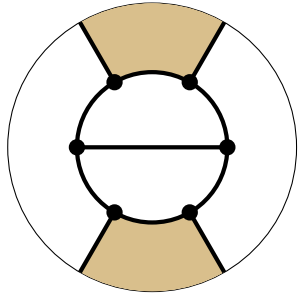
\begin{figure}
    \begin{tikzpicture}[scale=.4]

    \draw (5,5) circle [radius=4.8];
    \clip (5,5) circle [radius=4.8];

    \begin{scope}[ultra thick]
        \path (7.5,5) coordinate (1) 
                    arc [start angle=0, end angle=60, radius=2.5] coordinate (2)
                    arc [start angle=60, end angle=120,radius=2.5] coordinate (3)
                    arc [start angle=120, end angle=180,radius=2.5] coordinate (4)
                    arc [start angle=180, end angle=240,radius=2.5] coordinate (5)
                    arc [start angle=240, end angle=300,radius=2.5] coordinate (6)
                    arc [start angle=300, end angle=360,radius=2.5];

        \path [fill=fillyellow] (2) -- ++(60:2.3) arc [start angle=60, end angle=120, radius=4.8]
                -- ++(-60:2.3) arc [start angle=120, end angle=60, radius=2.5];
        \path [fill=fillyellow] (5) -- ++(240:2.3) arc [start angle=240, end angle=300, radius=4.8]
                -- ++(120:2.3) arc [start angle=300, end angle=240, radius=2.5];

        \draw (5,5) circle [radius=2.5];

        \draw (1) -- (4);
        \draw (2) -- +(60:2.5);
        \draw (3) -- +(120:2.5);
        \draw (5) -- +(240:2.5);
        \draw (6) -- +(300:2.5);

        \foreach \x in {1,...,6}
            \draw[fill] (\x) circle [radius=.2];
    \end{scope}

\end{tikzpicture}
    \caption{A picture in the disk, with four outer faces, two of which (shaded)
            form outer quadrilaterals.}
    \label{F:boundaryquad}
\end{figure}

\begin{prop}\label{P:boundary}
    Let $\mcC$ be a cycle in $\mcH$, and let $\mcC'$ be a cubic cycle such
    that $|E(\mcC) \cap E(\mcC')| \leq 1$. If $\mcD$ is a region bounded by a
    $\mcC'$-cycle $C'$, and $C$ is a facial $\mcC$-cycle, then either $C$ is
    contained in $\mcD$, or the edges of $\germ(C,\mcD)$ form outer
    quadrilaterals of $\germ(\mcP,\mcD)$. 
\end{prop}
\begin{proof}
    By Lemma \ref{L:boundary1}, we can assume that $\mcF$ is a simple outer
    face of $\germ(\mcP,\mcD)$ formed by the edges of $\germ(C,\mcD)$. By Lemma
    \ref{L:boundary2}, the edges of $\germ(C,\mcD)$ in the boundary of $\mcF$ form a single path
    $e_0,e_1,\ldots,e_{n+1}$, where $e_0$ and $e_{n+1}$ are incident with the
    boundary and $e_1,\ldots,e_n$ belong to $C'$, as shown below: 
    
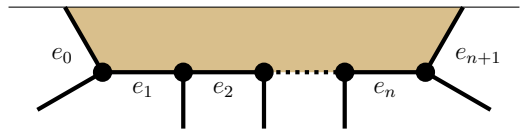
\begin{equation*}
        \begin{tikzpicture}[auto,ultra thick,scale=.5,emptynode/.style={inner sep=0},every node/.style={scale=.8}]

    \node[emptynode] (1) {};
    \node[emptynode] (2) [right=of 1] {};
    \node[emptynode] (2') [right=of 2] {};
    \node[emptynode] (3) [right=of 2'] {};
    \node[emptynode] (4) [right=of 3] {};
    \coordinate (0) at ($(1)+(120:2)$) {};
    \coordinate (5) at ($(4)+(60:2)$) {};

    \path [fill=fillyellow] (1.0) -- (0) -- (5) -- (4.0) -- (1.0);

    \draw (1) -- node[swap] {$e_1$} (2);
    \draw (2) -- node[swap] {$e_2$} (2');
    \draw [dotted] (2') -- (3);
    \draw (3) -- node[swap] {$e_n$} (4);
    \draw (1) -- node {$e_0$} (0);
    \draw (4) -- node [swap] {$e_{n+1}$} (5);
    \draw [thin] ($(0)-(1.5,0)$) -- ($(5)+(1.5,0)$);

    \draw (2) -- +(-90:1.5);
    \draw (2') -- +(-90:1.5);
    \draw (3) -- +(-90:1.5);

    \draw (1) -- +(-150:2);
    \draw (4) -- +(-30:2);

    \foreach \x in {1,2,2',3,4}
        \draw [fill] (\x) circle [radius=.2];
    
\end{tikzpicture}
    \end{equation*}
    Since $C'$ is cubic, it follows as well that $n \geq 1$.  Since $|E(\mcC)
    \cap E(\mcC')| \leq 1$ and the edges $e_1,\ldots,e_n$ belong to both $C$
    and $C'$, we must have $h(e_1) = \ldots = h(e_n) = e'$, the unique element
    of $E(\mcC) \cap E(\mcC')$. Since $\mcC$ (and $\mcC'$) are simple graphs,
    the edge $e'$ is incident with exactly two vertices $v$ and $v'$, and both
    $A_{ve'} = A_{v'e'} = 1$. Thus $n$ must be one, and we conclude that the
    closure of $\mcF$ is an outer quadrilateral.
\end{proof}

\section{Pictures over suns}\label{S:sun}

In this section, we look at pictures over a specific family of hypergraphs:
\begin{defn}\label{D:sun}
    The \emph{sun of size $n$}, where $n \geq 3$, is the hypergraph with vertex
    set $\{1,\ldots,n\}$, edge set $\{e_i,f_i : 1 \leq i \leq n\}$, and
    incidence relation
    \begin{equation*}
        A_{if_j} = \begin{cases} 1 & i = j \\
                                 0 & i \neq j 
                    \end{cases} 
        \quad \text{ and } \quad 
        A_{ie_j} = \begin{cases} 1 & i \cong j \text{ or } i \cong j+1 \mod n \\
                                 0 & \text{ otherwise} 
                   \end{cases}.
    \end{equation*}
\end{defn}
Note that a sun has a unique cycle.  The sun of size $n=6$ is shown in Figure
\ref{F:sun}.
\begin{figure}
    \begin{tikzpicture}[auto,ultra thick,vertex/.style={circle,draw,thin,inner sep=2.5},every node/.style={scale=.8},
                    hyperedge/.style={thin,pattern=north west lines},scale=.6]

    \foreach \x in {1,...,6} {
        \node[vertex] (\x) at (60*\x:3.5) {$\x$};
        \draw[hyperedge] (\x.60*\x+10) to [out=60*\x+45,in=60*\x+90] node {$f_\x$} ($1.7*(\x)$) 
                                       to [out=60*\x-90,in=60*\x-45] (\x.60*\x-10) 
                                       to [bend right] (\x.60*\x+10); }

    \draw (1) -- node [swap] {$e_1$} (2);
    \draw (2) -- node [swap] {$e_2$} (3);
    \draw (3) -- node [swap] {$e_3$} (4);
    \draw (4) -- node [swap] {$e_4$} (5);
    \draw (5) -- node [swap] {$e_5$} (6);
    \draw (6) -- node [swap] {$e_6$} (1);

\end{tikzpicture}
    \caption{The sun of size $6$.}
    \label{F:sun}
\end{figure}
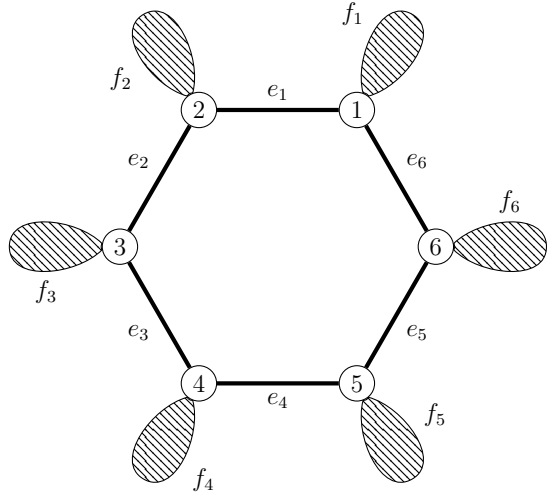

\begin{prop}\label{P:suncycle}
    Let $\mcH$ be a sun, and let $\mcP$ be a character-minimal $\mcH$-picture
    such that $\bd(\mcP)$ does not contain any edges from the cycle $\mcC$ of
    $\mcH$.  Then $\mcP$ is character-equivalent to a character-minimal picture
    $\mcP'$ with no closed loops, in which every $\mcC$-cycle is a facial
    cover. Furthermore, $\mcP'$ can be chosen so that every outer
    quadrilateral of $\mcP$ is an outer quadrilateral of $\mcP'$. 
\end{prop}
The last statement in this proposition implies in particular that if an edge
$e$ is part of an outer quadrilateral in $\mcP$, then $e$ is also an edge of
$\mcP'$, with the same label. The proof of Proposition \ref{P:suncycle} uses
the following lemma. 
\begin{lemma}\label{L:suncycle1}
    If $\mcP$ is a closed picture over a sun $\mcH$, then $\ch(\mcP)=0$.
\end{lemma}
\begin{proof}
    Let $n$ be the size of $\mcH$, and let $\mcU$ be the closure of the
    complement of the cycle $\mcC$ of $\mcH$, or in other words, the
    subhypergraph with vertex set $\{1,\ldots,n\}$ and edge set
    $\{f_1,\ldots,f_n\}$.  Then every vertex in $\mcP[\mcU]$ has degree one, so
    if $\mcP$ is closed then $\mcP[\mcU]$ is a matching. Furthermore, if $e$ is
    an edge of $\mcP[\mcU]$ with $h(e) = f_i$, and endpoints at vertices $v$ and $v'$, 
    then $h(v) = h(v') = i$. We conclude that $\mcP$ has an even number of
    vertices labelled by $i$, so $\ch(\mcP)_i = 0$ for every $i$. 
\end{proof}

\begin{proof}[Proof of Proposition \ref{P:suncycle}]
    Without loss of generality, we can assume that every picture lies 
    in a bounded region, rather than a sphere. As a consequence, every simple
    closed curve will bound a unique simple region. Suppose that $C$ is a
    $\mcC$-cycle in some $\mcH$-picture $\mcP$, and let $\mcD$ be the simple
    region bounded by $C$. For the purpose of this proof, we let $NE(C,\mcP)$
    be the number of edges $e$ in $\res(\mcP,\mcD)$ with $h(e) \in \mcU$, where
    $\mcU$ is the closure of the complement of $\mcC$, as in the proof of Lemma
    \ref{L:suncycle1}. We then set
    \begin{equation*}
        NE(\mcP) := \sum_{C \text{ a } \mcC\text{-cycle}} NE(C,\mcP).
    \end{equation*}
    Now suppose we start with some character-minimal picture $\mcP = \mcP_0$
    with no edges from $\mcC$ in $\bd(\mcP_0)$. Our strategy will be to reduce
    $NE(\mcP_0)$ via a sequence of surgeries. Let $C$ be a $\mcC$-cycle in 
    $\mcP_0$, and let $\mcD$ be the simple region bounded by $C$. Suppose $e$
    is an edge contained in $\mcD$ and incident with $C$, such that $h(e) =
    f_i$ for some $i$. Clearly $e$ cannot be incident with the boundary of
    $\mcP_0$, and since $e$ is not a closed loop, we conclude that $e$ is
    incident with two distinct vertices $v_0$ and $v_1$ with $h(v_0) = h(v_1) =
    i$. Let $e$, $a_0$, $a_1$ be the sequence of edges incident with $v_0$, as
    they appear in counter-clockwise order, and let $e$, $b_0$, $b_1$ be the
    edges incident with $v_1$ as they appear in clockwise order.  Then $a_0$,
    $e$, and $b_0$ all lie in the boundary of a common face, as do $a_1$, $e$,
    and $b_1$ (all these edges may lie in a common non-simple face, so we
    have to be careful with wording here). If $h(a_0)
    = h(b_0)$ and $h(a_1) = h(b_1)$, then we could delete $e$, $v_0$, and
    $v_1$, and connect $a_i$ with $b_i$ (shown in the diagram below, where we
    assume without loss of generality that $h(a_0) = e_i$)
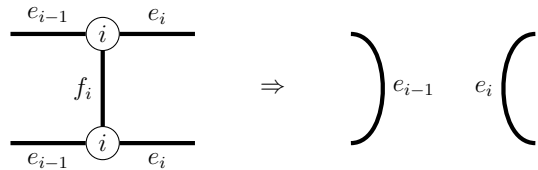
    \begin{equation*}
        \begin{tikzpicture}[auto,vertex/.style={circle,draw,thin,inner sep=2.5},every node/.style={scale=.8}, 
                    empty/.style={inner sep=0},ultra thick]

    \node[empty] (lt) {};
    \node [vertex] (v0) [right=of lt] {$i$}
        edge node [swap] {$e_{i-1}$} (lt);
    \node[empty] (rt) [right=of v0] {}
        edge node [swap] {$e_i$} (v0);

    \node [vertex] (v1) [below=of v0] {$i$}
        edge node {$f_i$} (v0);
    \node[empty] (lb) [left=of v1] {}
        edge node [swap] {$e_{i-1}$} (v1);
    \node[empty] (rb) [right=of v1] {}
        edge node {$e_{i}$} (v1);

    \node (m) at ($(rt)+(1,0)+.5*(v1) - .5*(v0)$) {$\surg$};
    
    \node[empty] (lt') at ($(rt)+(2,0)$) {};
    \node[empty] (lb') at ($(lt')+(lb)-(lt)$) {}
        edge [out=0,in=0] node [swap] {$e_{i-1}$} (lt');

    \node[empty] (rt') at ($(lt') + (rt) - (lt)$) {};
    \node[empty] (rb') at ($(rt') + (rb) - (rt)$) {}
        edge [out=180,in=180] node {$e_i$} (rt');
    
\end{tikzpicture}
    \end{equation*}
    to get a character-equivalent picture of smaller size, contradicting
    minimality. We conclude that we must have $\{h(e),h(a_j),h(b_j)\} =
    \{e_{i-1}, e_{i}, f_i\}$ for $j=0,1$. Thus we can perform the surgery where
    we make $v_j$ incident to $e$, $a_j$, and $b_j$, for $j=0,1$, (shown below,
    again with the assumption that $h(a_0) = e_i$) 
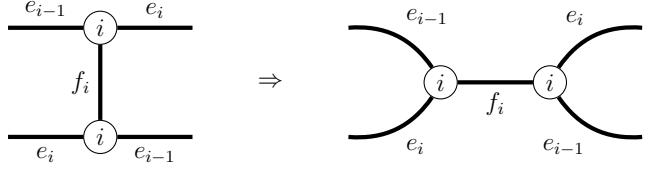
    \begin{equation*}
        \begin{tikzpicture}[auto,vertex/.style={circle,draw,thin,inner sep=2.5},every node/.style={scale=.8}, 
                    empty/.style={inner sep=0},ultra thick]

    \node[empty] (lt) {};
    \node [vertex] (v0) [right=of lt] {$i$}
        edge node [swap] {$e_{i-1}$} (lt);
    \node[empty] (rt) [right=of v0] {}
        edge node [swap] {$e_i$} (v0);

    \node [vertex] (v1) [below=of v0] {$i$}
        edge node {$f_i$} (v0);
    \node[empty] (lb) [left=of v1] {}
        edge node [swap] {$e_{i}$} (v1);
    \node[empty] (rb) [right=of v1] {}
        edge node {$e_{i-1}$} (v1);

    \node (m) at ($(rt)+(1,0)+.5*(v1) - .5*(v0)$) {$\surg$};
    
    \node[empty] (lt') at ($(rt)+(2,0)$) {};
    \node[empty] (lb') at ($(lt')+(lb)-(lt)$) {};
    \node[vertex] (v0') at ($(lt')+.5*(v1)-.5*(v0)+(v0)-(lt)$) {$i$}
        edge [bend right] node [swap] {$e_{i-1}$} (lt')
        edge [bend left] node {$e_i$} (lb');

    \node[vertex] (v1') [right=of v0'] {$i$}
        edge node {$f_i$} (v0');

    \node[empty] (rt') at ($(v1')+(v0')-(lb')$) {}
        edge [bend right] node [swap] {$e_i$} (v1');

    \node[empty] (rb') at ($(v1')+(v0')-(lt')$) {}
        edge [bend left] node {$e_{i-1}$} (v1');

\end{tikzpicture}
    \end{equation*}
    to get a new picture $\mcP_1$. Since the surgery did not change the number
    or labels of vertices, $\mcP_1$ will be character-minimal and
    character-equivalent to $\mcP_0$. Since $\bd(\mcP_0)$ does not contain
    any edges from $\mcC$, the boundary of an outer quadrilateral of $\mcP_0$
    must contain edges $e_1,e_2,e_3$, where $h(e_2) \in \mcC$, and $e_1$ and
    $e_3$ are the two unique edges of $h^{-1}(\mcU)$ which are incident with
    the endpoints of $e_2$. Since $e$ is not incident to the boundary, the
    edges $a_0$, $a_1$, $b_0$, and $b_1$ and the vertices $v_0$ and $v_1$ do
    not belong to an outer quadrilateral of $\mcP_0$, and hence every
    outer quadrilateral of $\mcP_0$ will be an outer quadrilateral of
    $\mcP_1$. 

    Recall that $e$ lies in $\mcD$. Without loss of generality, we can assume
    that $v_0$ belongs to $C$. By Lemma \ref{L:retractcycle}, $v_1$ must also
    belong to a $\mcC$-cycle. This creates two possible outcomes for
    $NE(\mcP_1)$.  First, suppose that both $v_0$ and $v_1$ lie on $C$. Then
    the surgery will ``pinch'' the region bounded by $C$, creating two
    $\mcC$-cycles $C_0$ and $C_1$ connected by $e$, as shown below. 
    
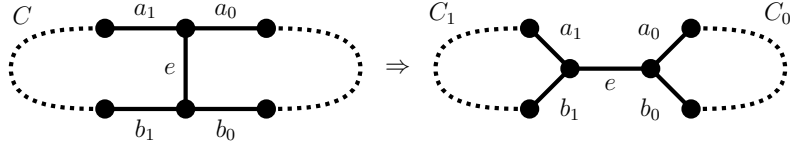
\begin{equation*}
        \begin{tikzpicture}[auto,ultra thick,scale=.5,empty/.style={inner sep=0},every node/.style={scale=.8}]

    \node[empty] (0) {};
    \node[empty] (1) [right=of 0] {}
        edge node [swap] {$a_1$} (0);
    \node[empty] (2) [right=of 1] {}
        edge node [swap] {$a_0$} (1);

    \node[empty] (3) [below=of 0] {};
    \node[empty] (4) [right=of 3] {}
        edge node {$b_1$} (3)
        edge node {$e$} (1);
    \node[empty] (5) [right=of 4] {}
        edge node {$b_0$} (4);

    \draw [dotted] (0) to [out=180,in=90] node[swap] {$C$} ($(0)+(-2.5,0)+.5*(4)-.5*(1)$) to [out=-90,in=180] (3);
    \draw [dotted] (2) to [out=0,in=90] ($(2)+(2.5,0)+.5*(4)-.5*(1)$) coordinate (r) to [out=-90,in=0] (5);

    \node (m) at ($(r)+(1,0)$) {$\surg$};

    \coordinate (l) at ($(m)+(1,0)$);
    
    \node[empty] (6) at ($(l)+(r)-(5)$) {};
    \node[empty] (7) at ($(6)+(3)-(0)$) {};
    \draw [dotted] (6) to [out=180,in=90] node [swap] {$C_1$} (l) to [out=-90,in=180] (7);
    \node[empty] (8) at ($.5*(6)+.5*(7)+.25*(2)-.25*(0)$) {}
        edge node [swap] {$a_1$} (6)
        edge node {$b_1$} (7);
    \node[empty] (9) at ($(8)+.5*(2)-.5*(0)$) {}
        edge node {$e$} (8);
    \node[empty] (10) at ($(6)+(2)-(0)$) {}
        edge node [swap] {$a_0$} (9);
    \node[empty] (11) at ($(7)+(2)-(0)$) {}
        edge node {$b_0$} (9);

    \draw [dotted] (10) to [out=0,in=90] node {$C_0$} ($(9)-(l)+(8)$) to [out=-90,in=0] (11);

    \foreach \x in {0,...,11}
        \draw [fill] (\x) circle [radius=.2];
    
\end{tikzpicture}
    \end{equation*}
    Any edge $e' \neq e$ in $\mcD$ will end up in the region bounded by 
    $C_0$ or the region bounded by $C_1$, so $NE(C,\mcP_0) = NE(C_1,\mcP_1) +
    NE(C_2,\mcP_1) + 1$. 

    The other possibility is that $v_1$ lies on a different $\mcC$-cycle $C'$,
    where (since $\mcC$-cycles cannot cross) $C'$ lies in the interior of the
    region bounded by $C$. In this case, the surgery will connect $C$ and $C'$
    to form a new cycle $C''$, as shown below.
    
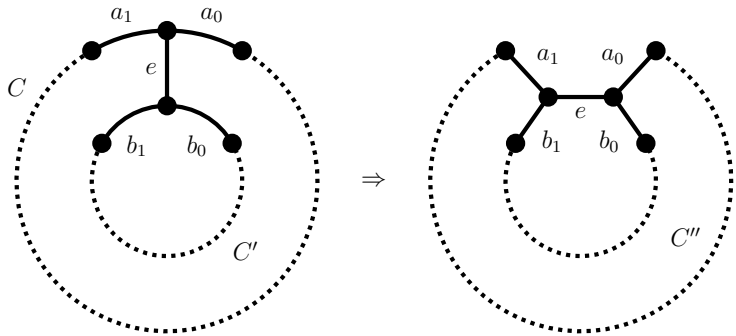
\begin{equation*}
        \begin{tikzpicture}[auto,ultra thick,scale=.5,empty/.style={inner sep=0},every node/.style={scale=.8}]

    \coordinate (O);

    \node[empty] (1) at ($(O)+(0,2)$) {};
    \draw (1) arc [start angle=90, end angle=150, radius=2] node [empty] (0) {};
    \draw (1) arc [start angle=90, end angle=30, radius=2] node [empty] (2) {};
    \draw [dotted] (0) arc [start angle=150, end angle=390, radius=2] node[pos=.7,swap] {$C'$} ;

    \path (1) arc [start angle=90, end angle=150, radius=2] node [pos=.7] {$b_1$};
    \path (1) arc [start angle=90, end angle=30, radius=2] node [swap,pos=.7] {$b_0$};

    \node[empty] (4) at ($(1)+(0,2)$) {}
        edge node [swap] {$e$} (1);
    \draw (4) arc [start angle=90, end angle=120, radius=4] node [empty] (3) {};
    \draw (4) arc [start angle=90, end angle=60, radius=4] node [empty] (5) {};
    \draw [dotted] (3) arc [start angle=120, end angle=420, radius=4] node[pos=.1,swap] {$C$} ;

    \path (4) arc [start angle=90, end angle=120, radius=4] node [swap,pos=.3] {$a_1$};
    \path (4) arc [start angle=90, end angle=60, radius=4] node [pos=.3] {$a_0$};

    \node (m) at ($(O)+(5.5,0)$) {$\surg$};

    \coordinate (O2) at ($(O)+(11,0)$);

    \coordinate (7') at ($(O2)+(0,2)$) {};
    \path (7') arc [start angle=90, end angle=150, radius=2] node [empty] (6) {};
    \path (7') arc [start angle=90, end angle=30, radius=2] node [empty] (8) {};

    \coordinate (10') at ($(7')+(0,2)$) {};
    \path (10') arc [start angle=90, end angle=120, radius=4] node [empty] (9) {};
    \path (10') arc [start angle=90, end angle=60, radius=4] node [empty] (11) {};

    \draw [dotted] (9) arc [start angle=120, end angle=420, radius=4] node[pos=.7] {$C''$};
    \draw [dotted] (6) arc [start angle=150, end angle=390, radius=2];

    \node[empty] (7) at ($.5*(6)+.5*(9)+.25*(11)-.25*(9)$) {}
        edge node {$b_1$} (6)
        edge node [swap] {$a_1$} (9);
    \node[empty] (10) at ($.5*(8)+.5*(11)+.25*(9)-.25*(11)$) {}
        edge node {$e$} (7)
        edge node {$a_0$} (11)
        edge node [swap] {$b_0$} (8);

    \foreach \x in {0,...,11}
        \draw[fill] (\x) circle [radius=.2];
 
\end{tikzpicture}
    \end{equation*}
    The edge $e$, along with all the edges inside the cycle $C'$, will end up
    on the outside of $C''$. The only edges remaining in the region bounded by
    $C''$ are edges that belonged to $\mcD$, so $NE(\mcC'',\mcP_1) =
    NE(C,\mcP_0) - NE(C',\mcP_0) -1$.  In both cases, no other cycles are
    changed by the surgery, so we conclude that $NE(\mcP_0) > NE(\mcP_1)$. 

    Iterating this procedure, we get a sequence $\mcP_0,\mcP_1,\mcP_2,\ldots$
    of character-minimal pictures, all character-equivalent, such that
    the outer quadrilaterals of $\mcP_i$ are outer quadrilaterals of
    $\mcP_{i+1}$, and $NE(\mcP_i) > NE(\mcP_{i+1})$. Since $NE(\mcP_i)$ cannot
    decrease indefinitely, this process must terminate at a picture $\mcP_n$
    with the property that if $e \in h^{-1}(\mcU)$ is incident with a
    $\mcC$-cycle $C$, then $e$ is not contained in the region bounded by $C$.
    Equivalently, we can say that $\res(\mcP_n,\mcD)$ is closed for every
    simple region $\mcD$ bounded by a $\mcC$-cycle. 

    Let $\mcP'$ be the picture $\mcP_n$ with all closed loops deleted. By Lemma
    \ref{L:closedloops}, $\mcP'$ is character-minimal and character-equivalent
    to $\mcP_n$ (and hence $\mcP_0$). In addition, it is easy to see that all
    outer quadrilaterals of $\mcP_n$ will be outer quadrilaterals of
    $\mcP'$, and that $\res(\mcP',\mcD)$ will be closed for every simple region
    $\mcD$ bounded by a $\mcC$-cycle in $\mcP'$. Suppose $\mcD$ is a simple
    region bounded by a $\mcC$-cycle such that $\res(\mcP',\mcD)$ is non-empty.
    Since $\mcP'$ does not contain any closed loops, $\res(\mcP',\mcD)$ must
    contain a vertex. But since $\res(\mcP',\mcD)$ is closed, we must have
    $\ch(\res(\mcP',\mcD)) = 0$ by Lemma \ref{L:suncycle1}, and hence we can
    delete $\res(\mcP',\mcD)$ from $\mcP'$ to get a character-equivalent
    picture of smaller size, contradicting the minimality of $\mcP'$. We
    conclude that $\res(\mcP',\mcD)$ must be empty for every simple region
    bounded by a $\mcC$-cycle, and hence every $\mcC$-cycle in $\mcP'$ is
    facial. By Proposition \ref{P:facialcover}, every $\mcC$-cycle in $\mcP'$
    is a facial cover, as required. 
\end{proof}

The following collorary is not needed for the proof of Theorem
\ref{T:invembedding}, but it does serve as a good example of how we will apply
Proposition \ref{P:suncycle} in the following section.
\begin{cor}\label{C:suncycle}
    Let $\mcH$ be the sun of size $n$, let $b : V(\mcH) \arr \Z_2$ be a vertex
    labelling function, and set $b_0 := \sum_{i=1}^n b_i$. Then the subgroup of
    $\Gamma = \Gamma(\mcH,b)$ generated by $S = \{x_{f_1},\ldots,x_{f_n}\}$ is
    isomorphic to $K = \Inv\langle f_1,\ldots,f_n: f_1 \cdots f_n =
    J^{b_0}\rangle$. 
\end{cor} 
\begin{proof}
    Clearly there is an $\mcH$-picture $\mcP$ with $\bd(\mcP) = f_1 \cdots f_n$
    and $\ch(\mcP) = (1,\ldots,1)$, so $x_{f_1} \cdots x_{f_n} = J^{b_0}$ in
    $\Gamma$ by Proposition \ref{P:vankampen2}. 

    Conversely, suppose that $x_{f_{i_1}} \cdots x_{f_{i_k}} = J^c$ holds in
    $\Gamma$. By Proposition \ref{P:vankampen2} again, there is an
    $\mcH$-picture $\mcP$ with $\bd(\mcP) = f_{i_1} \cdots f_{i_k}$ and
    $\ch(\mcP) \cdot b = c$. Let $\mfP$ be the set of
    character-minimal pictures which are character-equivalent to $\mcP$, and in
    which every $\mcC$-cycle is a facial cover. By Proposition
    \ref{P:suncycle}, $\mfP$ is non-empty, and since all elements
    of $\mfP$ have the same number of vertices,
    $\Cycle(\mcP',\mcC)$ is bounded for $\mcP' \in \mfP$.  Let $\mcP'$ be an
    element of $\mfP$ which maximizes $\Cycle(\mcP',\mcC)$.  If $C$ is a
    $\mcC$-cycle in $\mcP'$ which is not a copy of $\mcC$, then it
    is possible to cut $C$ into two $\mcC$-cycles as shown below, where
    the interior of $C$ is a face of $\mcP$.
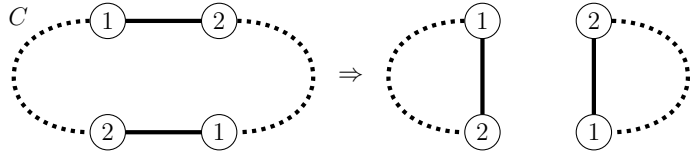
    \begin{equation*}
        \begin{tikzpicture}[auto,ultra thick,scale=.5,empty/.style={inner sep=0},every node/.style={scale=.8},vertex/.style={circle,draw,thin,inner sep=2.5}]

    \node[vertex] (1) {$1$};
    \node[vertex] (2) [right=of 1] {$2$}
        edge (1);

    \node[vertex] (3) [below=of 1] {$2$};
    \node[vertex] (4) [right=of 3] {$1$}
        edge (3);

    \draw [dotted] (1) to [out=180,in=90] node[swap] {$C$} ($(1)+(-2.5,0)+.5*(3)-.5*(1)$) to [out=-90,in=180] (3);
    \draw [dotted] (2) to [out=0,in=90] ($(2)+(2.5,0)+.5*(3)-.5*(1)$) coordinate (r) to [out=-90,in=0] (4);

    \node (m) at ($(r)+(1,0)$) {$\surg$};

    \coordinate (l) at ($(m)+(1,0)$);
    
    \node[vertex] (5) at ($(l)+(r)-(4)$) {$1$};
    \node[vertex] (6) [right=of 5] {$2$};
    \node[vertex] (7) [below=of 5] {$2$}
        edge (5);
    \node[vertex] (8) [right=of 7] {$1$}
        edge (6);

    \draw [dotted] (5) to [out=180,in=90] (l) to [out=-90,in=180] (7);
    \draw [dotted] (6) to [out=0,in=90] ($(6)-(l)+(7)$) to [out=-90,in=0] (8);

\end{tikzpicture}
    \end{equation*}
    Since this surgery does not change the character or modify any other
    $\mcC$-cycle in $\mcP$, we get an element of $\mfP$ with more
    $\mcC$-cycles than $\mcP'$, a contradiction. We conclude that every
    $\mcC$-cycle in $\mcP'$ is a facial copy. 

    Since $\mcC$ is cubic, each $\mcC$-cycle in $\mcP'$ must bound a unique
    face, even if $\mcP$ is a picture in a sphere. Let $\mcP''$ be the picture
    constructed by contracting each one of these faces to a vertex, as shown
    below (for $n=6$). 
    
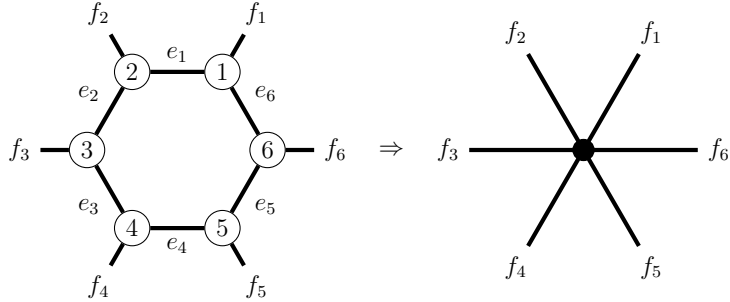
\begin{equation*}
        \begin{tikzpicture}[auto,ultra thick,vertex/.style={circle,draw,thin,inner sep=2.5},every node/.style={scale=.8},
                    hyperedge/.style={thin,pattern=north west lines},scale=.6]

    \begin{scope}
        \foreach \x in {1,...,6} {
            \node[vertex] (\x) at (60*\x:2) {$\x$};
            \node (\x') at (60*\x:3.5) {$f_\x$};
            \draw (\x) to (\x');
        }

        \draw (1) -- node [swap] {$e_1$} (2);
        \draw (2) -- node [swap] {$e_2$} (3);
        \draw (3) -- node [swap] {$e_3$} (4);
        \draw (4) -- node [swap] {$e_4$} (5);
        \draw (5) -- node [swap] {$e_5$} (6);
        \draw (6) -- node [swap] {$e_6$} (1);
    \end{scope}

    \node at (4.75,0) {$\surg$};

    \begin{scope}[shift={+(9,0)}]
        \foreach \x in {1,...,6} {
            \node (\x') at (60*\x:3) {$f_\x$};
            \draw (0,0) to (\x');
        }
        \draw[fill] (0,0) circle [radius=.2];
    \end{scope}
\end{tikzpicture}
    \end{equation*}
    Since every edge in $\mcP'$ labelled by an $e_i$ must belong to a
    $\mcC$-cycle by Lemma \ref{L:retractcycle}, every edge remaining in
    $\mcP''$ is labelled by an $f_i$ for some $1 \leq i \leq n$. If we label
    each vertex by the single relation $J^{b_0} f_1 \cdots f_n$, then $\mcP''$
    is a $K$-picture with $\bd(\mcP'') = f_{i_1}\cdots f_{i_k}$ and
    $\sign(\mcP'') = c$. By Proposition \ref{P:vankampen1}, the relation
    $f_{i_1} \cdots f_{i_k} = J^c$ holds in $K$.
\end{proof}
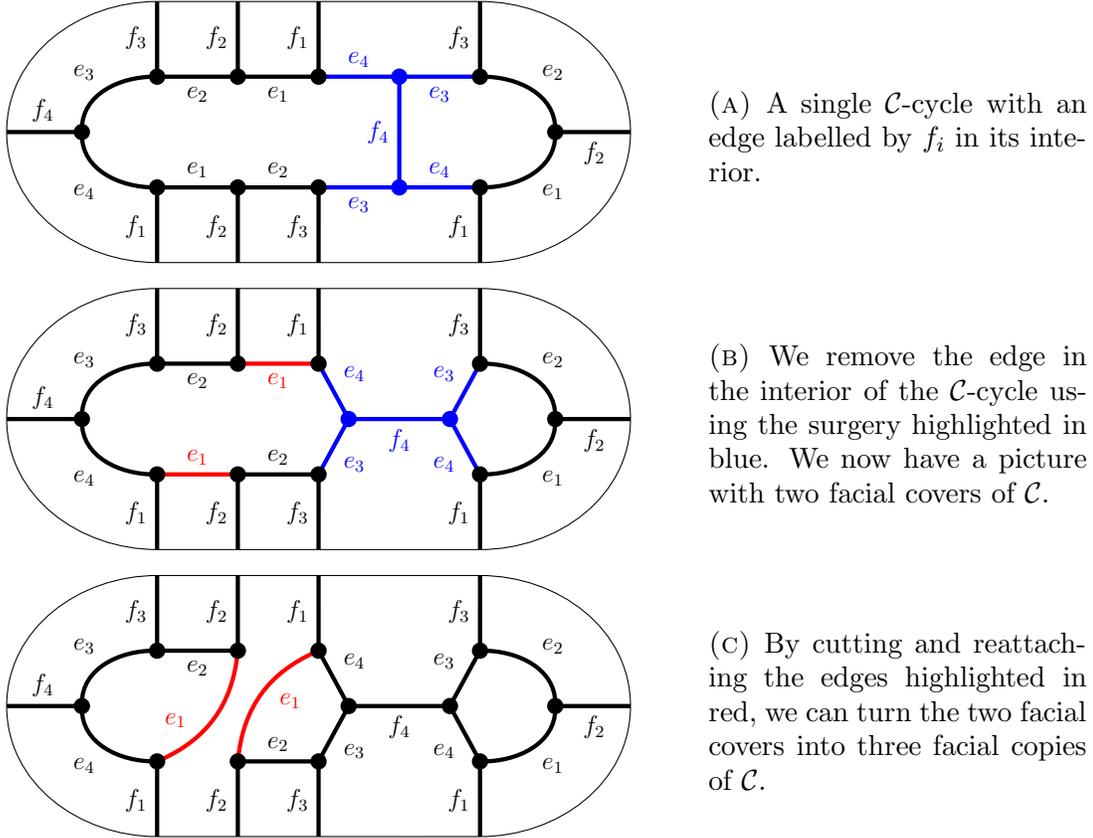
\begin{figure}
    \begin{subfigure}{1.2\textwidth}
        \begin{minipage}{.55\textwidth}
            \begin{tikzpicture}[auto,ultra thick,scale=.4,empty/.style={inner sep=0},every node/.style={scale=.8}]

    \node[empty] (0) {};
    \node[empty] (1) [right=of 0] {}
        edge node {$e_2$} (0);
    \node[empty] (2) [right=of 1] {}
        edge node {$e_1$} (1);
    \node[empty] (3) [right=of 2] {}
        edge[blue] node[swap] {$e_4$} (2);
    \node[empty] (4) [right=of 3] {}
        edge[blue] node {$e_3$} (3);

    \node[empty] (5) [below=1.4cm of 0] {};
    \node[empty] (6) [right=of 5] {}
        edge node[swap] {$e_1$} (5);
    \node[empty] (7) [right=of 6] {}
        edge node[swap] {$e_2$} (6);
    \node[empty] (8) [right=of 7] {}
        edge[blue] node {$e_3$} (7)
        edge[blue] node {$f_4$} (3);
    \node[empty] (9) [right=of 8] {}
        edge[blue] node[swap] {$e_4$} (8);
    
    \node[empty] (10) at ($.5*(0) + .5*(5) + (-2.5,0)$) {}
        edge [out=90,in=180] node {$e_3$} (0) 
        edge [out=-90,in=180] node [swap] {$e_4$} (5); 
    \node[empty] (11) at ($.5*(4) + .5*(9) + (2.5,0)$) {} 
        edge [out=90,in=0] node [swap] {$e_2$} (4) 
        edge [out=-90,in=0] node {$e_1$} (9); 
    
    \foreach \x in {0,1,2,4,5,6,7,9,10,11}
        \draw [fill] (\x) circle [radius=.2];

    \foreach \x in {3,8}
        \draw[fill,blue] (\x) circle [radius=.2];

    \coordinate (0') at ($(0)+(90:2.5)$)
        edge node [swap] {$f_3$} (0); 
    \coordinate (1') at ($(1)+(90:2.5)$)
        edge node [swap] {$f_2$} (1); 
    \coordinate (2') at ($(2)+(90:2.5)$)
        edge node [swap] {$f_1$} (2); 
    \coordinate (4') at ($(4)+(90:2.5)$)
        edge node [swap] {$f_3$} (4); 

    \coordinate (5') at ($(5)+(-90:2.5)$)
        edge node {$f_1$} (5); 
    \coordinate (6') at ($(6)+(-90:2.5)$)
        edge node {$f_2$} (6); 
    \coordinate (7') at ($(7)+(-90:2.5)$)
        edge node {$f_3$} (7); 
    \coordinate (9') at ($(9)+(-90:2.5)$)
        edge node {$f_1$} (9); 

    \coordinate (10') at ($(10)+(180:2.5)$)
        edge node {$f_4$} (10); 
    \coordinate (11') at ($(11)+(0:2.5)$)
        edge node {$f_2$} (11); 

    \draw[thin] (0') to (1') to (2') to (4') 
        to [out=0,in=90] (11') to [out=-90,in=0] (9')
        to (7') to (6') to (5') 
        to [out=180,in=-90] (10') to [out=90,in=180] (0');
\end{tikzpicture}
        \end{minipage}
        \begin{minipage}{.3\textwidth}
            \centering
            \caption{A single $\mcC$-cycle with an edge labelled by $f_i$ in its interior.}
            \label{F:sunexample1}
        \end{minipage}
        \vspace{.1in}
    \end{subfigure}

    \begin{subfigure}{1.2\textwidth}
        \begin{minipage}{.55\textwidth}
            \begin{tikzpicture}[auto,ultra thick,scale=.4,empty/.style={inner sep=0},every node/.style={scale=.8}]

    \node[empty] (0) {};
    \node[empty] (1) [right=of 0] {}
        edge node {$e_2$} (0);
    \node[empty] (2) [right=of 1] {}
        edge[red] node {$e_1$} (1);
    \node[empty] (p) [right=of 2] {};
    \node[empty] (4) [right=of p] {};

    \node[empty] (5) [below=1.4cm of 0] {};
    \node[empty] (6) [right=of 5] {}
        edge[red] node[swap] {$e_1$} (5);
    \node[empty] (7) [right=of 6] {}
        edge node[swap] {$e_2$} (6);
    \node[empty] (p) [right=of 7] {};
    \node[empty] (9) [right=of p] {};
    
    \node[empty] (10) at ($.5*(0) + .5*(5) + (-2.5,0)$) {}
        edge [out=90,in=180] node {$e_3$} (0) 
        edge [out=-90,in=180] node [swap] {$e_4$} (5); 
    \node[empty] (11) at ($.5*(4) + .5*(9) + (2.5,0)$) {} 
        edge [out=90,in=0] node [swap] {$e_2$} (4) 
        edge [out=-90,in=0] node {$e_1$} (9); 

    \node[empty] (3) at ($.5*(2) + .5*(7) + (1,0)$) {}
        edge[blue] node[swap] {$e_4$} (2)
        edge[blue] node {$e_3$} (7);
    \node[empty] (8) at ($.5*(4) + .5*(9) + (-1,0)$) {}
        edge[blue] node {$e_3$} (4)
        edge[blue] node[swap] {$e_4$} (9)
        edge[blue] node {$f_4$} (3);
    
    \foreach \x in {0,1,2,4,5,6,7,9,10,11}
        \draw [fill] (\x) circle [radius=.2];

    \foreach \x in {3,8}
        \draw[fill,blue] (\x) circle [radius=.2];

    \coordinate (0') at ($(0)+(90:2.5)$)
        edge node [swap] {$f_3$} (0); 
    \coordinate (1') at ($(1)+(90:2.5)$)
        edge node [swap] {$f_2$} (1); 
    \coordinate (2') at ($(2)+(90:2.5)$)
        edge node [swap] {$f_1$} (2); 
    \coordinate (4') at ($(4)+(90:2.5)$)
        edge node [swap] {$f_3$} (4); 

    \coordinate (5') at ($(5)+(-90:2.5)$)
        edge node {$f_1$} (5); 
    \coordinate (6') at ($(6)+(-90:2.5)$)
        edge node {$f_2$} (6); 
    \coordinate (7') at ($(7)+(-90:2.5)$)
        edge node {$f_3$} (7); 
    \coordinate (9') at ($(9)+(-90:2.5)$)
        edge node {$f_1$} (9); 

    \coordinate (10') at ($(10)+(180:2.5)$)
        edge node {$f_4$} (10); 
    \coordinate (11') at ($(11)+(0:2.5)$)
        edge node {$f_2$} (11); 

    \draw[thin] (0') to (1') to (2') to (4') 
        to [out=0,in=90] (11') to [out=-90,in=0] (9')
        to (7') to (6') to (5') 
        to [out=180,in=-90] (10') to [out=90,in=180] (0');
\end{tikzpicture}
        \end{minipage}
        \begin{minipage}{.3\textwidth}
            \centering
            \caption{We remove the edge in the interior of the $\mcC$-cycle using
                    the surgery highlighted in blue. We now have a picture with two
                    facial covers of $\mcC$.}
            \label{F:sunexample2}
        \end{minipage}
        \vspace{.1in}
    \end{subfigure}

    \begin{subfigure}{1.2\textwidth}
        \begin{minipage}{.55\textwidth}
            \begin{tikzpicture}[auto,ultra thick,scale=.4,empty/.style={inner sep=0},every node/.style={scale=.8}]

    \node[empty] (0) {};
    \node[empty] (1) [right=of 0] {}
        edge node {$e_2$} (0);
    \node[empty] (2) [right=of 1] {};
    \node[empty] (p) [right=of 2] {};
    \node[empty] (4) [right=of p] {};

    \node[empty] (5) [below=1.4cm of 0] {}
        edge [red,bend right] node[pos=.3] {$e_1$} (1);
    \node[empty] (6) [right=of 5] {}
        edge[red,bend left] node[swap,pos=.6] {$e_1$} (2);
    \node[empty] (7) [right=of 6] {}
        edge node[swap] {$e_2$} (6);
    \node[empty] (p) [right=of 7] {};
    \node[empty] (9) [right=of p] {};
    
    \node[empty] (10) at ($.5*(0) + .5*(5) + (-2.5,0)$) {}
        edge [out=90,in=180] node {$e_3$} (0) 
        edge [out=-90,in=180] node [swap] {$e_4$} (5); 
    \node[empty] (11) at ($.5*(4) + .5*(9) + (2.5,0)$) {} 
        edge [out=90,in=0] node [swap] {$e_2$} (4) 
        edge [out=-90,in=0] node {$e_1$} (9); 

    \node[empty] (3) at ($.5*(2) + .5*(7) + (1,0)$) {}
        edge node[swap] {$e_4$} (2)
        edge node {$e_3$} (7);
    \node[empty] (8) at ($.5*(4) + .5*(9) + (-1,0)$) {}
        edge node {$e_3$} (4)
        edge node[swap] {$e_4$} (9)
        edge node {$f_4$} (3);
    
    \foreach \x in {0,...,11}
        \draw [fill] (\x) circle [radius=.2];

    \coordinate (0') at ($(0)+(90:2.5)$)
        edge node [swap] {$f_3$} (0); 
    \coordinate (1') at ($(1)+(90:2.5)$)
        edge node [swap] {$f_2$} (1); 
    \coordinate (2') at ($(2)+(90:2.5)$)
        edge node [swap] {$f_1$} (2); 
    \coordinate (4') at ($(4)+(90:2.5)$)
        edge node [swap] {$f_3$} (4); 

    \coordinate (5') at ($(5)+(-90:2.5)$)
        edge node {$f_1$} (5); 
    \coordinate (6') at ($(6)+(-90:2.5)$)
        edge node {$f_2$} (6); 
    \coordinate (7') at ($(7)+(-90:2.5)$)
        edge node {$f_3$} (7); 
    \coordinate (9') at ($(9)+(-90:2.5)$)
        edge node {$f_1$} (9); 

    \coordinate (10') at ($(10)+(180:2.5)$)
        edge node {$f_4$} (10); 
    \coordinate (11') at ($(11)+(0:2.5)$)
        edge node {$f_2$} (11); 

    \draw[thin] (0') to (1') to (2') to (4') 
        to [out=0,in=90] (11') to [out=-90,in=0] (9')
        to (7') to (6') to (5') 
        to [out=180,in=-90] (10') to [out=90,in=180] (0');
\end{tikzpicture}
        \end{minipage}
        \begin{minipage}{.3\textwidth}
            \centering
            \caption{By cutting and reattaching the edges highlighted in red, we can
                     turn the two facial covers into three facial copies of $\mcC$.} 
            \label{F:sunexample3}
        \end{minipage}
    \end{subfigure}
    \caption{Using surgery, we turn a picture over the sun of size $4$ into an
        equivalent picture where all $\mcC$-cycles are facial copies.}
    \label{F:sunexample}
\end{figure}
The proof techniques of Proposition \ref{P:suncycle} and Corollary
\ref{C:suncycle} are illustrated by the following example.
\begin{example}
    Let $\mcH$ be the sun of size $4$, and let $K = \Inv\langle f_1,\ldots,f_4 :
    f_1 f_2 f_3 f_4 = 1 \rangle$. It is not hard to see that the relation $w :=
    f_1 f_2 f_3 f_4 (f_1 f_2 f_3)^2 = 1$ holds in $K$. Given a $K$-picture with
    boundary word $w$, we can replace every vertex with a $\mcC$-cycle to get
    an equivalent $\mcH$-picture, for instance as shown in Figure
    \ref{F:sunexample3}. However, there are $\mcH$-pictures with boundary word
    $w$, as shown in Figure \ref{F:sunexample1}, which do not come from a
    $K$-picture in this way. Nonetheless, if $\mcP$ is a minimal $\mcH$-picture
    whose boundary contains only $f_i$'s, then the proofs of Proposition
    \ref{P:suncycle} and Corollary \ref{C:suncycle} provide a method to
    transform $\mcP$ to an equivalent picture which does come from a $K$-picture.
    The transformation process is shown in Figure \ref{F:sunexample}.
\end{example}

To finish the section, we prove one more technical lemma:
\begin{lemma}\label{L:nosillycycles}
    Let $\mcP$ be a character-minimal picture with no closed loops over the sun
    $\mcH$ of size $n$.  Then there is no cycle in $\mcP$ with edge labels
    contained in $\{e_i,f_i,f_{i+1}\}$ for some fixed $1 \leq i \leq n$ (where
    $f_{n+1} := f_1$). 
\end{lemma}
\begin{proof}
    Suppose $C$ is a cycle of this form, and let $\mcU_i$ be the closure of the
    subhypergraph $\{f_i,f_{i+1}\}$. Since $\mcP$ has no closed loops, $C$ must
    have at least one vertex. Since $\mcH$ is simple, $C$ must have at least two
    vertices, and there cannot be two consecutive edges of $C$ with the same
    label. Thus $C[\mcU_i]$ is a (non-empty) matching between the vertices of
    $C$, where paired vertices must have the same label (either $i$ or $i+1$,
    depending on whether the edge connecting them is labelled by an $f_i$ or
    $f_{i+1}$).  Since every vertex of $C$ occurs in $C[\mcU_i]$, $C$ must
    consist of a sequence of vertices $v_{1},\ldots,v_{2n},v_{2n+1}=v_1$, such
    that $v_{2i-1}$ and $v_{2i}$ are connected by an edge $b_i$ with $h(b_i)
    \in \{f_i,f_{i+1}\}$, and $v_{2i}$ and $v_{2i+1}$ are connected by an edge
    labelled by $e_i$, $i=1,\ldots,n$. Let $a_i$ be the edge incident to
    $v_{i}$ which is not in $C$.  Since every vertex $v_i$ is incident to an
    edge of $C$ labelled by $e_i$, we conclude that
    \begin{equation*}
        h(a_i) = \begin{cases} e_{i-1} & h(v_i) = i \\
                               e_{i+1} & h(v_i) = i+1
                \end{cases}, 
    \end{equation*}
    where $e_{0} := e_n$ and $e_{n+1} := e_1$. But since
    $h(v_{2i-1})=h(v_{2i})$, this means that $h(a_{2i-1}) = h(a_{2i})$ for all
    $i=1,\ldots,n$. Thus if we delete all the vertices and edges of $C$, we can
    connect $a_{2i-1}$ and $a_{2i}$ along the path previously taken by $b_i$
    to get a new $\mcH$-picture $\mcP'$ which is character-equivalent to
    $\mcP$. But the size of $\mcP'$ would then be strictly smaller than the
    size of $\mcP$, contradicting minimality.
\end{proof}

\section{Stellar cycles and constellations}\label{S:constellation}

In this section, we extend the results for suns from the previous section to
more complicated hypergraphs through the notion of a \emph{constellation}.
The result is the ``constellation theorem'', which will be the key to the
proof of Theorem \ref{T:invembedding}.
\begin{defn}\label{D:stellar}
    Let $\mcH$ be a hypergraph with vertex labelling function $b : V(\mcH)
    \arr \Z_2$. A cycle $\mcC$ in $\mcH$ is \emph{$b$-stellar} if
    \begin{enumerate}[(a)]
        \item the neighbourhood $\mcN(\mcC)$ is isomorphic to a sun,
        \item $\mcN(\mcC)$ is a retract of $\mcH$, and
        \item $b_v = 0$ for all $v \in V(\mcC)$. 
    \end{enumerate}
\end{defn}

\begin{defn}\label{D:constellation}
    Let $\mcH$ be a hypergraph with vertex labelling $b : V(\mcH) \arr \Z_2$.
    A \emph{$b$-constellation} is a collection $\Phi$ of subhypergraphs of
    $\mcH$ satisfying the following properties:
    \begin{enumerate}[(a)]
        \item If $\mcC \in \Phi$, then $\mcC$ is a cycle, the neighbourhood
            $\mcN(\mcC)$ is isomorphic to a sun, and $\mcC$ is either:
            \begin{enumerate}[(i)]
                \item $b$-stellar, or
                \item a sequence of edges $e_1 e_2 \cdots e_n$ (in order), $n
                    \geq 3$, such that $e_k$ belongs to a $b$-stellar cycle $\mcC'
                    \in \Phi$ for all $3\leq k \leq n$.
            \end{enumerate}
        \item For every element $\mcC \in \Phi$, either:
            \begin{enumerate}[(i)]
                \item there is an edge $e$ in $\mcC$ which does not belong to
                    any cycle in $\Phi \setminus \{\mcC\}$, or
                \item there is another cycle $\mcC' \in \Phi$ such that
                    $E(\mcC) \cap E(\mcC') \neq \emptyset$, and $\mcC'$ contains
                    an edge $e$ which does not belong to any cycle in 
                    $\Phi \setminus \{\mcC'\}$.
            \end{enumerate}
        \item If $\mcC_0, \mcC_1 \in \Phi$, where $\mcC_0 \neq \mcC_1$, then
            $|E(\mcC_0) \cap E(\mcC_1)| \leq 1$, and if neither $\mcC_0$ or
            $\mcC_1$ is $b$-stellar, then $E(\mcC_0) \cap E(\mcC_1) = \emptyset$. 
    \end{enumerate}
    If $\Phi$ is a $b$-constellation, then a cycle $C$ in a picture $\mcP$ is a
    \emph{$\Phi$-cycle} if $C$ is a $\mcC$-cycle for some $\mcC \in \Phi$. 
\end{defn}
Roughly speaking, property (a) in Definition \ref{D:constellation} says that
every cycle $\mcC$ in $\Phi$ is either stellar or (mostly) covered by other
stellar cycles, while properties (b) and (c) state that the cycles in $\Phi$ do not
overlap too much. 
\begin{example}\label{Ex:constellation}
    Consider the hypergraph $\mcH$ shown in Figure \ref{F:cubewithtail}, and
    let $b$ be the vertex labelling function with $b_9=1$ and $b_v=0$ for $v
    \neq 9$.
    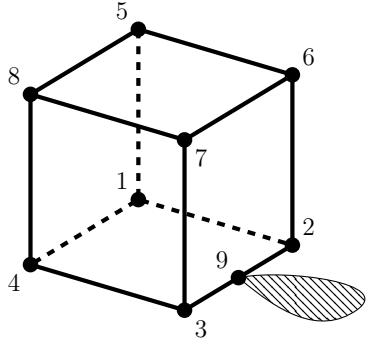
\begin{figure}
        \begin{tikzpicture}[auto,ultra thick,scale=.5,vertex/.style={circle,draw,thin,inner sep=2.5},every node/.style={scale=.8},
                    hyperedge/.style={thin,pattern=north west lines},empty/.style={inner sep=0}]

    \node[empty] (1) at (0.000000,0.000000) {};
    \node[empty] (2) at (4.095760,-1.212019) {};
    \node[empty] (3) at (1.227878,-2.942962) {};
    \node[empty] (4) at (-2.867882,-1.730943) {};
    \node[empty] (5) at (0.000000,4.531539) {};
    \node[empty] (6) at (4.095760,3.319520) {};
    \node[empty] (7) at (1.227878,1.588576) {};
    \node[empty] (8) at (-2.867882,2.800596) {};
    \node[empty] (9) at (2.662,-2.077) {};
    \node[empty] (10) at (5.734,-2.99) {};

    \draw[dashed] (1) -- (2); \draw[dashed] (1) -- (4); \draw[dashed] (1) -- (5);
    \draw (2) -- (3); \draw (2) -- (6);
    \draw (3) -- (4); \draw (3) -- (7);
    \draw (4) -- (8);
    \draw (5) -- (6); \draw (5) -- (8);
    \draw (6) -- (7); \draw (7) -- (8);

    \draw[hyperedge] (9.10) to [out=10,in=35] (10) to [out=-145,in=-40] (9.-40);

    \node[above left] at (1) {$1$};
    \node[above right] at (2) {$2$};
    \node[below right] at (3) {$3$};
    \node[below left] at (4) {$4$};
    \node[above left] at (5) {$5$};
    \node[above right] at (6) {$6$};
    \node[below right] at (7) {$7$};
    \node[above left] at (8) {$8$};
    \node[above left] at (9) {$9$};
    
    \foreach \x in {1,...,9}
        \draw[fill] (\x) circle [radius=.15];

\end{tikzpicture}
        \caption{A hypergraph based on Figure \ref{F:cube}, but where the open
            neighbourhood of each face is a sun.}
        \label{F:cubewithtail}
    \end{figure}
    Let $C_1$, $C_2$, and $C_3$ be the cycles with vertex sets $\{1,2,5,6\}$,
    $\{1,4,5,8\}$, and $\{3,4,7,8\}$ respectively, and let $C_4$ be the cycle
    with vertex set $\{1,2,3,4,9\}$. As in Example \ref{Ex:cube}, the
    neighbourhoods $\mcN(C_i)$,$i=1,2,3$, are retracts of $\mcH$, and hence
    $C_1$,$C_2$, and $C_3$ are $b$-stellar. The cycle $C_4$ is not a retract,
    nor is $b_v = 0$ for all $v \in V(C_4)$, so $C_4$ is not $b$-stellar. But
    $\mcN(C_4)$ is a sun, and the edges $12$, $14$, and $34$ belong to $C_1$,
    $C_2$, and $C_3$ respectively.  Thus $\Phi = \{C_1,\ldots,C_4\}$ is a
    $b$-constellation. 

    The cycle with vertex set $\{5,6,7,8\}$ is also $b$-stellar, and could be
    added to this constellation, but the cycle with vertex set $\{2,3,6,7,9\}$
    cannot be added since it is not $b$-stellar and shares edges with $C_4$. 
\end{example}
We can now state the main theorem of this section:
\begin{thm}[Constellation theorem]\label{T:constellation}
    Let $\mcH$ be a hypergraph with vertex labelling $b$, and let $\Phi$ be a
    $b$-constellation. Let $\mcP$ be an $\mcH$-picture such that:
    \begin{enumerate}[label={(p.\arabic*)}]
        \item\label{P1}  $\bd(\mcP)$ does not contain any edges from any cycle $\mcC \in \Phi$, and 
        \item\label{P2}  either $b=0$ or $\mcP$ is closed. 
    \end{enumerate}
    Then $\mcP$ is $b$-equivalent to a picture $\mcP'$ such that all
    $\Phi$-cycles in $\mcP'$ are facial copies. 
\end{thm}
The rest of this section is concerned with the proof of Theorem
\ref{T:constellation}. To aid the reader, the proof is split into a number of
lemmas, which are in turn grouped into subsections. In all the lemmas, $\mcH$
will be a hypergraph and $b$ will be a vertex labelling. We refer to the
hypotheses \ref{P1} and \ref{P2} of Theorem \ref{T:constellation} as necessary.

\subsection{Structure of constellations} We start by proving two lemmas about
constellations.
\begin{lemma}\label{L:struct1}
    If $\Phi$ is a $b$-constellation in $\mcH$, and $e$ is an edge of $\mcH$,
    then there are at most two cycles in $\Phi$ containing $e$.
\end{lemma}
\begin{proof}
    Suppose $\mcC$ is an element of $\Phi$ containing $e$. Since $\mcN(\mcC)$ is
    a sun, $\mcC$ is cubic. Thus if $v$ is an endpoint of $e$, then there is a
    unique edge $f$ incident to $v$ and not contained in $\mcC$. If $\mcC' \in
    \Phi \setminus \{\mcC\}$ contains $e$, then $\mcC'$ must also contain $f$,
    since $|E(\mcC) \cap E(\mcC')| \leq 1$. Hence if $\mcC'$ and $\mcC'' \in
    \Phi \setminus \{\mcC\}$ both contain $e$, then $\{e,f\} \subset E(\mcC')
    \cap E(\mcC'')$, and consequently $\mcC' = \mcC''$.
\end{proof}

\begin{lemma}\label{L:struct2}
    Let $\Phi$ be a $b$-constellation in $\mcH$, and let $\Phi' \subseteq \Phi$. 
    If all elements of $\Phi'$ are $b$-stellar or $\Phi'$ contains all $b$-stellar
    elements of $\Phi$, then $\Phi'$ is a $b$-constellation.
\end{lemma}
\begin{proof}
    Every subset of $\Phi$ satisfies part (c) of Definition
    \ref{D:constellation}.  If all elements of $\Phi'$ are $b$-stellar, or
    $\Phi'$ contains all $b$-stellar cycles in $\Phi$, then $\Phi'$ also
    satisfies part (a). If $e$ is an edge of $\mcC \in \Phi'$ such that $e$
    does not belong to any element of $\Phi \setminus \{\mcC\}$, then clearly
    $e$ does not belong to any element of $\Phi' \setminus \{\mcC\}$.  Suppose
    every edge of $\mcC \in \Phi'$ belongs to some element of $\Phi \setminus
    \{\mcC\}$. Then by definition there is a cycle $\mcC' \in \Phi$ such that
    $E(\mcC) \cap E(\mcC') \neq \emptyset$ and $\mcC'$ contains an edge $e$
    which does not belong to any cycle in $\Phi \setminus \{\mcC'\}$.  If
    $\mcC'$ belongs to $\Phi'$, then we are done. If $\mcC'$ does not belong to
    $\Phi'$, then let $e'$ be the unique edge of $E(\mcC) \cap E(\mcC')$.  By
    Lemma \ref{L:struct1}, the only cycles of $\Phi$ which contain $e'$ are
    $\mcC$ and $\mcC'$, and hence no element of $\Phi' \setminus
    \{\mcC\}$ contains $e'$. 
\end{proof}

\subsection{Stellar cycles}
Now we turn to the core of the argument: showing that $\mcC$-cycles can be
turned into facial covers if $\mcC$ is stellar. The proof relies on hypothesis
\ref{P2} in the following way: 
\begin{lemma}\label{L:hypp2}
    Suppose $\mcC$ is a $b$-stellar cycle in $\mcH$, and $\mcP$ is a
    $b$-minimal $\mcH$-picture satisfying hypothesis \ref{P2}. Then every
    $\mcC$-cycle in $\mcP$ bounds a simple region $\mcD$ such that
    $\ch(\germ(\mcP,\mcD)) \cdot b = 0$.
\end{lemma}
\begin{proof}
    If $b=0$, then the lemma is vacuously true. If $\mcP$ is closed and
    $\ch(\mcP) \cdot b = 0$, then $\mcP$ is $b$-equivalent to the
    empty picture. But since $\mcP$ is $b$-minimal, in this case $\mcP$
    must have size zero, and again the lemma is vacuously true. 

    Suppose that $\mcP$ is closed and $\ch(\mcP) \cdot b = 1$. Since $\mcP$ is
    closed, we can think of $\mcP$ as a picture in the sphere, in which case
    every $\mcC$-cycle bounds two simple regions $\mcD_1$ and $\mcD_2$. 
    Every vertex of $\mcP$ appears in one of $\germ(\mcP,\mcD_1)$ or
    $\germ(\mcP,\mcD_2)$, with only the vertices of $C$ appearing in both.
    Hence 
    \begin{equation*}
        \ch(\mcP) = \ch(\germ(\mcP,\mcD_1)) + \ch(\germ(\mcP,\mcD_2)) - \ch(C).
    \end{equation*}
    Since $\mcC$ is $b$-stellar, $b_v = 0$ for all $v \in V(\mcC)$, so
    $\ch(C) \cdot b = 0$. Consequently,
    \begin{equation*}
        \ch(\germ(\mcP,\mcD_1)) \cdot b + \ch(\germ(\mcP,\mcD_2)) \cdot b =
            \ch(\mcP) \cdot b = 1
    \end{equation*}
    and we conclude that one of $\ch(\germ(\mcP,\mcD_i)) \cdot b$, $i=1,2$
    must be $0$. 
\end{proof}
\begin{lemma}\label{L:stellarfacial}
    Let $\Phi$ be a $b$-constellation in which every cycle $\mcC \in \Phi$ is
    $b$-stellar. Let $\mcP$ be an $\mcH$ picture satisfying \ref{P1} and
    \ref{P2}. Then $\mcP$ is $b$-equivalent to a $b$-minimal picture $\mcP'$
    with no closed loops, such that every $\Phi$-cycle in $\mcP'$ is facial. 
\end{lemma}
\begin{proof}
    The proof is by induction on the size of $\Phi$. If $\Phi$ is empty, then
    the lemma is true by Definition \ref{D:minimal} and Lemma
    \ref{L:closedloops}. Suppose the lemma is true for all $b$-constellations
    of size $m$, where $m \geq 0$, and let $\Phi$ be a $b$-constellation of
    size $m+1$ in which every cycle is $b$-stellar. If we pick an element $\mcC
    \in \Phi$, then $\Phi' := \Phi \setminus \{\mcC\}$ will be a
    $b$-constellation of size $m$ by Lemma \ref{L:struct2}, so that every
    $\mcH$-picture satisfying \ref{P1} and \ref{P2} (with respect to $\Phi$)
    will be $b$-equivalent to some $b$-minimal picture $\mcP$, also satisfying
    \ref{P1} and \ref{P2}, such that every $\Phi'$-cycle in $\mcP$ is facial.
    Thus, to show that the lemma holds for $\Phi$, we can assume that we start
    with a picture $\mcP$ of this form.
    
    Let $NF(\mcP)$ denote the number of non-facial $\mcC$-cycles in
    $\mcP$, where $\mcC \in \Phi$ is the cycle chosen above. Similarly to the
    proof of Proposition \ref{P:suncycle}, our strategy will be to perform a
    sequence of surgeries, starting from $\mcP_0 := \mcP$, each of which
    decreases $NF(\mcP)$. Suppose that $\mcP$ has a non-facial
    $\mcC$-cycle $C$.  By Lemma \ref{L:hypp2}, there is a simple
    region $\mcD$ bounded by $C$ such that $\ch(\germ(\mcP,\mcD)) \cdot b = 0$.
    Let $f : \mcH \arr \mcN(\mcC)$ be a retract onto $\mcN(\mcC)$, so that
    $\widehat{\mcP} := f(\germ(\mcP,\mcD))$ is an $\mcN(\mcC)$-picture. Recall
    that, since $\mcN(\mcC)$ is open, an $\mcN(\mcC)$-picture like $\widehat{\mcP}$ can be
    regarded as an $\mcH$-picture.  By the definition of $\germ$, the labels
    of edges and vertices in the boundary of the outer faces of
    $\germ(\mcP,\mcD)$ must belong to $\mcN(\mcC)$. Since $f$ is a retract, we
    conclude that the closures of the outer faces of $\germ(\mcP,\mcD)$ will be
    identical to the closures of the outer faces of $\widehat{\mcP}$, and in
    particular, $\bd(\widehat{\mcP}) = \bd(\germ(\mcP,\mcD))$. The construction
    in Proposition \ref{P:morphismpic} relabels or deletes vertices not in $\mcN(\mcC)$,
    so $\widehat{\mcP}$ might not be character-equivalent to
    $\germ(\mcP,\mcD)$. But since $\mcC$ is $b$-stellar, $b_v=0$ for all $v \in
    V(\mcC)$, and consequently,
    $\ch(\widehat{\mcP}) \cdot b = 0$. We conclude that $\widehat{\mcP}$ is
    $b$-equivalent to $\germ(\mcP,\mcD)$.  Since the size of $\widehat{\mcP}$
    is less than or equal to the size of $\germ(\mcP,\mcD)$, and
    $\germ(\mcP,\mcD)$ is $b$-minimal, we also conclude that $\widehat{\mcP}$
    is $b$-minimal as an $\mcH$-picture.  It follows that $\widehat{\mcP}$ is
    $b$-minimal (and hence character-minimal) as an $\mcN(\mcC)$-picture. From,
    again, the definition of $\germ$, we know that $\bd(\germ(\mcP,\mcD)) =
    \bd(\widehat{\mcP})$ does not contain any edges from $\mcC$.  Finally,
    $\mcN(\mcC)$ is a sun, so we can apply Proposition \ref{P:suncycle} to
    $\widehat{\mcP}$ to get an $\mcN(\mcC)$-picture $\widehat{\mcP}'$ such that
    $\widehat{\mcP}'$ is character-equivalent to
    and of the same size as $\widehat{\mcP}$, every $\mcC$-cycle in
    $\widehat{\mcP}'$ is facial, every outer quadrilateral of
    $\widehat{\mcP}$ is an outer quadrilateral of $\widehat{\mcP}'$, and
    $\widehat{\mcP}'$ has no closed loops. In particular, $\widehat{\mcP}'$ is
    $b$-minimal and $b$-equivalent to $\germ(\mcP,\mcD)$, and every outer
    quadrilateral of $\germ(\mcP,\mcD)$ is an outer quadrilateral of
    $\widehat{\mcP}'$.

    Let $\mcP_1$ be the result of replacing $\germ(\mcP,\mcD)$ with
    $\widehat{\mcP}'$.  Clearly $\mcP_1$ is $b$-minimal and $b$-equivalent to
    $\mcP$.  Suppose $C'$ is a $\mcC'$-cycle in $\mcP_1$, where $\mcC' \in
    \Phi'$. Now $\bd(\mcP) = \bd(\mcP_1)$ does not contain any edges from the
    cycles of $\Phi$, so by Lemma \ref{L:retractcycle}, every edge of $\mcP$
    which is labelled by an edge of $\mcC'$ belongs to a unique $\mcC'$-cycle,
    and the same holds for $\mcP_1$. If $C'$ contains a boundary edge $e$ of
    $\widehat{\mcP}'$, then there is a $\mcC'$-cycle $C''$ in $\mcP$ which
    also contains $e$. By hypothesis, $C''$ is facial, and hence by Proposition
    \ref{P:boundary}, the edges of $\germ(C'',\mcD)$ form outer
    quadrilaterals in $\germ(\mcP,\mcD)$. But since the outer quadrilaterals
    of $\germ(\mcP,\mcD)$ are also outer quadrilaterals of
    $\widehat{\mcP}'$, the edges of $C''$ are unchanged in $\mcP_1$. Since $e$
    belongs to a unique $\mcC'$-cycle in $\mcP_1$, we must have $C' = C''$, and hence $C'$
    is facial in $\mcP_1$. On the other hand, if $C'$ does not contain a
    boundary edge of $\widehat{\mcP}'$, then $C'$ either does not intersect
    $\mcD$, or lies entirely inside of $\mcD$. In the former case, $C'$ will
    also be a facial $\mcC'$-cycle in $\mcP$, and will remain facial in $\mcP_1$.
    In the latter case, $C'$ would have to consist only of edges labelled by
    $E' = E(\mcC') \cap E(\mcN(\mcC))$.  Since $\mcN(\mcC)$ is a sun and
    $|E(\mcC') \cap E(\mcC)| \leq 1$, the intersection $E'$ is either empty, or
    is equal to (in the notation of Definition \ref{D:sun})
    $\{e_i,f_i,f_{i+1}\}$ for some $i$. But since $\widehat{\mcP}'$ is
    character-minimal and has no closed loops, Lemma \ref{L:nosillycycles} implies that
    $\widehat{\mcP}'$ does not contain any cycles of this form. We conclude
    that every $\mcC'$-cycle in $\mcP_1$ remains facial. Finally, every
    $\mcC$-cycle in $\mcP_1$ belongs either to $\widehat{\mcP}'$ or is
    inherited unchanged from $\mcP$, with the consequence that $NF(\mcP_1) <
    NF(\mcP)$.

    Iterating this procedure, we get a sequence $\mcP_0 = \mcP, \mcP_1,\ldots,
    \mcP_k$ of $b$-minimal pictures, all $b$-equivalent, such that all
    $\Phi'$-cycles in $\mcP_i$ are facial for $1 \leq i \leq k$, and such
    that all $\Phi$-cycles in $\mcP_k$ are facial. Deleting all closed loops
    from $\mcP_k$ finishes the proof.
\end{proof} 

\subsection{Covers versus copies}
For the next lemma, we show that if all $\Phi$-cycles are facial covers, then
we can turn $\Phi$-cycles into facial copies. 
\begin{lemma}\label{L:copycycle}
    Let $\Phi$ be a $b$-constellation. Suppose that $\mcP$ is an $\mcH$-picture
    satisfying \ref{P1}, and such that all $\Phi$-cycles in $\mcP$ are facial
    covers. Then $\mcP$ is character-equivalent to a picture $\mcP'$ in which
    all $\Phi$-cycles are facial copies. 
\end{lemma}
\begin{proof}
    The proof is similar to the proof of Corollary \ref{C:suncycle}. Let
    $\mfP$ be the set of pictures which are character-equivalent to $\mcP$,
    have the same size as $\mcP$, and in which all $\Phi$-cycles are facial
    covers. Since all elements of $\mfP$ have the same number of vertices,
    $\Cycle(\mcP',\Phi)$ is bounded across $\mcP' \in \mfP$. Let $\mcP'$ be an
    element of $\mfP$ which maximizes $\Cycle(\mcP',\Phi)$.

    Suppose that $C$ is a $\mcC$-cycle in $\mcP'$ which is not a copy of
    $\mcC$, where $\mcC \in \Phi$ has an edge $e$ which does not belong to any
    other cycle in $\Phi$. Since $C$ is a cover, there are two distinct edges
    $e_1$ and $e_2$ in $C$ with $h(e_1) = h(e_2) = e$. As in the proof of
    Corollary \ref{C:suncycle}, we can cut $C$ at $e_1$ and $e_2$ to get a new
    picture $\mcP''$, character-equivalent to $\mcP'$, in which $C$ has been
    replaced by two facial covers. By Lemma \ref{L:retractcycle} and the
    hypothesis on $e$, the edges $e_1$ and $e_2$ are not contained in any other
    $\Phi$-cycle.  Thus all other $\Phi$-cycles in $\mcP'$ are unchanged in $\mcP''$,
    and $\mcP''$ will be an element of $\mfP$ with $\Cycle(\mcP',\Phi) <
    \Cycle(\mcP'',\Phi)$, a contradiction. We conclude that if $\mcC \in \Phi$ has an
    edge which does not belong to any other cycle in $\Phi$, then all $\mcC$-cycles
    in $\mcP'$ are facial copies. 

\begin{figure}
    \begin{tikzpicture}[auto,ultra thick,scale=.5,empty/.style={inner sep=0},every node/.style={scale=.8},vertex/.style={circle,draw,thin,inner sep=2.5}]

    \node[vertex] (1) {$a$};
    \node[vertex] (2) [right=of 1] {$b$}
        edge node {$e$} (1);

    \node[vertex] (3) [below=of 1] {$b$};
    \node[vertex] (4) [right=of 3] {$a$}
        edge node {$e$} (3);

    \draw [dashed] (1) to [out=180,in=90] node[swap] {$C$} ($(1)+(-2.5,0)+.5*(3)-.5*(1)$) to [out=-90,in=180] (3);
    \draw [dashed] (2) to [out=0,in=90] ($(2)+(2.5,0)+.5*(3)-.5*(1)$) coordinate (r) to [out=-90,in=0] (4);

    \node (m) at ($(r)+(1,0)$) {$\surg$};

    \coordinate (l) at ($(m)+(1,0)$);
    
    \node[vertex] (5) at ($(l)+(r)-(4)$) {$a$};
    \node[vertex] (6) [right=of 5] {$b$};
    \node[vertex] (7) [below=of 5] {$b$}
        edge node {$e$} (5);
    \node[vertex] (8) [right=of 7] {$a$}
        edge node[swap] {$e$} (6);

    \draw [dashed] (5) to [out=180,in=90] (l) to [out=-90,in=180] (7);
    \draw [dashed] (6) to [out=0,in=90] ($(6)-(l)+(7)$) to [out=-90,in=0] (8);

    \node[empty] (9) [above=.6cm of 1] {}
        edge[dotted] (1);
    \node[empty] (10) [above=.6cm of 2] {}
        edge[dotted] (2);
    \node[empty] (11) at ($1.5*(9)+.5*(10)-(1)$)  {}
        edge[dotted,out=0,in=90] (10)
        edge[out=180,in=90] node {$f$} (9);

    \node[empty] (12) [below=.6cm of 3] {}
        edge[dotted] (3);
    \node[empty] (13) [below=.6cm of 4] {}
        edge[dotted] (4);
    \node[empty] (14) at ($1.5*(13)+.5*(12)-(4)$)  {}
        edge[dotted,out=180,in=-90] (12)
        edge[out=0,in=-90] node {$f$} (13);

    \node[empty] (15) [above=.6cm of 5] {}
        edge[dotted] (5);
    \node[empty] (16) [above=.6cm of 6] {}
        edge[dotted] (6);
    \node[empty] (17) at ($1.5*(15)+.5*(16)-(5)$)  {}
        edge[dotted,out=0,in=90] (16);

    \node[empty] (18) [below=.6cm of 7] {}
        edge[dotted] (7);
    \node[empty] (19) [below=.6cm of 8] {}
        edge[dotted] (8)
        edge[out=-90,in=-90] node[swap,pos=.8] {$f$} (17);
    \node[empty] (20) at ($1.5*(19)+.5*(18)-(8)$) {}
        edge[dotted,out=180,in=-90] (18)
        edge[out=90,in=90] node[pos=.2] {$f$} (15);
\end{tikzpicture}
    \caption{Surgery in the proof of Lemma \ref{L:copycycle}. Edges
        labelled by $e$ and $f$ are cut and reconnected to turn three
        $\Phi$-cycles into four $\Phi$-cycles. The edges of $C$ are dashed,
        while the edges of $\mcC_1$ and $\mcC_2$ are dotted. Interiors of
        cycles are faces in $\mcP'$ and $\mcP''$.}
    \label{F:copycycle}
\end{figure}
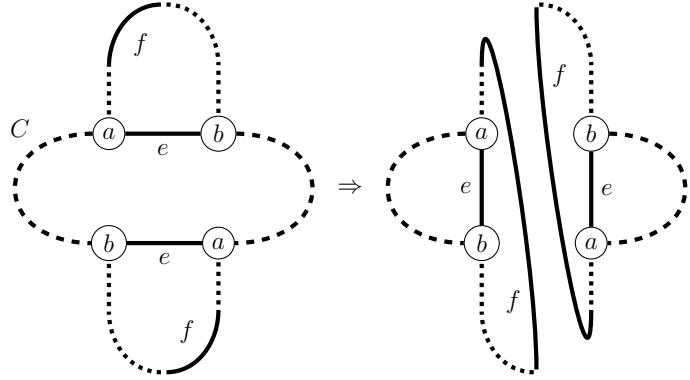
    Now suppose that $C$ is a $\mcC$-cycle in $\mcP'$ which is not a copy of
    $\mcC$, where every edge of $\mcC \in \Phi$ belongs to some other cycle of
    $\Phi$. By Definition \ref{D:constellation}, part (b), there is another
    cycle $\mcC' \in \Phi$ such that $E(\mcC) \cap E(\mcC') = \{e\}$, and
    $\mcC'$ contains an edge $e'$ which does not belong to any cycle in $\Phi
    \setminus \{\mcC'\}$. As above, there are distinct edges $e_1$ and $e_2$
    in $C$ such that $h(e_1) = h(e_2) = e$. By Lemma \ref{L:retractcycle}, each
    edge $e_i$ belongs to a unique $\mcC'$-cycle $C_i$, $i=1,2$. Since $C_1$
    and $C_2$ are copies of $\mcC'$, we must have $C_1 \neq C_2$. Let $f_i$ be
    the unique edge in $C_i$ with $h(f_i) = e'$. Thus we can construct a 
    new picture $\mcP''$ by cutting and reconnecting $e_i$ and $f_i$, $i=1,2$,
    as shown in Figure \ref{F:copycycle}.

    By Lemmas \ref{L:retractcycle} and \ref{L:struct1}, $C$ and $C_i$
    are the only $\Phi$-cycles containing $e_i$, while $C_i$ is the only
    $\Phi$-cycle containing $f_i$. Thus all $\Phi$-cycles in $\mcP''$, aside
    from those shown in Figure \ref{F:copycycle}, are $\Phi$-cycles in $\mcP'$,
    and hence will be facial covers. Because $C$ is a cover and $C_1$ and $C_2$
    are copies, the new cycles created in Figure \ref{F:copycycle} will be
    facial covers. We conclude again that $\mcP''$ will be an element of $\mfP$
    with $\Cycle(\mcP',\Phi) < \Cycle(\mcP'',\Phi)$, a contradiction. Thus
    every $\Phi$-cycle in $\mcP'$ is a facial copy.
\end{proof}

\subsection{Cycles covered by other cycles}
Finally, we prove two lemmas that will allow us to handle non-stellar cycles.
\begin{lemma}\label{L:covering1}
    Let $\Phi' \subseteq \Phi$ be a pair of $b$-constellations. Suppose that
    $\mcP$ is an $\mcH$-picture satisfying \ref{P1} with respect to $\Phi$,
    such that all $\Phi'$-cycles in $\mcP$ are facial covers. Let $\mcC_0$ be a
    connected closed subhypergraph of $\mcC \in \Phi \setminus
    \Phi'$,\footnote{In other words, $\mcC_0$ is either equal to $\mcC$, or a
    path in $\mcC$.} such that every edge of $\mcC_0$ is contained in an
    element of $\Phi'$. If $C$
    is a $\mcC$-cycle in $\mcP$, then:
    \begin{enumerate}[(a)]
        \item $C[\mcC_0]$ is a cover of $\mcC_0$, and
        \item if $C_0$ is a connected component of $C[\mcC_0]$, then
            all edges not contained in $C$ and incident with a vertex of $C_0$
            lie on the same side of $C$. 
    \end{enumerate}
\end{lemma}
\begin{proof}
    Let $e$ be an edge of $C$ with $h(e) \in E(\mcC_0)$. By hypothesis and Lemma
    \ref{L:retractcycle}, $e$ belongs to a $\mcC'$-cycle
    $C'$, where $\mcC' \in \Phi'$. Since $C'$ is a facial cover, $e$ has two
    distinct endpoints $a$ and $b$ with $h(a) \neq h(b)$. We conclude that
    $C[\mcC_0]$ is a cover of $\mcC_0$.
    
    Now let $x$ and $y$ be the edges of $\mcP$ incident to $a$ and $b$
    respectively, but not contained in $C$. Since $|E(\mcC) \cap E(\mcC')| \leq
    1$ and $\mcC$ is simple, we conclude that $x$ and $y$ belong to $C'$. Let $\mcD$ be a simple
    region bounded by $C$. By Proposition \ref{P:boundary}, since $C'$ is facial the
    edges $x$,$y$, and $e$ either belong to $\mcD$, or form an outer
    quadrilateral in $\germ(\mcP,\mcD)$. It follows that $x$ and $y$ lie on the
    same side of $C$. Since $C_0$ is connected, all edges not contained in $C$
    and incident to $C_0$ lie on the same side of $C$. 
\end{proof}

\begin{lemma}\label{L:covering2}
    Let $\Phi' \subseteq \Phi$ be a pair of $b$-constellations. Suppose that
    $\mcP$ is a $b$-minimal picture with no closed loops satisfying \ref{P1}
    and \ref{P2}, such that every $\Phi'$-cycle is a facial cover. If every
    edge of $\mcC \in \Phi$ is contained in some element of $\Phi'$, then every
    $\mcC$-cycle in $\mcP$ is a facial cover. 
\end{lemma}
\begin{proof}
    Let $C$ be a $\mcC$-cycle. Applying Lemma \ref{L:covering1} with $\mcC_0 =
    \mcC$, we get immediately that $C$ is a cover, and that the edges incident
    to $C$ all lie on the same side of $C$. 

    Suppose $b=0$, and let $\mcD$ be a simple region bounded by $C$. If all the
    edges incident to $C$ lie in $\mcD$, then $\germ(\mcP,\mcD)$ is closed, and
    $\ch(\germ(\mcP,\mcD)) \cdot b = 0$. Since $\mcP$ is $b$-minimal,
    $\germ(\mcP,\mcD)$ must be $b$-minimal, and this implies that
    $\germ(\mcP,\mcD)$ must have size zero, in contradiction of the fact that
    $\germ(\mcP,\mcD)$ contains $C$. Thus all the edges incident to $C$ lie
    outside the interior of $C$, and we conclude that $\res(\mcP,\mcD)$ is
    closed. But once again, this implies that $\res(\mcP,\mcD)$ must have size
    zero, and since $\mcP$ has no closed loops, this implies that
    $\res(\mcP,\mcD)$ is empty. We conclude that $C$ is facial.

    Now suppose that $\mcP$ is closed. Similarly to the proof of Lemma
    \ref{L:hypp2}, the fact that $\mcP$ is $b$-minimal with size greater than
    zero implies that $\ch(\mcP) \cdot b = 1$. Now $C$ bounds two simple
    regions $\mcD_1$ and $\mcD_2$ in the sphere, where we assume that all edges
    incident to $C$ are contained in $\mcD_1$, so $\res(\mcP,\mcD_2)$ is
    closed. If $\ch(\res(\mcP,\mcD_2)) \cdot b = 1$, then $\res(\mcP,\mcD_2)$
    is $b$-equivalent to $\mcP$, contradicting the $b$-minimality of $\mcP$.
    Thus $\ch(\res(\mcP,\mcD_2)) \cdot b = 0$, and as above,
    $\res(\mcP,\mcD_2)$ must be empty. We conclude again that $C$ is facial. 
\end{proof}

\subsection{Proof of Theorem \ref{T:constellation}}
    Given a $b$-constellation $\Phi$, let $\Phi_0$ be the set of $b$-stellar
    cycles in $\Phi$, and let $\Phi_1$ be the set of cycles $\mcC \in \Phi$
    such that every edge of $\mcC$ belongs to a cycle in $\Phi_0$. Then $\Phi_0
    \subseteq \Phi_1$, and by Lemma \ref{L:struct2}, $\Phi_0$ and $\Phi_1$
    are $b$-constellations. 

    Given a picture $\mcP_0$ satisfying \ref{P1} and \ref{P2}, Lemma
    \ref{L:stellarfacial} states that we can find a $b$-equivalent picture
    $\mcP_1$ which is $b$-minimal and has no closed loops, such that every
    $\Phi_0$-cycle in $\mcP_1$ is facial. By Proposition \ref{P:facialcover}, every
    $\Phi_0$-cycle in $\mcP_1$ is a facial cover. By Lemma \ref{L:covering2},
    every $\Phi_1$-cycle in $\mcP_1$ is also a facial cover. 

    This leaves the cycles in $\Phi \setminus \Phi_1$. By definition, any
    element of $\Phi \setminus \Phi_1$ is non-stellar. If a non-stellar
    cycle $\mcC \in \Phi$ shares an edge with another cycle $\mcC' \in \Phi$,
    then $\mcC'$ must be $b$-stellar. Hence $\Phi \setminus \Phi_1$ consists
    of the non-stellar cycles $\mcC \in \Phi$ which have an edge $e$ not
    contained in any element of $\Phi \setminus \{\mcC\}$. For the purpose of
    this proof, we say that such an edge $e \in \mcC$ is \emph{independent}. By part (a)
    of Definition \ref{D:constellation}, every element $\mcC \in \Phi \setminus
    \Phi_1$ has either one or two independent edges. In the latter case, the
    two edges will be incident with a common vertex of $\mcC$. 

    Suppose that $C$ is a $\mcC$-cycle in $\mcP_1$, where $\mcC \in \Phi
    \setminus \Phi'$. Let $\mcC_0$ be the path containing the
    non-independent edges of $\mcC$, regarded as a closed subhypergraph.  By
    Lemma \ref{L:covering1}, $C[\mcC_0]$ is a cover of $\mcC_0$, and since
    $\mcC_0$ is a path rather than a cycle, all connected components of
    $C[\mcC_0]$ are copies of $\mcC_0$. If $C[\mcC_0]$ is non-empty, we can
    write $C$ as a sequence $C_1 C_2 \cdots C_{2k-1} C_{2k}$ of paths $C_{i}$,
    where $k \geq 1$, the path $C_{2i-1}$ is a connected component of
    $C[\mcC_0]$ for all $i=1,\ldots,k$, and the edges of $C_{2i}$ are labelled
    by independent edges of $\mcC$ for all $i=1,\ldots,k$.

    \begin{figure}
        \begin{subfigure}{\textwidth}
            \centering
            \begin{tikzpicture}[auto,ultra thick,scale=.5,emptynode/.style={inner sep=0},every node/.style={scale=.8},vertex/.style={circle,draw,thin,inner sep=2.5}]

    \begin{scope} 
        \coordinate (0) at (0,0);
        \node[vertex] (1) at ($(0)+(-60:2)$) {$a_1$};
        \node[emptynode] (2) [right=of 1] {};
        \node[vertex] (3) [right=2cm of 2] {$a_2$};
        \coordinate (4) at ($(3)+(60:2)$) {};

        \draw (1) -- (2);
        \draw[dotted] (2) -- (3);
        \draw (1) to [bend left] node {$e$} (0);
        \draw (3) to [bend right] node [swap] {$e$} (4);
        \draw [decorate,decoration={brace,amplitude=5},thin] ($(1)+(90:.75)$) -- node [yshift=7] {$C_1$} ($(3)+(90:.75)$);

        \draw (1) -- +(-90:1.5);
        \draw (2) -- +(-90:1.5);
        \draw (3) -- +(-90:1.5);

        \draw[dotted] (0) to [out=80,in=180] ($.5*(0)+.5*(4)+(0,2.2)$) coordinate (a)
            to [out=0,in=100] (4);
        \node[above] at (a) {$C$};

        \foreach \x in {0,2,4}
            \draw [fill] (\x) circle [radius=.2];
    \end{scope}
    
    \node at ($(4)+(2,0)$) {$\surg$};
    
    \begin{scope}[shift={++($(4)+(4,0)$)}] 
        \coordinate (0) at (0,1.2);
        \node[vertex] (1) at ($(0)+(-60:3.5)$) {$a_1$};
        \node[emptynode] (2) [right=of 1] {};
        \node[vertex] (3) [right=2cm of 2] {$a_2$};
        \coordinate (4) at ($(3)+(60:3.5)$) {};

        \draw (1) -- (2);
        \draw[dotted] (2) -- (3);
        \draw (1) to [out=45,in=135] node [pos=.2] {$C^{(1)}$} node[swap] {$e$} (3);
        \draw (0) to node {$e$} (4);

        \draw (1) -- +(-90:1.5);
        \draw (2) -- +(-90:1.5);
        \draw (3) -- +(-90:1.5);

        \draw[dotted] (0) to [out=80,in=180] ($.5*(0)+.5*(4)+(0,2.2)$) coordinate (a)
            to [out=0,in=100] (4);
        \node[above] at (a) {$C^{(2)}$};

        \foreach \x in {0,2,4}
            \draw [fill] (\x) circle [radius=.2];
    \end{scope}
    
\end{tikzpicture}
            \caption{one independent edge.}
            \label{F:oneindependent}
        \end{subfigure}

        \begin{subfigure}{\textwidth}
            \centering
            \begin{tikzpicture}[auto,ultra thick,scale=.5,emptynode/.style={inner sep=0},every node/.style={scale=.8},vertex/.style={circle,draw,thin,inner sep=2.5}]

    \begin{scope} 
        \coordinate (0) at (0,0);
        \node[vertex] (1) at ($(0)+(-60:2)$) {$a_1$};
        \node[emptynode] (2) [right=of 1] {};
        \node[vertex] (3) [right=2cm of 2] {$a_2$};
        \coordinate (4) at ($(3)+(60:2)$) {};

        \draw (1) -- (2);
        \draw[dotted] (2) -- (3);
        \draw (1) to [bend left] node {$e_1$} (0);
        \draw (3) to [bend right] node [swap] {$e_2$} (4);
        \draw [decorate,decoration={brace,amplitude=5},thin] ($(1)+(90:.75)$) -- node [yshift=7] {$C_1$} ($(3)+(90:.75)$);

        \draw (1) -- +(-90:1.5);
        \draw (2) -- +(-90:1.5);
        \draw (3) -- +(-90:1.5);

        \draw[dotted] (0) to [out=80,in=180] ($.5*(0)+.5*(4)+(0,2.2)$) coordinate (a)
            to [out=0,in=100] (4);
        \node[above] at (a) {$C$};

        \foreach \x in {0,2,4}
            \draw [fill] (\x) circle [radius=.2];
    \end{scope}
    
    \node at ($(4)+(2,0)$) {$\surg$};
    
    \begin{scope}[shift={++($(4)+(4,0)$)}] 
        \coordinate (0) at (0,1.2);
        \node[vertex] (1) at ($(0)+(-60:4.3)$) {$a_1$};
        \node[emptynode] (2) [right=of 1] {};
        \node[vertex] (3) [right=2cm of 2] {$a_2$};
        \coordinate (4) at ($(3)+(60:4.3)$) {};
        \node[vertex] (5) at ($.5*(0)+.5*(4)$) {$c$};
        \node[vertex] (6) at ($(5)+(0,-2)$) {$c$};

        \draw (1) -- (2);
        \draw[dotted] (2) -- (3);
        \draw (0) to node {$e_1$} (5) to node {$e_2$} (4);
        \draw (1) to [out=45,in=190] node [pos=.4] {$C^{(1)}$} node[swap] {$e_1$} (6);
        \draw (6) to [out=-10,in=135] node[swap] {$e_2$} (3);

        \draw (5) to node {$g$} (6);

        \draw (1) -- +(-90:1.5);
        \draw (2) -- +(-90:1.5);
        \draw (3) -- +(-90:1.5);

        \draw[dotted] (0) to [out=80,in=180] ($.5*(0)+.5*(4)+(0,2.2)$) coordinate (a)
            to [out=0,in=100] (4);
        \node[above] at (a) {$C^{(2)}$};

        \foreach \x in {0,2,4}
            \draw [fill] (\x) circle [radius=.2];
    \end{scope}
    
\end{tikzpicture}
            \caption{two independent edges.}
            \label{F:twoindependent}
        \end{subfigure}
        \caption{Surgery to cut apart the cycle $C$ in the proof of Theorem \ref{T:constellation}.}
    \end{figure}
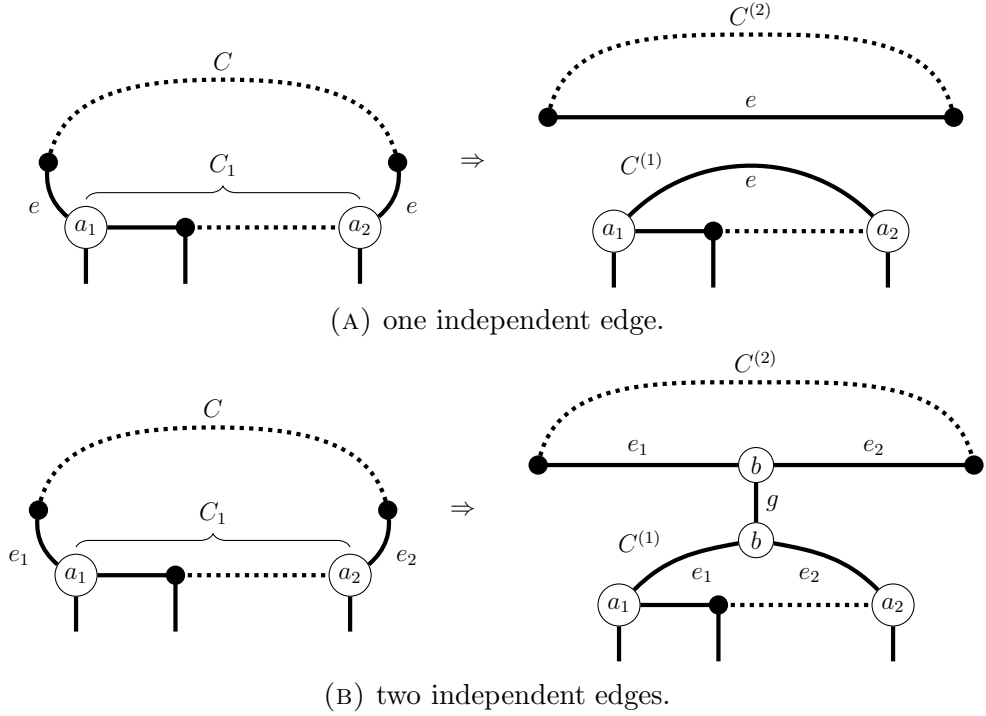
    Let $v_1$ (resp. $v_2$) be the first (resp. last) vertex of $C_1$, and let
    $f_1$ (resp. $f_2$) be the edge of $C_{2k}$ (resp. $C_2$) which is incident
    to $v_1$ (resp. $v_2$). The labels $a_i = h(v_i)$ are the endpoints of the
    path $\mcC_0$ in $\mcC$. If $\mcC$ has one independent edge $e$, then $e$
    will join $a_1$ and $a_2$, and $h(f_1) = h(f_2) = e$. By Lemma
    \ref{L:covering1}, every edge incident to $C_1$ lies on the same side of
    $C$, so we can cut and reconnect $f_1$ and $f_2$ as shown in Figure
    \ref{F:oneindependent}.
    If $\mcC$ has two independent edges, then we can do something similar.  In
    this case, there will be an independent edge $e_i = h(f_i)$ incident to $a_i$,
    $i=1,2$, and both edges will be incident with a third vertex $c$. Let $g$
    be the edge incident to $c$ not in $\mcC$. Since $\mcC$ is cubic, any cycle
    containing $g$ must also contain either $e_1$ or $e_2$. Since these edges
    are independent, $g$ is not contained in any element of $\Phi$.  We can
    then cut $f_1$ and $f_2$ and reconnect them by adding two vertices labelled
    by $c$, and a edge labelled by $g$, as shown in Figure \ref{F:twoindependent}.

    These surgeries are not as well behaved as those considered previously:
    When $\mcC$ has two independent edges, we end up increasing the size, so
    the result will no longer be $b$-minimal. When $\mcC$ has a single
    independent edge, it is possible that $f_1 = f_2$, in which case we create a
    closed loop. However, in both cases we split $C$ into two cycles $C^{(1)}$ and
    $C^{(2)}$, where $C^{(1)}$ is a facial copy of $\mcC$, and
    $C^{(2)}[\mcC_0]$ has fewer connected components than $C[\mcC_0]$.
    Furthermore, we do not change any other $\mcC$-cycle. And since we only
    change or add edges whose labels are not contained in any element of
    $\Phi \setminus \{\mcC\}$, we conclude that $C$ is the only $\Phi$-cycle
    changed by this surgery. As a result we may repeat this type of surgery
    to get a picture $\mcP_2$ (not necessarily
    $b$-minimal, and possibly containing closed loops) which is character-equivalent
    to $\mcP_1$, and in which every $\Phi$-cycle is either a facial cover, or
    labelled only by independent edges. 

    Let $\mfP$ be the collection of pictures which are character-equivalent
    to $\mcP_2$, and in which every $\Phi$-cycle is either a facial cover
    or labelled only by independent edges. Let $\mcP_3$ be an element of $\mfP$
    of minimum size, and let $\mcP_4$ be the picture $\mcP_3$ with all closed
    loops deleted. Clearly $\mcP_4$ is also an element of $\mfP$ of minimum size.
    Suppose $\mcP_4$ has a $\mcC$-cycle $C$ which is not a facial cover for
    some $\mcC \in \Phi$. By definition, $C$ is labelled by independent edges
    of $\mcC$. Since $\mcP_4$ has no closed loops and $\mcC$ is simple, $C$ has
    at least two edges. Consequently $\mcC$ must have two independent edges,
    say $e_1$ and $e_2$. As above, let $c$ be the vertex incident to both $e_1$
    and $e_2$, and let $g$ be the edge incident to $c$ and not in $\mcC$. Since
    every edge of $C$ is labelled by $e_1$ or $e_2$, every vertex of $C$ must
    be labelled by $c$. We can now argue similarly to Lemma \ref{L:nosillycycles}:
    $C$ must consist of a sequence of edges $f_1, \ldots, f_{2k}$, where $h(f_{2i-1})
    = e_1$ and $h(f_{2i}) = e_2$ for all $1 \leq i \leq k$. Let $v_1,\ldots,v_{2k}$
    be the vertices of $C$ in order, so $f_i$ has endpoints $v_i$ and
    $v_{i+1}$, where $v_{2k+1} := v_1$, and let $g_i$ be the edge of $\mcP_4$
    incident to $v_i$ with $h(g_i) = g$. Let $\mcP_5$ be the picture $\mcP_4$ with
    all the edges and vertices of $C$ deleted, and $g_{2i-1}$ and $g_{2i}$ joined
    into a single edge along the path taken by $f_{2i-1}$.  Since the edges
    $e_1$, $e_2$, and $g$ do not belong to any element of $\Phi \setminus
    \{\mcC\}$, this process does not create or change any other $\Phi$-cycle.
    Thus $\mcP_5 \in \mfP$, in contradiction of the minimality of $\mcP_4$. 
    We conclude that every $\Phi$-cycle in $\mcP_4$ is a facial cover, and the
    theorem follows from Lemma \ref{L:copycycle}.

\section{Proof of the embedding theorem}\label{S:proof}

In this section we finish the proof of Theorem \ref{T:invembedding} (and thus
complete the proof of Theorem \ref{T:embedding}). We continue
with the notation from Section \ref{S:wagonwheel}, so $\mcI := \Inv\langle
S:R\rangle$ is a presentation by involutions over $\Z_2$, the set of 
relations is written as $R = \{r_1,\ldots, r_m\}$ where $r_i = J^{p_i} s_{i1}
\cdots s_{in_i}$, $G$ is the group with presentation $\mcI$, and $\mcW :=
\mcW(\mcI)$ is the corresponding wagon wheel hypergraph. Although we do not yet
assume that $\mcI$ is collegial, for convenience we assume that the length
$n_i$ of the relation $r_i$ is at least $4$ for all $1 \leq i \leq m$ (by
Remark \ref{R:atleastthree}, this assumption holds if $\mcI$ is collegial). We
also use the same notation for the vertices and edges of $\mcW$, so
\begin{equation*}
    V = \bigcup_{i=1}^m V_i, \text{ where }
    V_i := \{(i,j,k) : j \in \Z_{n_i}, 1 \leq k \leq 3 \}, \text{ and }
\end{equation*}
\begin{equation*}
    E = S \sqcup \bigcup_{i=1}^m E_i, \text{ where } E_i := \{a_{ij},b_{ij},c_{ij},d_{ij} : j \in \Z_{n_i}\}.
\end{equation*}
In addition, we make the following definitions:
\begin{itemize}
    \item Let $\mcW_i$ be the closed subhypergraph of $\mcW$ containing
        vertices $V_i$ and edges $E_i$. (The open neighbourhood $\mcN(\mcW_i)$
        is shown in Figure \ref{F:wagonwheel}.) 
    \item Let $\mcA_i$, $1 \leq i \leq m$ be the cycle containing edges
        $a_{i1}$, $b_{i1}$, $a_{i2}$, $b_{i2},\ldots,a_{in_i}$, $b_{in_i}$. 
    \item Let $\mcB_i$, $1 \leq i \leq m$, be the cycle containing edges $d_{i1},\ldots,d_{in_i}$. 
    \item Let $\mcC_{ij}$, $1 \leq i \leq m$, $j \in \Z_{n_i}$, be the cycle
        containing edges $a_{ij}$, $b_{ij}$, $c_{ij}$, $d_{ij}$, and $c_{i,j-1}$. 
    \item Let $\Phi = \{\mcC_{ij} : 1 \leq i \leq m, j \in \Z_{n_i}\} \cup \{\mcB_i : 1 \leq i \leq m\}$.
\end{itemize}

Before we can prove Theorem \ref{T:invembedding}, we need some preliminary 
lemmas.
\begin{lemma}\label{L:morphism}
    If $b$ is an $\mcI$-labelling of $\mcW$, then there is a well-defined
    morphism $G \arr \Gamma(\mcW,b)$ over $\Z_2$ sending $s \mapsto x_s$ for all
    $s \in S$. 
\end{lemma}
\begin{proof}
    There is a well-defined morphism $\mcF(S) \times \Z_2 \arr \Gamma(\mcW,b)$
    over $\Z_2$ sending $s \mapsto x_s$.  As can be seen from Figure
    \ref{F:wagonwheel}, there is a $\mcN(\mcW_i)$-picture $\mcP$ with
    $\bd(\mcP) = s_{i1} \cdots s_{in_i}$ and $\ch(\mcP) \cdot b = \sum_{v \in
    V_i} b_v = p_i$.  By Proposition \ref{P:vankampen2}, the relation $r_i$ holds in
    $\Gamma(\mcW,b)$ for all $1 \leq i \leq n$. 
\end{proof}

\begin{lemma}\label{L:retract1}
    $\mcN(\mcB_i)$ is a retract of $\mcW$ for all $1 \leq i \leq n$.
\end{lemma}
\begin{figure}
    \begin{tikzpicture}[auto,ultra thick,scale=.4,emptynode/.style={inner sep=0},every node/.style={scale=.8}]

    \begin{scope}
        \foreach \x in {1,...,4}
            \node[emptynode] (\x) at (90*\x-45:1.8) {};
        \foreach \x in {5,...,12}
            \node[emptynode] (\x) at (45*\x-225:4) {};
        \draw circle [radius=1.8];
        \draw circle [radius=4];
        \node[below left] at (1) {$\mcB_i$};
        \draw (1) -- (6);
        \draw (2) -- (8);
        \draw (3) -- (10);
        \draw (4) -- (12);

        \foreach \x in {5,7,9,11} {
            \draw[thin] (\x.45*\x-215) to ++(45*\x-195:1.8) coordinate (e1) {};
            \draw[thin] (\x.45*\x-235) to ++(45*\x-255:1.8) coordinate (e2) {};
            \path[fill,pattern=north west lines] (\x.45*\x-215) to (e1) to (e2)
                to (\x.45*\x-235) to [bend right] (\x.45*\x-215);
        }

        \foreach \x in {1,...,12}
            \draw[fill] (\x) circle [radius=.2];
    \end{scope}

    \coordinate (r1) at ($(5)+(0:1.8)$);
    \node at ($(r1)+(1.75,0)$) {$\surg$};

    \begin{scope}[shift={++($2*(r1)+(1.7,0)$)}]
        \foreach \x in {1,...,4}
            \node[emptynode] (\x) at (90*\x-45:1.8) {};
        \foreach \x in {5,...,8}
            \node[emptynode] (\x) at (90*\x-45:4) {};
        \draw circle [radius=1.8];
        \draw circle [radius=4];
        \node[below left] at (1) {$\mcB_i$};
        \draw[blue] (1) -- (5);
        \draw[blue] (2) -- (6);
        \draw[blue] (3) -- (7);
        \draw[blue] (4) -- (8);

        \foreach \x in {1,...,8}
            \draw[fill] (\x) circle [radius=.2];
    \end{scope}

    \coordinate (r2) at ($2*(r1)+(5.7,0)$);
    \node at ($(r2)+(1.75,0)$) {$\surg$};

    \begin{scope}[shift={++($(r2)+(7.5,0)$)}]
        \path (3,-1) arc(0:45:3 and 2) coordinate (1) 
                arc (45:135:3 and 2) coordinate (2)
                arc (135:225:3 and 2) coordinate (3)
                arc (225:315:3 and 2) node[pos=.5,swap] {$\mcB_i$} coordinate (4)
                arc (315:360:3 and 2);
        \path (3,1) arc(0:45:3 and 2) coordinate (5) 
                arc (45:135:3 and 2) coordinate (6)
                arc (135:225:3 and 2) coordinate (7)
                arc (225:315:3 and 2) coordinate (8)
                arc (315:360:3 and 2);
        
        \draw [help lines,<-] ($(1)+(0,.4)$) -- ($(5)+(0,-.4)$);
        \draw [help lines,<-] ($(2)+(0,.4)$) -- ($(6)+(0,-.4)$);
        \draw [help lines,<-] ($(3)+(0,.4)$) -- ($(7)+(0,-.4)$);
        \draw [help lines,<-] ($(4)+(0,.4)$) -- ($(8)+(0,-.4)$);
        
        \draw[blue] (1) to [out=35,in=-90] ($.5*(1)+.5*(5)+(1.5,0.5)$) to [out=100,in=20] (5);
        \draw[blue] (2) to [out=145,in=-90] ($.5*(2)+.5*(6)+(-1.5,0.5)$) to [out=80,in=160] (6);

        \draw (0,-1) ellipse (3 and 2);
        \draw (0,1)  ellipse (3 and 2);

        \draw[blue] (3) to [out=-125,in=-90] ($.5*(3)+.5*(7)+(-1.5,-0.5)$) to [out=80,in=-160] (7);
        \draw[blue] (4) to [out=-35,in=-90] ($.5*(4)+.5*(8)+(1.5,-0.5)$) to [out=100,in=-20] (8);

        \foreach \x in {1,...,8}
            \draw[fill] (\x) circle [radius=.2];

    \end{scope}
\end{tikzpicture}
    \caption{To retract $\mcN(\mcW_i)$ onto $\mcN(\mcB_i)$, we remove
            intermediary vertices to get a simplified wagon wheel shape,
            and then fold this wagon wheel onto the central cycle. In
            the example shown above, $n_i=4$.}
    \label{F:retractcenter}
\end{figure}
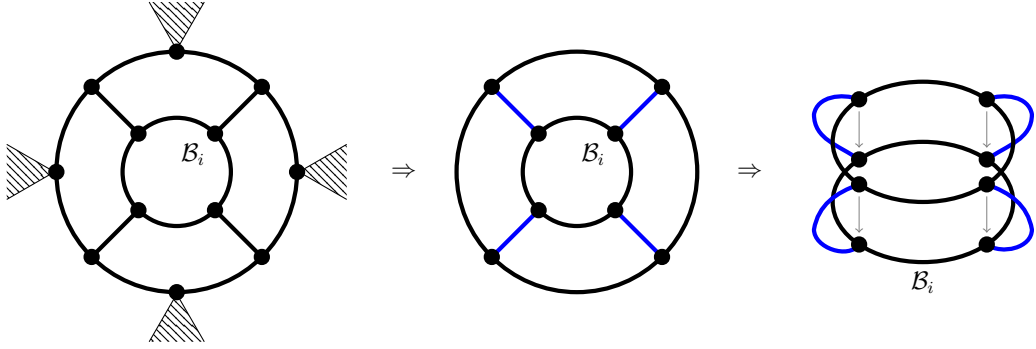
\begin{proof}
    Define $r : \mcW \arr \mcN(\mcB_i)$ by 
    \begin{align*}
        r((i',j,k)) & = \begin{cases} \vareps & i' \neq i \\
                                    \vareps & i=i' \text{ and } k=1 \\
                                    (i,j,3) & i=i' \text{ and } k=2,3
                      \end{cases}, \text{ and} \\
        r(e) & =  \begin{cases} \vareps & e \in S \\
                              \vareps & e \in E_{i'} \text{ with } i \neq i' \\
                              c_{ij} & e = c_{ij} \\
                              d_{ij} & e \in \{a_{ij},b_{ij},d_{ij}\}
                \end{cases}.
    \end{align*}
    It is clear that $r$ is the identity on $\mcN(\mcB_i)$. The only 
    vertices mapped to $\vareps$ which are incident with edges in
    $E(\mcW) \setminus r^{-1}(\vareps)$ are the vertices $(i,j,1)$
    for $j \in \Z_{n_i}$. For these vertices, the incident edges
    $a_{ij}$ and $b_{ij}$ are identified as required by condition (2) of
    Definition \ref{D:generalizedmorphism}. 

    Since $r$ is the identity on $\mcN(\mcB_i)$, condition (1) of Definition
    \ref{D:generalizedmorphism} holds for the vertices $(i,j,3)$, $1 \leq j
    \leq n_i$. The vertices $(i,j,2)$ are incident with three edges of
    $E(\mcW) \setminus r^{-1}(\vareps)$, namely $a_{i,j+1}$, $c_{ij}$, and
    $b_{i,j}$, and these edges are mapped to the three edges $d_{i,j+1}$,
    $c_{ij}$, and $d_{i,j}$ incident to $(i,j,3)$. We conclude that
    condition (1) of Definition \ref{D:generalizedmorphism} also holds for
    the vertices $(i,j,2)$, and hence $r$ is a generalized morphism.

    The map $r$ can be visualized as deleting everything outside of $\mcW_i$ to
    get a simplified wagon wheel shape, and then folding this wagon wheel onto
    $\mcN(\mcB_i)$. This is depicted in Figure \ref{F:retractcenter}.
\end{proof}

\begin{lemma}\label{L:retract2}
    Let $s = s_{i{j}}$ for some $1 \leq i \leq m$ and $j \in \Z_{n_i}$.  If
    $R$ is cyclically reduced, and $\mult(s; r_{i'})$ is even for all $1 \leq
    i' \leq m$, then $\mcN(\mcC_{i{j}})$ is a retract of $\mcW$.
\end{lemma}
\begin{proof}
    We start by showing that $\mcN(\mcC_{i{j}})$ is a retract of $\mcN(\mcW_i)$.
    Since $\mcW_i$ depends only on the cyclic order of $s_{i1} \cdots
    s_{in_i}$, we can assume without loss of generality that $j=1$.
    Suppose that $\mult(s;r_{i}) = 2k$, and let $1=j_1 < j_2 < \cdots < j_{2k}
    \leq n_i$ be a list of the indices $1 \leq l \leq n_i$ such that
    $s_{il}=s$. Since $R$ is cyclically reduced, $j_{i+1} > j_i+1$ for all
    $i=1,\ldots,2k-1$, and $j_{2k} < n_i$. For convenience, let 
    \begin{align*}
        \mcJ_r = \{1,j_2,\ldots,j_{2k}\}, & \quad \mcJ_l = \{0,j_2-1,\ldots,j_{2k}-1\},\\
        \mcJ_{r}^{odd} = \{1,j_3,j_5\ldots,j_{2k-1}\}, & \quad
        \mcJ_{l}^{odd} = \{0,j_3-1,j_5-1\ldots,j_{2k-1}-1\}, \\
        \mcJ_{r}^{even} = \{j_2,j_4,\ldots,j_{2k}\}, & \quad\text{ and } 
            \mcJ_{l}^{even} = \{j_2-1,j_4-1,\ldots,j_{2k}-1\}. \\
    \end{align*}
    These sets represent the indices of vertices and edges on the right and
    left of the cycles $\mcC_{ij_{p}}$. To talk about edges which do not belong
    to these cycles, we also define
    \begin{align*}
        \overline{\mcJ}_r & = \{j_1+1,j_1+2,\ldots,j_2-1,j_3+1,\ldots,j_4-1,\ldots,j_{2k-1}+1,\ldots,j_{2k}-1\} \text{ and } \\
        \overline{\mcJ}_l & = \{j_2+1,j_2+2,\ldots,j_3-1,j_4+1,\ldots,j_5-1,\ldots,j_{2k}+1,\ldots,n_i\}.
    \end{align*}
    Define $q_i : \mcN(\mcW_i) \arr \mcN(\mcC_{i1})$ by 
    \begin{align*}
        q_i((i,j,k)) & = \begin{cases} \vareps & k=1 \text{ and } j \not\in \mcJ_r \\
                                     (i,1,1) & k=1 \text{ and } j \in \mcJ_r \\
                                     \vareps & k=2,3 \text{ and } j \not\in \mcJ_r \cup \mcJ_l \\
                                     (i,0,k) & k=2,3 \text{ and } j \in \mcJ_{l}^{odd} \cup \mcJ_{r}^{even} \\
                                     (i,1,k) & k=2,3 \text{ and } j \in \mcJ_{r}^{odd} \cup \mcJ_{l}^{even}.
                        \end{cases},
    \end{align*}
    \begin{align*}
        q_i(a_{ij}) & = \begin{cases} a_{i1} & j \in \mcJ_r^{odd} \\
                                    b_{i1} & j \in \mcJ_r^{even} \\
                                    b_{i0} & j \in \overline{\mcJ}_l \\
                                    a_{i2} & j \in \overline{\mcJ}_r \\
                      \end{cases},
        & q_i(b_{ij}) & = \begin{cases}
                        a_{i1} & j \in \mcJ_{r}^{even} \\
                        b_{i1} & j \in \mcJ_{r}^{odd} \\
                        b_{i0} & j \in \overline{\mcJ}_l \\
                        a_{i2} & j \in \overline{\mcJ}_r \\
                     \end{cases}, \\
        q_i(c_{ij}) & = \begin{cases}
                        c_{i0} & j \in \mcJ_l^{odd} \cup \mcJ_r^{even} \\
                        c_{i1} & j \in \mcJ_r^{odd} \cup \mcJ_l^{even} \\
                        \vareps & j \not\in \mcJ_r \cup \mcJ_l \\
                      \end{cases}, 
        & q_i(d_{ij}) & = \begin{cases}
                        d_{i1} & j \in \mcJ_r \\
                        d_{i0} & j \in \overline{\mcJ}_l \\
                        d_{i2} & j \in \overline{\mcJ}_r \\
                     \end{cases}, \\
        q_i(s') & = \begin{cases} s & s' = s \\
                           \vareps & s' \in S \setminus \{s\} 
                  \end{cases}. & \\ 
    \end{align*}
    The map $q_i$ can be visualized as deleting vertices and edges $c_{ij}$ (the 
    spokes of the wagon wheel) not in the cycles $\mcC_{ij_p}$, and then
    folding up the cycles $\mcC_{ij_p}$ onto $\mcC_{i1}$, alternating the
    directions of the folds after each $\mcC_{ij_p}$ like a napkin. The smallest
    example, when $k=1$, is depicted in Figure \ref{F:retractside}. As in the
    proof of Lemma \ref{L:retract1}, it follows that $q_i$ is a retract. 
\begin{figure}
    \begin{tikzpicture}[auto,ultra thick,scale=.4,emptynode/.style={inner sep=0},every node/.style={scale=.8}]

    \begin{scope}
        \foreach \x in {1,...,4}
            \node[emptynode] (\x) at (90*\x-45:1.8) {};
        \foreach \x in {5,...,12}
            \node[emptynode] (\x) at (45*\x-225:4) {};
        \draw circle [radius=1.8];
        \draw circle [radius=4];
        \draw (1) -- (6);
        \draw (2) -- (8);
        \draw (3) -- (10);
        \draw (4) -- (12);

        \foreach \x in {5,7,9,11} {
            \draw[thin] (\x.45*\x-215) to ++(45*\x-195:1.8) coordinate (e1) {};
            \draw[thin] (\x.45*\x-235) to ++(45*\x-255:1.8) coordinate (e2) {};
            \path[fill,pattern=north west lines] (\x.45*\x-215) to (e1) to (e2)
                to (\x.45*\x-235) to [bend right] (\x.45*\x-215);
        }
        \node[above,fill=white,yshift=12,inner sep=2] at (7) {$s$};
        \node[below,fill=white,yshift=-12,inner sep=2] at (11) {$s$};

        \foreach \x in {1,...,12}
            \draw[fill] (\x) circle [radius=.2];
        
        \path (6) arc(45:90:4) node [pos=.5,swap] {$C_{i1}$};
    \end{scope}

    \coordinate (r1) at ($(5)+(0:1.8)$);
    \node at ($(r1)+(1.75,0)$) {$\surg$};

    \begin{scope}[shift={++($2*(r1)+(1.7,0)$)}]
        \foreach \x in {1,...,4}
            \node[emptynode] (\x) at (90*\x-45:1.8) {};
        \foreach \x in {5,...,12}
            \node[emptynode] (\x) at (45*\x-225:4) {};
        \draw (1) arc (45:135:1.8);
        \draw (3) arc (225:315:1.8);
        \draw (6) arc (45:135:4);
        \draw (10) arc (225:315:4);
        \draw[blue] (2) arc (135:225:1.8);
        \draw[blue] (4) arc (315:405:1.8);
        \draw[blue] (8) arc (135:225:4);
        \draw[blue] (12) arc (315:405:4);
        \draw (1) -- (6);
        \draw (2) -- (8);
        \draw (3) -- (10);
        \draw (4) -- (12);

        \foreach \x in {7,11} {
            \draw[thin] (\x.45*\x-215) to ++(45*\x-195:1.8) coordinate (e1) {};
            \draw[thin] (\x.45*\x-235) to ++(45*\x-255:1.8) coordinate (e2) {};
            \path[fill,pattern=north west lines] (\x.45*\x-215) to (e1) to (e2)
                to (\x.45*\x-235) to [bend right] (\x.45*\x-215);
        }
        \node[above,fill=white,yshift=12,inner sep=2] at (7) {$s$};
        \node[below,fill=white,yshift=-12,inner sep=2] at (11) {$s$};

        \foreach \x in {1,2,3,4,6,7,8,10,11,12}
            \draw[fill] (\x) circle [radius=.2];
        \path (6) arc(45:90:4) node [pos=.5,swap] {$C_{i1}$};
    \end{scope}

    \coordinate (r2) at ($2*(r1)+(5.7,0)$);
    \node at ($(r2)+(1.75,0)$) {$\surg$};

    \begin{scope}[shift={++($(r2)+(7.5,-1)$)}]
        \foreach \x in {1,2}
            \node[emptynode] (\x) at (90*\x-45:1.8) {};
        \foreach \x in {5,...,8}
            \node[emptynode] (\x) at (45*\x-225:4) {};
        \draw (1) -- (6);
        \draw (2) -- (8);
        \draw (1) arc (45:135:1.8);
        \draw (6) arc (45:135:4);

        \begin{scope}[rotate around x=-97]
            \foreach \x in {3,4}
                \node[emptynode] (\x) at (90*\x-45:1.8) {};
            \foreach \x in {9,...,12}
                \node[emptynode] (\x) at (45*\x-225:4) {};
            \draw (3) -- (10);
            \draw (4) -- (12);
            \draw (4) arc(-45:-135:1.8);
            \draw (12) arc(-45:-135:4);
            \foreach \x in {11} {
                \draw[thin] (\x.45*\x-215) to ++(45*\x-195:1.8) coordinate (e1) {};
                \draw[thin] (\x.45*\x-235) to ++(45*\x-255:1.8) coordinate (e2) {};
                \path[fill,pattern=north west lines] (\x.45*\x-215) to (e1) to (e2)
                    to (\x.45*\x-235) to [bend right] (\x.45*\x-215);
            }
        \end{scope}

        \foreach \x in {7} {
            \draw[thin] (\x.45*\x-215) to ++(45*\x-195:1.8) coordinate (e1) {};
            \draw[thin] (\x.45*\x-235) to ++(45*\x-255:1.8) coordinate (e2) {};
            \path[fill,pattern=north west lines] (\x.45*\x-215) to (e1) to (e2)
                to (\x.45*\x-235) to [bend right] (\x.45*\x-215);
        }
        \node[above,fill=white,yshift=12,inner sep=2] at (7) {$s$};

        \draw[blue] (1) -- (4);
        \draw[blue] (2) -- (3);
        \draw[blue] (12) -- (6);
        \draw[blue] (10) -- (8);
        \draw[help lines, ->] ($(12)+(.65,.2)$) to [bend right] ($(6)+(.39,-.4)$);

        \foreach \x in {1,2,3,4,6,7,8,10,11,12}
            \draw[fill] (\x) circle [radius=.2];
        \path (6) arc(45:90:4) node [pos=.5,swap] {$C_{i1}$};
    \end{scope}
\end{tikzpicture}
    \caption{To retract $\mcN(\mcW_i)$ onto $\mcN(\mcC_{i1})$, we remove
            intermediary vertices and edges, and then fold up the remaining
            cycles like a napkin. In the example above, $n_i=4$ and $\mult(s;r_i)=2$.
            In general, if $\mult(s;r_i)=2k$ then we make $k$ folds.}
    \label{F:retractside}
\end{figure}
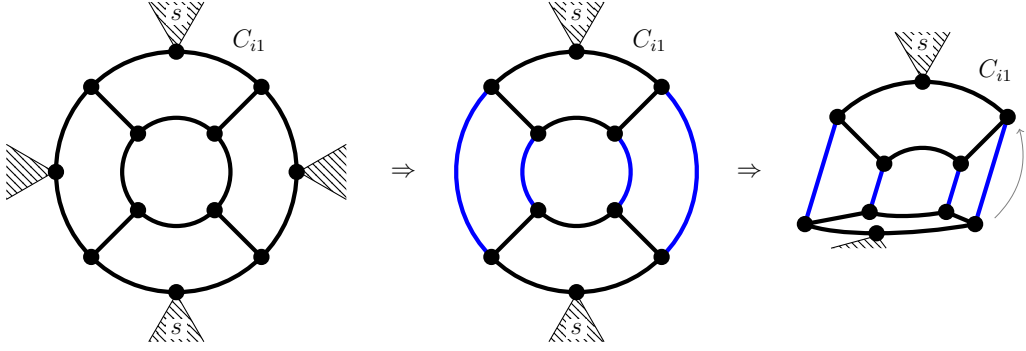

    Now we look at $\mcW_{i'}$ for $i' \neq i$. If $\mult(s;r_{i'}) = 0$, then
    define 
    \begin{equation*}
        q_{i'} : \mcN(\mcW_{i'}) \arr \mcN(\mcC_{ij}) : x \mapsto \vareps.
    \end{equation*}
    If $\mult(s;r_{i'}) > 0$, then find $j'$ such that $s_{i'j'} = s$, and let
    $f : \mcN(\mcW_{i'}) \arr \mcN(\mcC_{i'j'})$ be the retract defined above onto
    $\mcN(\mcC_{i'j'})$. Now $\mcN(\mcC_{i'j'})$ and $\mcN(\mcC_{ij})$ are both suns,
    so there is an isomorphism $g : \mcN(\mcC_{i'j'}) \arr \mcN(\mcC_{ij})$ with
    $g(s) = s$, and we let $q_{i'} = g \circ f$.

    The morphisms $q_{i'}$, $1 \leq i' \leq m$ all send $s' \in S$ to either
    $\vareps$ if $s' \neq s$, or to $s$ if $s' = s$. If $i' \neq i''$, then
    the open subhypergraphs $\mcN(\mcW_{i'})$ and $\mcN(\mcW_{i''})$ have
    no vertices in common. All common edges of $\mcN(\mcW_{i'})$ and
    $\mcN(\mcW_{i''})$ belong to $S$, so $q_{i'}$ and $q_{i''}$ agree on the
    intersection. Every vertex of $\mcW$ belongs to some $\mcW_i$. There may be
    elements $s'$ of $S$ which do not appear in any relation $r_{i'}$, and
    hence do not belong to any $\mcN(\mcW_i)$; for these edges, we can add
    additional morphisms which send $s' \mapsto \vareps$. By Proposition
    \ref{P:gluing}, there is a morphism $q : \mcW \arr \mcN(\mcC_{ij})$ which
    agrees with $q_{i'}$ on $\mcN(\mcW_{i'})$, and in particular is a retract.
\end{proof}

\begin{lemma}\label{L:constellation}
    Suppose $\mcI$ is collegial. Then there is an $\mcI$-labelling $b$ of
    $\mcW$ such that $\Phi$ is a $b$-constellation. 
\end{lemma}
\begin{proof}
    Let $b$ be any $\mcI$-labelling such that
    \begin{itemize}
        \item $|b^{-1}(1) \cap V(\mcW_i)| \leq 1$ for all $1 \leq i \leq m$, 
        \item $b((i,j,2)) = b((i,j,3)) = 0$ for all $1 \leq i \leq m$ and $j \in \Z_{n_i}$, and
        \item if $b((i,j,1)) = 1$, then either $\mult(s_{ij},r_{i'})$ is odd for some
            $1 \leq i' \leq m$, or $\mult(s_{ij'},r_{i'})$ is even for all $j' \in \Z_{n_i}$
            and $1 \leq i' \leq m$.
    \end{itemize}
    We will show that $\Phi$ is a $b$-constellation. First, we observe that
    $\mcB_i$ is $b$-stellar for all $1 \leq i \leq n$. Indeed, $\mcN(\mcB_i)$
    is a sun, and by Lemma \ref{L:retract1}, $\mcN(\mcB_i)$ is a retract
    of $\mcW$. Finally, $b|_{\mcB_i} = 0$, so $\mcB_i$ is $b$-stellar. 
    Similarly, Lemma \ref{L:retract2} implies that $\mcC_{ij}$ will be
    $b$-stellar for $1 \leq i \leq m$ and $j \in \Z_{n_i}$ as long as
    $b((i,j,1)) = 0$ and $\mult(s_{ij}; r_{i'})$ is even for all $1 \leq i'
    \leq m$. 

    Suppose $\mcC_{ij}$ is not $b$-stellar. If $\mult(s_{ij};r_{i'})$ is
    odd for some $1 \leq i' \leq m$, then by Definition \ref{D:collegial},
    $\mult(s_{i,j+1};r_{i''})$ and $\mult(s_{i,j-1};r_{i''})$ are even for all
    $1 \leq i'' \leq m$. By the definition of $b$, we must have $b((i,j+1,1)) =
    b((i,j-1,1)) = 0$. Thus $\mcC_{i,j+1}$ and $\mcC_{i,j-1}$ are $b$-stellar,
    so $\mcC_{ij}$ (which consists of edges $a_{ij}, b_{ij}, c_{ij}, d_{ij},
    c_{i,j-1}$) shares edges $c_{ij}$, $d_{ij}$, and $c_{i,j-1}$ with
    $b$-stellar cycles. If $\mult(s_{ij};r_{i'})$ is even for all $1 \leq i'
    \leq m$, then we must have $b((i,j,1)) = 1$. By the definition of $b$,
    this means that $\mult(s_{ij'},r_{i'})$ is even for all $1\leq i' \leq m$
    and $j' \in \Z_{n_i}$,
    and $b((i,j+1,1)) = b((i,j-1,1)) = 0$. It follows that $\mcC_{i,j+1}$ and
    $\mcC_{i,j-1}$ are $b$-stellar, and once again $\mcC_{ij}$ will share
    edges $c_{ij}$, $d_{ij}$, and $c_{i,j-1}$ with $b$-stellar cycles. Thus
    $\Phi$ satisfies condition (a) of Definition \ref{D:constellation}.

    The cycle $\mcC_{ij}$ is the only cycle in $\Phi$ containing edges $a_{ij}$
    and $b_{ij}$. Every edge $d_{ij}$ of $\mcB_{i}$ is also contained in
    $\mcC_{ij}$, so $\Phi$ satisfies condition (b) of Definition \ref{D:constellation}.

    Finally, it is easy to see that $|E(\mcC_{ij}) \cap E(\mcD_{i'})| \leq 1$
    for all $1 \leq i,i' \leq m$ and $1 \leq i' \leq m$, and (since $n_i \geq 4$
    for all $1 \leq i \leq m$) that $|E(\mcC_{ij}) \cap E(\mcC_{i'j'})| \leq 1$ for all
    distinct $(i,j)$ and $(i',j')$. We showed above that if $\mcC_{ij}$ is
    not $b$-stellar, then $\mcC_{i,j+1}$ and $\mcC_{i,j-1}$ are $b$-stellar, so
    there is no pair of non-$b$-stellar cycles in $\Phi$ with a common edge. 
    Thus $\Phi$ satisfies condition (c) of Definition \ref{D:constellation}.
\end{proof}

\begin{lemma}\label{L:region}
    Let $b$ be an $\mcI$-labelling, and let $\mcP$ be a $\mcW$-picture in which
    all $\Phi$-cycles are facial copies, and such that all edges in $\bd(\mcP)$
    belong to $S$.  Then there is a $G$-picture $\mcP'$ with $\bd(\mcP') =
    \bd(\mcP)$ and $\sign(\mcP') = \ch(\mcP) \cdot b$.
\end{lemma}
\begin{proof}
    Let $\mcW \setminus S$ denote the closed subhypergraph containing all
    vertices of $\mcW$ and all edges except those in $S$. Equivalently, 
    $\mcW \setminus S$ is the subhypergraph with connected components $\mcW_i$,
    $1 \leq i \leq m$. Suppose that $\mcP_0$ is a connected component of
    $\mcP[\mcW \setminus S]$. Since $\bd(\mcP)$ does not contain any edges
    from outside $S$, $\mcP_0$ must be closed. Since $\mcP_0$ is connected, there must be some $1 \leq
    i \leq m$ such that $\mcP_0$ is a $\mcW_i$-picture. By Lemma \ref{L:retractcycle},
    $\mcP_0$ contains a $\Phi$-cycle. By hypothesis, every
    $\Phi$-cycle in $\mcP_0$ is a facial copy, so every $\mcC_{ij}$-cycle
    will contain an edge labelled by $d_{ij}$. Thus $\mcP_0$ contains an
    edge labelled by $d_{ij}$ for some $j \in \Z_{n_i}$, and by Lemma
    \ref{L:retractcycle} again, $\mcP_0$ contains a $\mcB_i$-cycle $B$. Since
    $\mcB_i \in \Phi$, $B$ is also a facial copy, and hence consists of edges
    $\hat{d}_j$, $j \in \Z_{n_i}$, such that $h(\hat{d}_j) = d_{ij}$. 
    Let $\hat{c}_j$, $j \in \Z_{n_i}$, be the third edge incident to the common
    endpoint of $\hat{d}_j$ and $\hat{d}_{j+1}$. Every edge $\hat{d}_{j}$ is
    contained in a unique $\mcC_{ij}$-cycle $C_j$, which will consist of edges
    $\hat{c}_j$, $\hat{d}_j$, $\hat{c}_{j-1}$, and two additional edges
    $\hat{a}_j$ and $\hat{b}_j$ with $h(\hat{a}_j) = a_{ij}$ and $h(\hat{b}_j)
    = b_{ij}$. Let 
    \begin{equation*}
        \mcP_0' = B \cup \bigcup_{j\in \Z_{n_i}} C_j.
    \end{equation*}
    The common endpoint of $\hat{a}_j$ and $\hat{b}_j$ has degree two
    in $\mcP[\mcW \setminus S]$, while the endpoints of $\hat{c}_j$ have
    degree three and are incident with $\hat{c}_j$, $\hat{d}_j$, $\hat{d}_{j+1}$ and
    $\hat{c}_j$, $\hat{a}_{j+1}$, $\hat{b}_j$ respectively. Hence $\mcP_0'$ contains
    every edge of $\mcP[\mcW \setminus S]$ incident to a vertex of $\mcP_0'$.
    Thus $\mcP_0'$ is a maximal connected subgraph of $\mcP[\mcW \setminus S]$,
    so $\mcP_0' = \mcP_0$. In particular, we conclude that $|h^{-1}(v) \cap
    V(\mcP_0)| = 1$ for all $v \in V(\mcW_i)$. Since $b$ is an
    $\mcI$-labelling, it follows that $\ch(\mcP_0) \cdot b = a_i$. 

    Now let $\mcD(B)$ be the closure of the face bounded by $B$, let
    $\mcD(C_{j})$ be the closure of the face bounded by $C_{j}$, $j \in
    \Z_{n_i}$, and let
    \begin{equation*}
        \mcD(\mcP_0) := \mcD(B) \cup \bigcup_{j \in \Z_{n_i}} \mcD(C_{j}).
    \end{equation*}
    Clearly $\mcP_0$ is contained in $\mcD(\mcP_0)$, and conversely every edge
    or vertex of $\mcP$ in $\mcD(\mcP_0)$ belongs to $\mcP_0$. Since $\hat{a}_j$
    and $\hat{b}_j$ belong to the boundary of $\mcD(C_{j})$ and are not contained
    in $B$ or any $C_{j'}$, $j' \neq j$, we conclude that $\hat{a}_j$ and
    $\hat{b}_j$ belong to the boundary of $\mcD(\mcP_0)$. Every other edge of $C_j$
    belongs either to $B$, to $C_{j+1}$, or to $C_{j-1}$, so the boundary 
    of $\mcD(\mcP_0)$ does not contain any other edges of $C_j$. Similarly, the
    boundary does not contain any edges of $B$, or any of the vertices
    $(i,j,3)$, $j \in \Z_{n_i}$. We conclude that $\mcD(\mcP_0)$ is bounded by
    the $\mcA_i$-cycle $A = \hat{a}_1\hat{b}_1 \cdots \hat{a}_{n_i}
    \hat{b}_{n_i}$. In particular, $\mcD(\mcP_0)$ is a simple region, and
    $\bd(\germ(\mcP, \mcD(\mcP_0))) = s_{i1} \cdots s_{in_i}$ or $s_{in_i}
    \cdots s_{i1}$ depending on the orientation of $A$. 

    Since $\mcD(\mcP_0)$ contains only edges and vertices of $\mcP_0$, if $\mcP_1$
    is another connected component of $\mcP[\mcW\setminus S]$, then $\mcD(\mcP_0)$
    and $\mcD(\mcP_1)$ are completely disjoint. Thus, as in the proof of Corollary
    \ref{C:suncycle}, we can collapse each region $\mcD(\mcP_0)$ to a single vertex
    labelled by $r_i = J^{a_i} s_{i1} \cdots s_{in_i}$ to form a $G$-picture
    $\mcP'$. Since the edges of $S$ will be unchanged, $\bd(\mcP) =
    \bd(\mcP')$.  Since every vertex of $\mcP$ belongs to a unique connected
    component of $\mcP[\mcW\setminus S]$, we conclude that $\sign(\mcP') =
    \ch(\mcP) \cdot b$. 
\end{proof}

\begin{proof}[Proof of Theorem \ref{T:invembedding}]
    Let $G$ be the group with presentation $\Inv\langle S:R \rangle$, and suppose that $\mcI
    = \Inv\langle S:R\rangle$ is collegial. Let $\mcI^+ = \Inv\langle S:R^+\rangle$, and let $G^+$ be the
    even quotient of $G$, as in Definition \ref{D:evenodd}. By Lemma
    \ref{L:constellation}, we can choose an $\mcI$-labelling $b$ of $\mcW$ such
    that $\Phi$ is a $b$-constellation. Every $b$-stellar cycle is also $0$-stellar,
    so $\Phi$ is also a $0$-constellation by Definition \ref{D:constellation}.
    By Lemma \ref{L:morphism}, there are morphisms
    $\phi : G \arr \Gamma(\mcW,b)$ and $\phi^+ : G^+ \arr \Gamma(\mcW,0)$, both
    sending $s \mapsto x_s$ (note that $0$ is an $\mcI^+$-labelling). 

    To show that $\phi$ is injective, we start with $\phi^+$. If $\phi^+(w) =
    1$ for some $w \in \mcF_2(S)$, then by Proposition \ref{P:vankampen2} there
    is a $\mcW$-picture $\mcP$ with $\bd(\mcP) = w$.  By Theorem
    \ref{T:constellation} ($b=0$ case), we can choose $\mcP$ so that all
    $\Phi$-cycles in $\mcP$ are facial copies. By Lemma \ref{L:region}, there
    is a $G^+$-picture $\mcP'$ such that $\bd(\mcP') = \bd(\mcP)$, so $w = 1$
    in $G^+$. Since $\phi^+(J^a w) = 1$ for $a \in \Z_2$ and $w \in \mcF_2(S)$ 
    if and only if $a = 0$ and $\phi^+(w)=1$, it follows that $\phi^+$ is
    injective. 

    Now there is a commutative diagram
    \begin{equation*}
        \begin{tikzpicture}
            \node (1) at (0,0) {$G$};
            \node (2) at (3,0) {$\Gamma(\mcW,b)$}
                edge [<-] node [above] {$\phi$} (1);
            \node (3) at (0,-2) {$G^+$}
                edge [<-] node [left] {$q_1$} (1);
            \node (4) at (3,-2) {$\Gamma(\mcW,0)$}
                edge [<-] node [below] {$\phi^+$} (3)
                edge [<-] node [right] {$q_2$} (2);
            \node at ($.5*(2)+.5*(4)+(1,0)$) {,};
        \end{tikzpicture}
    \end{equation*}
    where $q_1$ and $q_2$ are the quotient maps $G \arr G^+ =
    \left(G/(J_G)\right) \times \Z_2$ and $\Gamma(\mcW,b) \arr \Gamma(\mcW,0) =
    \left(\Gamma(\mcW,b) / (J_\Gamma)\right) \times \Z_2$ by $J_G$ and
    $J_\Gamma := J_{\Gamma(\mcW,b)}$ respectively.  Since $J_G$ is central of
    order $\leq$ two, we conclude that $\ker q_1 = \{1, J_G\}$.  Since
    $\phi^+$ is injective, if $\phi(w) = 1$ for $w \in G$, then $q_1(w) = 1$,
    and consequently $w \in \{1,J_G\}$. Thus it remains only to show that
    $\phi(J_G) = 1$ if and only if $J_G=1$ in $G$. By definition, $\phi(J_G) =
    J_{\Gamma}$, and if $J_{\Gamma} = 1$, then by Proposition \ref{P:vankampen2}
    there is a closed $\mcW$-picture $\mcP$ with $\ch(\mcP) \cdot b = 1$. Since
    $\mcP$ is closed, Theorem \ref{T:constellation} again implies that we can
    choose $\mcP$ so that all $\Phi$-cycles in $\mcP$ are facial copies. By
    Lemma \ref{L:region}, there is a closed $G$-picture $\mcP'$ such that
    $\sign(\mcP') = 1$, and we conclude that $J_G = 1$ in $G$. Thus $\phi$ is
    injective.  
\end{proof}

\bibliographystyle{amsalpha}
\bibliography{bcs}

\end{document}